\documentclass[11pt]{article}
\usepackage[left=0.75in,right=0.75in,top=0.5in,bottom=0.75in]{geometry}
\usepackage{scrlayer-scrpage}

\usepackage[svgnames,dvipsnames]{xcolor}
\usepackage{breakcites}
\definecolor{lightblue}{HTML}{044E9E}
\usepackage[breaklinks=true, hyperfootnotes = true,hypertexnames=false,colorlinks,citecolor=lightblue,urlcolor=lightblue, linkcolor=lightblue]{hyperref}

\usepackage{mathtools}
\mathtoolsset{showonlyrefs}
\numberwithin{equation}{section}

\usepackage{enumerate}
\usepackage{enumitem}
\usepackage{multirow}

\usepackage{bm}
\usepackage{amsmath}
\usepackage{graphicx}
\usepackage{booktabs}
\usepackage{psfrag}
\usepackage{epsf}
\usepackage[round,longnamesfirst]{natbib}
\usepackage{url} \usepackage{amsthm}
\usepackage{xpatch}
\usepackage[normalem]{ulem}

\usepackage{longtable}
\usepackage{mathrsfs}
\usepackage{tikz}
\usepackage{color}
\usepackage{amssymb}
\usepackage{latexsym}
\usepackage{xr}
\usepackage{epstopdf}
\usepackage{float}
\usepackage{dsfont}
\usepackage{mathabx}

\def\RR{\mathbb R}

\def\NN{\mathbb N}

\def\EE{\mathbb E}
\def\PP{\mathbb P}

\newcommand{\argmin}{{\rm{argmin}}}

\newcommand{\clh}{\mathcal{H}}

\newcommand{\clb}{\mathcal{B}}

\newcommand{\clk}{\mathcal{K}}

\newcommand{\cln}{\mathcal{N}}
\newcommand{\cls}{\mathcal{S}}

\def\calF{\mathcal F}

\def\calG{\mathcal G}
\def\clH{\mathcal H}

\def\calK{\mathcal K}

\def\clp{\mathcal{P}}

\def\calS{\mathcal S}

\newcommand{\distr}{\operatorname{d}}

\newcommand{\op}{\operatorname{op}}

\def\TV{\text{TV}}

\newtheorem{theorem}{Theorem}[section]
\newtheorem{proposition}{Proposition}[section]
\newtheorem{lemma}{Lemma}[section]
\newtheorem{corollary}{Corollary}[section]
\newtheorem{assumption}{Assumption}[section]

\theoremstyle{definition}
\newtheorem{remark}{Remark}[section]

\begin{document} 

\allowdisplaybreaks

\title{RKHS-Based Inference for Nonlinear Granger Causality \\ via Conditional Centering
}

\author{
Yuhan Tian\textsuperscript{a}
\and
Adam Waterbury\textsuperscript{b}
\and
Marie-Christine D\"uker\textsuperscript{a}\footnote{Corresponding author (\texttt{marie.dueker@tum.de})}
}

\date{}

\maketitle

\begin{center}
\textsuperscript{a}
Department of Mathematics, School of Computation, Information and Technology,\\
Technical University of Munich, Boltzmannstra{\ss}e 3,
85748 Garching bei M\"unchen, Germany

\vspace{0.5em}

\textsuperscript{b}
Department of Mathematics, Denison University,\\
100 West College Street, Granville, OH 43023, USA

\end{center}

\begin{abstract}
\noindent
Granger causality is commonly formulated through linear prediction in vector autoregressive models, limiting its ability to detect nonlinear predictive relationships. We propose a reproducing kernel Hilbert space (RKHS)-based test for nonlinear Granger non-causality in conditional mean for nonlinear autoregressive processes. The key idea is a conditional-centering decomposition of the target regression function into an own-history component and an orthogonal component capturing the additional predictive contribution of the potential source history. Non-causality is
characterized by the vanishing of the latter component. The own-history component is estimated by kernel ridge regression, and the residuals are embedded in a second RKHS using a conditionally centered kernel. This yields an RKHS-valued residual moment whose squared norm forms the test statistic.
We establish a weighted chi-square null limit and consistency against fixed alternatives for the population-centered statistic. An empirically centered version is shown to retain the null limit and enables spectral calibration of critical values and $p$-values without resampling.
Simulation studies demonstrate accurate size control and power against nonlinear alternatives, and real-data applications illustrate the usefulness of the method for detecting nonlinear predictive relationships in time series.
\end{abstract}

\section{Introduction}

Granger causality formalizes whether the past of one time series improves prediction of another
time series \citep{granger1969investigating,g81}. In linear vector autoregressions, this question
reduces to restrictions on finitely many regression coefficients, and the corresponding asymptotic
theory is classical \citep{Hannan1970:Multiple,geweke1982measurement,lutkepohl2013introduction}. Modern applications,
however, often involve feedback, threshold effects, and other nonlinear dependencies for which
linear restrictions do not provide a satisfactory description of predictive content. Recent reviews
and methodological developments emphasize this gap between the general predictive idea of
Granger causality and the linear models most commonly used for inference
\citep{tank2021neural,shojaie2022granger}.

This paper develops a novel RKHS-based statistical test for nonlinear Granger non-causality in conditional mean. To our knowledge, it is the first to combine a population-level orthogonal decomposition of the conditional mean with conditional centering to obtain an RKHS-valued moment condition and an asymptotic test for nonlinear Granger causality. Our approach proceeds from an oracle construction based on a population conditionally centered kernel to a feasible procedure that estimates the conditional-centering map from the observed data while retaining the same asymptotic null distribution. We further identify settings in which conditional centering can be weakened, including inference for nonlinear departures from linear autoregressive dependence and independence under stationarity. The resulting procedure admits spectral calibration without resampling and is illustrated through simulations and two real-data applications.

Several strands of the literature have addressed nonlinear Granger causality beyond the classical linear vector autoregressive framework. Early nonparametric approaches formulate nonlinear Granger non-causality through conditional distributional restrictions. For example, \citet{hiemstra1994testing} proposed a correlation-integral-based test, while \citet{diks2006new} developed a corrected density-based procedure for conditional dependence. These methods target a distributional notion of Granger non-causality and rely on nonparametric estimation of conditional dependence structures.

A related line of work studies nonlinear causality in conditional mean. \citet{nishiyama2011consistent} formulate the null hypothesis as a conditional moment restriction and construct a weighted chi-square test using orthogonal decompositions of function spaces and user-specified basis expansions. The present work is conceptually related in that nonlinear Granger non-causality is also characterized through an orthogonality condition. The novel step in our approach is to use conditional centering to obtain a kernel representation of the relevant orthogonal component, leading to a single RKHS-valued moment condition rather than a finite collection of scalar moment conditions (see \eqref{eq:ZT_def}). Section~\ref{se:comparison_nishiyama} provides a finite-sample comparison
with \citet{nishiyama2011consistent}. Related
specification-testing approaches use Gaussian-process and kernel
conditional-moment constructions \cite{Muandet2020,escanciano2024gaussian}.

Kernel methods for nonlinear Granger causality provide another
important connection. \citet{marinazzo2008kernel} introduced a nonlinear Granger causality measure based on geometric projections in reproducing kernel feature spaces, and \citet{marinazzo2008kernelgranger} extended this idea to dynamical networks. These methods quantify the additional predictive information contributed by a candidate causal series through empirical feature-space constructions, but do not formulate causality through a population-level decomposition of the conditional mean function or derive an asymptotic null distribution for the resulting causality measure.

More broadly, kernel methods have also been used to construct nonparametric tests of dependence through the Hilbert--Schmidt independence criterion (HSIC) \citep{GreBouSmoSchHSIC}, which compares RKHS embeddings of the joint distribution with those induced by the product of the marginals. While classical HSIC theory was developed for independent observations, several extensions have addressed dependent and time-series data \citep{NIPS2008_addfa9b7,pmlr-v32-chwialkowski14,WangLiZhuHSIC}. Similar to HSIC, the proposed statistic is based on RKHS-valued kernel moments and leads to a kernel-based quadratic-form test statistic. However, while HSIC tests dependence, our framework tests nonlinear Granger non-causality through an orthogonal decomposition of the conditional mean function. The resulting asymptotic analysis therefore relies on RKHS-valued moments associated with the autoregressive structure rather than on joint-distribution embeddings.

More recent work has emphasized flexible predictive models for nonlinear network recovery. Neural Granger causality methods identify causal relationships through structured sparsity in neural-network predictors \citep{tank2021neural,sultan2024granger}. Other approaches employ nonlinear state-space transformations \citep{wismuller2021large} or kernel ridge regression predictors \citep{fulmyk2023nonlinear} and assess causality through predictive improvement. These methods are primarily designed for edge recovery and prediction, rather than for inference.

In contrast, the present paper starts from a population-level orthogonal decomposition of the conditional mean and identifies nonlinear Granger causality with the component that cannot be explained by the response history alone. Our novel use of conditional centering provides a kernel representation of this component and leads naturally to an RKHS-valued moment whose squared norm forms the test statistic. Under nonlinear Granger non-causality, the relevant component vanishes and the moment is centered at zero. This yields a nonparametric inferential framework based on an RKHS-valued moment condition, rather than on predictive improvement, variable selection, or finite-dimensional basis expansions.

Our main technical contribution is to establish oracle equivalence for the empirically centered statistic, showing that replacing population conditional centering by its empirical counterpart does not change the asymptotic null distribution. The main difficulty arises from the dependence between the estimated centering map and the innovation-weighted statistic. We handle this dependence through an exact regularized-operator factorization that isolates the interaction between the centering error and the innovation term. These terms are controlled using geometric ergodicity, empirical operator bounds, and martingale arguments. Together with Hilbert-space weak-convergence results, this establishes the asymptotic equivalence between the feasible and oracle statistics.

The asymptotic results also lead to a practical calibration procedure. The
weighted chi-square distribution is approximated through the eigenvalues of
the corresponding empirically centered Gram matrix, yielding
critical values and $p$-values without resampling or repeated model fitting.
This spectral approximation is supported by
the established oracle-equivalence and empirical covariance-operator results.
Simulation studies show that the procedure achieves good empirical size
control and substantial power against nonlinear alternatives across a range
of regularization choices and innovation distributions. The proposed methodology is implemented in the R package \texttt{GrangerRKHS}.

We illustrate the proposed procedure in two real-data applications involving cardiorespiratory and solar--geomagnetic time series. The applications demonstrate its ability to detect nonlinear and potentially time-varying predictive relationships that may be only partially captured by classical linear Granger causality.

The rest of the paper is organized as follows.
Section~\ref{sec:modelandproblemformulation} introduces the nonlinear Granger
causality problem and the orthogonal decomposition used throughout the paper.
Section~\ref{se:RKHS-test} defines the kernel ridge regression residuals, the
oracle statistic based on a population conditionally centered kernel, and the empirical-centering construction.
Section~\ref{sec:theoreticalresults} establishes the asymptotic null
distribution and fixed-alternative behavior of the oracle statistic and proves the asymptotic validity of the feasible procedure under the null.
Section~\ref{se:practical} summarizes implementation and describes the
spectral calibration of the statistic.
Section \ref{subsec:weaken_conditional_centering} discusses alternative ways to handle the first-stage contribution, including inference for departures from a linear autoregressive baseline and a simpler centering construction
when the stationary law factorizes under the Granger non-causality null.
Simulations and data applications can be found in
Sections~\ref{sec:simulation} and~\ref{se:real_data}. Appendix \ref{app:notation_tables} collects the main notation used throughout
the paper in a series of reference tables.  Proofs are collected in
Appendices~\ref{app:appendixasymptotic_distribution_consistency}--\ref{appendix:proofSec6}. A comprehensive sensitivity analysis of the tuning
parameters, together with a discussion of their selection, is included in
Appendix \ref{sec:appendix_bandwidth}.

\textit{Notation:}
Let $\NN\doteq\{1,2,\ldots\}$ and
$\NN_0\doteq\NN\cup\{0\}$. Throughout, all random objects are
defined on a common probability space. Probabilities and expectations
are denoted by $\PP$ and $\EE$, respectively. For a probability
measure $\mu$, we write $\EE_\mu$ for expectation with respect to
$\mu$.
For measurable spaces $\calS_1$ and $\calS_2$, we write
$\clb(\calS_1:\calS_2)$ for the collection of measurable maps from
$\calS_1$ to $\calS_2$. For a Polish space $\calS$, let
$\mathcal B(\calS)$ denote its Borel $\sigma$-field and let
$\clp(\calS)$ denote the set of probability measures on $\calS$.
For $x\in\RR^m$, $m\in\NN$, we write $\|x\|^2 \doteq \sum_{i=1}^m x_i^2$.
For a bounded real-valued function $h$, set $\|h\|_\infty \doteq \sup_x |h(x)|$.
For probability measures $\mu,\nu\in\clp(\calS)$, the total variation distance is $\|\mu-\nu\|_{\TV}
\doteq
\sup_{A\in\mathcal B(\calS)}
|\mu(A)-\nu(A)|$.
For random elements $Z_T$ in a normed space, and for a deterministic
positive sequence $\{\alpha_T\}$, we write
$Z_T=O_{\PP}(\alpha_T)$ if
$\|Z_T\|/\alpha_T=O_{\PP}(1)$, and
$Z_T=o_{\PP}(\alpha_T)$ if
$\|Z_T\|/\alpha_T=o_{\PP}(1)$. In particular,
$Z_T=o_{\PP}(1)$ means that $Z_T\to0$ in probability. We write
$Z_T\xrightarrow{\PP}Z$ for convergence in probability,
$Z_T\xrightarrow{\distr}Z$ for convergence in distribution, and
$Z_T\overset{\distr}{=}Z$ for equality in distribution.
For a Hilbert space $\mathcal H$, let $I_{\mathcal H}$ denote the identity
operator on $\mathcal H$; for $n\in\NN$, let $I_n$ denote the $n\times n$
identity matrix.
For bounded self-adjoint operators $A,B:\mathcal H\to\mathcal H$, write
$A\preceq B$ if
$\langle h,Ah\rangle_{\mathcal H}\leq\langle h,Bh\rangle_{\mathcal H}$
for every $h\in\mathcal H$.
For Hilbert spaces $\clh_1$ and $\clh_2$, and a bounded linear
operator $A:\clh_1\to\clh_2$, define
\begin{equation}
\|A\|_{\op}
\doteq
\sup_{\|h\|_{\clh_1}\le 1}\|Ah\|_{\clh_2}.
\end{equation}
We denote by $\cls_2(\clh_1,\clh_2)$ the Hilbert space of
Hilbert--Schmidt operators from $\clh_1$ to $\clh_2$. Its inner
product and induced norm are
\begin{equation}
\langle A,B\rangle_{\cls_2(\clh_1,\clh_2)}
\doteq
\sum_{i=1}^{\infty}
\langle A e_i,B e_i\rangle_{\clh_2},
\qquad
\|A\|_{\cls_2(\clh_1,\clh_2)}^2
\doteq
\sum_{i=1}^{\infty}
\|A e_i\|_{\clh_2}^2,
\end{equation}
where $\{e_i\}$ is an orthonormal basis of $\clh_1$. When the
domains are clear from context, we suppress them and write
$\cls_2$ and $\|\cdot\|_{\cls_2}$.
For a Hilbert space $\mathcal H$, we say that $U$ is a
centered $\mathcal H$-valued Gaussian random variable with covariance
operator $\Gamma$ if $\Gamma:\mathcal H\to\mathcal H$ is positive,
self-adjoint, and trace-class, and, for every $h,g\in\mathcal H$,
\begin{equation}
\operatorname{Cov}\bigl(
\langle U,h\rangle_{\mathcal H},
\langle U,g\rangle_{\mathcal H}
\bigr)
=
\langle \Gamma h,g\rangle_{\mathcal H}.
\end{equation}
Given two Hilbert spaces $\clh_1$ and $\clh_2$, each pair of elements $h_1 \in \clh_1$ and $h_2 \in \clh_2$ defines a rank-one operator $h_1 \otimes h_2: \clh_2 \to \clh_1$ via $(h_1 \otimes h_2)h \doteq \langle h, h_2 \rangle_{\clh_2} h_1$ for all $h \in \clh_2$.

\section{Model and Problem Formulation}\label{sec:modelandproblemformulation}

Consider a nonlinear vector autoregressive model of the form
\begin{equation}
\begin{aligned}
X_t &= g_X(X_{t-1},Y_{t-1})+\varepsilon_{X,t}, \\
Y_t &= g_Y(X_{t-1},Y_{t-1})+\varepsilon_{Y,t},
\end{aligned}
\label{eq:nonlinearVAR}
\end{equation}
for $t\in\NN$, where $\{(X_t,Y_t)\}_{t\in\NN_0}$ is a bivariate time series taking values in $\RR^2$,  $g_X,g_Y\in\clb(\RR^2:\RR)$ are measurable functions, and $\{\varepsilon_{X,t}\}$ and $\{\varepsilon_{Y,t}\}$ are suitable innovation sequences. For each $t\in\NN_0$, let $\mathcal F_t\doteq\sigma((X_s,Y_s):0\le s\le t)$ denote the natural filtration generated by the process $\{(X_t,Y_t)\}_{t\in\NN_0}$. We assume that the innovations  $\varepsilon_{X,t}$ and $\varepsilon_{Y,t}$ are $\mathcal F_t$-measurable, independent of $\mathcal{F}_{t-1}$, and satisfy $\EE[\varepsilon_{X,t}\mid\mathcal F_{t-1}] = \EE[\varepsilon_{Y,t} \mid\mathcal F_{t-1}]=0$ for every $t\in\NN$. Under the assumptions stated below, the Markov chain $\{(X_t,Y_t)\}_{t\in\NN_0}$ admits a unique stationary distribution $\pi \in \clp(\RR^2)$.  The Hilbert space
\begin{equation}
L^2(\pi)
=
\left\{
g \in \clb(\RR^2:\RR) :
\EE_{\pi}[(g(X_0,Y_0))^2]
=
\int_{\RR^2} (g(x,y))^2   \pi(dx,dy)
< \infty
\right\}
\end{equation}
is equipped with the usual inner product. Throughout, we suppress the distinction between functions in the various RKHSs and equivalence classes in the corresponding \(L^2\) spaces, as it is generally immaterial for this work. Denote by $\pi_X \in \clp(\RR)$ the marginal stationary distribution of $X$. Similarly, $L^2(\pi_X)$ denotes the space of square-integrable functions $f \in \clb(\RR:\RR)$ with respect to $\pi_X$. Nonlinear Granger non-causality from $Y$ to $X$ is defined by the condition that the conditional mean of $X_t$ depends only on $X_{t-1}$, that is,
\begin{equation}\label{eq:H0_original}
H_0: g_X(x,y) = f_X(x)
\quad \text{for } \pi\text{-a.e. }(x,y)\in\RR^2
\end{equation}
for some function $f_X \in L^2(\pi_X)$. However, directly testing whether a bivariate function $g_X(x,y)$ is invariant in its second argument is not straightforward. A natural approach would be to test whether the derivative with respect to $y$ vanishes, but this typically requires additional smoothness assumptions. To obtain a formulation that is both well-defined and convenient for testing, we instead use an orthogonal decomposition based on conditional expectation.

Consider the conditional expectation operator $\Pi_X : L^2(\pi) \to L^2(\pi_X)$ defined by
\begin{equation}
(\Pi_X g)(x) = \EE_{\pi}[g(X_0,Y_0) \mid X_0 = x]
= \int_{\RR} g(x,y)  \pi_{Y\mid X}(dy \mid x),
\end{equation}
where $\pi_{Y\mid X}(\cdot \mid x)$ is a regular conditional probability distribution such that
\begin{equation}
\pi(dx,dy) = \pi_{Y\mid X}(dy \mid x)  \pi_X(dx).
\end{equation}
Its adjoint $\Pi_X^* : L^2(\pi_X) \to L^2(\pi)$ is given by $\Pi_X^* f(x,y) = f(x)$.
Define $R = \Pi_X^* \Pi_X$.
The operator $R$ is bounded, self-adjoint, and idempotent, and hence an orthogonal projection. In particular, its range is closed, which yields the decomposition
\begin{equation}\label{eq:decomposition_L2}
L^2(\pi) = L_X^2(\pi) \oplus L_{X,0}^2(\pi),
\end{equation}
where
\begin{align}
L_X^2(\pi)
\notag&=
\left\{
g \in L^2(\pi) : g(x,y) = f(x) \text{ for }\pi\text{-a.e. }(x,y)\in\RR^2 \text{ for some } f \in L^2(\pi_X)
\right\},\\
L_{X,0}^2(\pi)
&=
\left\{
g \in L^2(\pi) : \EE_{\pi}[g(X_0,Y_0) \mid X_0] = 0 \ \text{a.s.}
\right\}. \label{eq:L2x0pidefinition}
\end{align}

Throughout the paper, we use $f$ to denote functions of a single variable and $g$ to denote functions of two variables, unless otherwise stated. Whenever $g_X\in L^2(\pi)$, which is ensured by Assumption~\ref{assump:bounded_dynamics}, applying the decomposition in \eqref{eq:decomposition_L2} to the regression function $g_X$ in \eqref{eq:nonlinearVAR} yields
\begin{equation}\label{eq:obs_decomposed}
X_t = f_X(X_{t-1}) + g_{X,0}(X_{t-1}, Y_{t-1}) + \varepsilon_{X,t},
\end{equation}
where
$f_X = \Pi_X g_X \in L^2(\pi_X)$, 
$\Pi_X^* f_X = Rg_X \in L_X^2(\pi)$
and $g_{X,0} = (I_{L^2(\pi)} - R)g_X \in L_{X,0}^2(\pi)$. This decomposition is unique in $L^2(\pi)$. It follows from Lemma~\ref{lem:gc_equivalence} that the stationary $L^2(\pi)$ definition of nonlinear Granger non-causality \eqref{eq:H0_original} is equivalent to
\begin{equation}\label{eq:H0_reformulated}
H_0: g_{X,0} = 0
\quad \text{in } L^2(\pi),
\end{equation}
that is, testing whether the part of $g_X$ that depends on $Y$ is zero in the stationary $L^2(\pi)$ sense. This reformulation shows that the problem of testing for nonlinear Granger causality reduces to extracting information about the function $g_{X,0}$ from the observed data. The next section develops an approach based on reproducing kernel Hilbert spaces and kernel ridge regression for this purpose, under the assumption that the regression functions belong to the appropriate RKHS.

\begin{remark}
The decomposition in \eqref{eq:obs_decomposed} is not restricted to the bivariate lag-1 setting. In a general nonlinear VAR model, one may write
\begin{equation}
g_X(U, V) = f_X(U) + g_{X,0}(U, V),
\end{equation}
where $U$ denotes the collection of conditioning variables, including the lagged values of the target variable $X$ and all other variables conditioned upon, and $V$ denotes the collection of lagged variables corresponding to the potential source of Granger causality being tested. In this setting, the null hypothesis is that the collection of variables $V$ does not Granger cause $X$ conditional on $U$, which is equivalent to testing whether $g_{X,0}= 0$. The methodology developed in this paper applies in this more general setting, but the bivariate lag-1 case is adopted here to simplify notation and exposition.
\end{remark}

\section{RKHS-Based Test Construction}
\label{se:RKHS-test}

The previous section shows that nonlinear Granger causality can be formulated as testing whether the component $g_{X,0} \in L^2_{X,0}(\pi)$ in the decomposition \eqref{eq:obs_decomposed} is zero. This formulation suggests that, if $f_X$ were known, then subtracting $f_X(X_{t-1})$ from $X_t$ would isolate the contribution of $g_{X,0}$. In practice, however, $f_X$ is unknown and must be estimated from data. This naturally leads to a residual-based approach, where the signal of interest is extracted from the residuals obtained by  estimating $f_X$. Throughout the remainder of the paper, we assume that the observed sample $\{(X_t,Y_t)\}_{t=0}^T$ is drawn from the stationary process \eqref{eq:nonlinearVAR}, with $(X_0,Y_0) \sim \pi$. This yields an effective sample size of $T$ for the regression problem.

\subsection{Kernel Ridge Regression Estimation of \texorpdfstring{$f_X$}{fX}}

The function $f_X$ is estimated using kernel ridge regression. We briefly introduce the key components here and refer to \cite{duker2025kernel} for a more detailed treatment of KRR in nonlinear VAR models. Let $\clh_1 \subset L^2(\pi_X)$ be a reproducing kernel Hilbert space associated with a symmetric positive definite kernel $K_1 : \RR \times \RR \to \RR$, equipped with inner product $\langle \cdot,\cdot\rangle_{\clh_1}$, and satisfying the reproducing property $f(x) = \langle f, K_1(x,\cdot) \rangle_{\clh_1}$, $f \in \clh_1$.
The estimator $\widehat f \in \clh_1$ is defined as
\begin{equation}
\widehat f
=
\argmin_{f \in \clh_1} \left\{
\frac{1}{T} \sum_{t=1}^T \left(X_t - f(X_{t-1})\right)^2
+
\lambda_{T,f} \|f\|_{\clh_1}^2\right\},
\label{eq:krr_objective}
\end{equation}
where $\lambda_{T,f} > 0$ is a regularization parameter. By the representer theorem \citep{scholkopf2001generalized}, the estimator admits the representation
\begin{equation}
\widehat f(x)
=
\sum_{t=1}^T \alpha_t K_1(X_{t-1}, x),
\qquad
\alpha
=
(\alpha_1, \dots, \alpha_T)^\top
=
\left(K_1 + \lambda_{T,f} T I_T\right)^{-1} X,
\label{eq:fhat_repr}
\end{equation}
where $X = (X_1,\dots,X_T)^\top$ and $K_1 \in \RR^{T \times T}$ is the Gram matrix with entries
\begin{align}
(K_1)_{ij} = K_1(X_{i-1}, X_{j-1}), \qquad  i,j = 1,\dots,T.
\label{eq:K1_gram_matrix_def}
\end{align}

To isolate the contribution of $g_{X,0}$, we remove the estimated own-history component and define
\begin{equation}
\widehat e_t = X_t-\widehat f(X_{t-1}),
\qquad
\widehat e=(\widehat e_1,\ldots,\widehat e_T)^\top.
\label{eq:residuals}
\end{equation}
Under the null, $\widehat f$ is consistent for $f_X$ by \citet{duker2025kernel}; however, finite-sample estimation error may leave residual variation aligned with functions of $X$, motivating the conditional-centering step in Section \ref{se:conditional_centering}.

\subsection{Test Statistic via a Conditionally Centered Kernel} \label{se:conditional_centering}

To further extract information about $g_{X,0}$ from the residuals, we introduce a second symmetric positive definite kernel
$
K_2^c :\RR^2\times\RR^2\to\RR,
$
with associated reproducing kernel Hilbert space $\clh_2 \subseteq L^2(\pi)$.
The kernel is used to construct the second-stage test statistic. It is therefore convenient to work with an RKHS $\clh_2$ that possesses the orthogonality structure separating the component $g_{X,0}$ from functions depending only on $X$, while remaining sufficiently rich to represent $g_{X,0}$. The following assumptions formalize these requirements. Assumption \ref{ass:cond_center} below introduces a conditional-centering property that facilitates the separation of the nuisance component $f_X$ from the signal of interest. 
Assumption~\ref{ass:rkhs_dynamics} then links the decomposition \eqref{eq:obs_decomposed} to the RKHSs used in the procedure.

\begin{assumption}
\label{ass:cond_center}
The kernel $K_2^c$ satisfies
\begin{equation}
\EE_{\pi}\!\left[
K_2^c\big((X_0,Y_0),\cdot \big)
\mid X_0
\right]
=
0
\quad
\text{a.s.}
\label{eq:cond_center}
\end{equation}
\end{assumption}

By Lemma~\ref{lem:cond_center_consequences} in
Appendix \ref{app:auxiliary_results}, Assumption~\ref{ass:cond_center} and
boundedness of $K_2^c$ imply that every $h\in\clh_2$ is conditionally
mean-zero given $X_0$. In particular, Lemma~\ref{lem:cond_center_consequences} implies that $\clh_2 \subseteq L^2_{X,0}(\pi)$,
so every function in $\clh_2$ is orthogonal to functions depending only on
$X$. Consequently, when the residuals \eqref{eq:residuals} are embedded into
$\clh_2$, the influence of the component $f_X$ is further mitigated, allowing
the second-stage analysis to more directly extract information about
$g_{X,0}$. We note, however, that this is not the only possible route to
inference. Section~\ref{subsec:weaken_conditional_centering} discusses
settings in which the same contribution from estimating $f_X$ can be controlled through
a linear-functional central limit theorem or through model-implied independence
under the null hypothesis. These alternatives also support related inferential targets, including nonlinear
deviations from a linear autoregressive baseline and independence of the processes under stationarity.

The conditional-centering condition is related to the extensive literature on
RKHS conditional mean embeddings. \citet{Grnewlder2012ConditionalME} showed
that conditional mean embeddings can be interpreted through Hilbert space-valued
regularized least-squares regression, and studied an empirical method for
estimating such embeddings. \citet{NEURIPS2020_f340f1b1} developed a
measure-theoretic treatment of conditional mean embeddings as conditional
expectations taking values in the relevant RKHS and established consistency of a regularized empirical
estimator. \citet{mollenhauer2020nonparametric} studied the related problem of
nonparametrically approximating conditional expectation operators on RKHS via
finite-rank operator approximations. \citet{li2022optimal} established sharp
learning rates for regularized RKHS conditional mean embedding estimators in
i.i.d.\ settings.

At the same time, \citet{klebanov2020rigorous} clarified the
operator-theoretic foundations of RKHS conditional mean embeddings and
emphasized that the assumptions needed to justify such constructions are
strong and, in practice, often difficult to verify. This motivates, in part,
our treatment of Assumption~\ref{ass:cond_center} as an explicit
population-level conditional-centering condition, rather than as a
straightforward consequence of other model-level assumptions.

Since Assumption~\ref{ass:cond_center} is formulated with respect to the
unknown stationary distribution $\pi$, a kernel satisfying
\eqref{eq:cond_center} is generally not available in closed form. We therefore first develop the
asymptotic theory in
Section~\ref{subsec:asymptotic_distribution_consistency} for an oracle
statistic based on population conditional centering. The construction in
Section~\ref{sec:empirical_conditional_centering} then provides a feasible
procedure by estimating the conditional-centering map from the observed data
and using it to construct an empirically centered kernel.
Section~\ref{sec:full_sample_empirical_centering} establishes that the
resulting statistic retains the oracle asymptotic null distribution under
explicit rate and tuning conditions.

\begin{assumption}
\label{ass:rkhs_dynamics}
Recall the decomposition
\eqref{eq:obs_decomposed},
assume that
\begin{equation}
f_X\in\clh_1,
\qquad
g_{X,0}\in\clh_2.
\end{equation}
\end{assumption}

Assumption~\ref{ass:rkhs_dynamics} says that both components of the regression
function can be represented within suitable RKHSs. The first-stage estimator
$\widehat f$ introduced above is constructed in $\clh_1$, while the second-stage
analysis is carried out in $\clh_2$. Together with the conditional-centering condition in Assumption \ref{ass:cond_center}, the
assumption that $f_X, \widehat{f} \in \clh_1$ and that $g_{X,0}\in\clh_2$ ensures that the first-stage and
second-stage components are suitably separated, so that the second-stage statistic targets 
$g_{X,0}$ rather than the approximation error of the first-stage estimator.

Given Assumption~\ref{ass:rkhs_dynamics}, it is natural to expect that
estimating $g_{X,0}$ requires fitting a second-stage kernel ridge regression in $\clh_2$. However, for the purpose of hypothesis testing, such an estimation
step is unnecessary. As discussed in
Section~\ref{sec:motivationforteststat} below, it is sufficient to
work directly with the RKHS-valued residual moment.  Recalling the  residual estimator from \eqref{eq:residuals}, we define 
\begin{equation}
Z_T(x,y)
=
\frac{1}{\sqrt T}
\sum_{t=1}^T
\widehat e_t
K_2^c\big((X_{t-1},Y_{t-1}),(x,y)\big), \qquad x,y \in \RR.
\label{eq:ZT_def}
\end{equation}
Note that  $Z_T\in\clh_2$ a.s. In this setting, where the conditionally-centered kernel is assumed to be available, the proposed test statistic is given by 
\begin{equation}
Q_T
=
\|Z_T\|_{\clh_2}^2
=
\frac{1}{T} \widehat e^\top K_2^c \widehat e,
\label{eq:QT_def}
\end{equation}
where $K_2^c\in\RR^{T\times T}$ is the Gram matrix with entries
\begin{equation}
(K_2^c)_{ij}
=
K_2^c\big(
(X_{i-1},Y_{i-1}),
(X_{j-1},Y_{j-1})
\big),
\qquad
i,j=1,\ldots,T.
\label{eq:mathbfK2grammatrixdef}
\end{equation}

As shown in Theorem~\ref{thm:ZT_limit} and
Corollary~\ref{cor:QT_limit}, under the null hypothesis in \eqref{eq:H0_reformulated},
$Z_T$ converges weakly to a centered Gaussian element in $\clh_2$, and
$Q_T$ converges to a weighted chi-square distribution, with weights determined by the eigenvalues of the kernel integral operator associated with $K_2^c$. Under fixed
alternatives satisfying $g_{X,0} \in \clh_2$, the contribution of a nonzero $g_{X,0}$ causes $Q_T$ to diverge at rate $T$, yielding a consistent test for nonlinear Granger
causality.

\subsection{Operator Representation for Asymptotic Analysis}\label{sec:operatorrepresen}

For the asymptotic analysis developed in Section \ref{sec:theoreticalresults}, it is convenient to express the above quantities in operator form. To this end, the relevant operators and their empirical counterparts are introduced, and the test statistic is rewritten in this framework.

For any vector $v = (v_1,\dots,v_T)^\top \in \RR^T$, define empirical operators $\clk_1: \RR^T \to \clh_1$ and $\clk_2^c : \RR^T \to \clh_2$ by 
\begin{equation}
(\mathcal K_1 v)(x)
=
\sum_{t=1}^T v_t K_1(X_{t-1}, x),
\qquad
(\mathcal K_2^c v)(x,y)
=
\sum_{t=1}^T v_t K_2^c\big((X_{t-1},Y_{t-1}), (x,y)\big).
\end{equation}
Note that for $f \in \clh_1$ and $g \in \clh_2$, the adjoint operators $\clk_1^* : \clh_1 \to \RR^T$ and $(\clk_2^c)^* : \clh_2 \to \RR^T$ satisfy 
\begin{equation}
\mathcal K_1^* f = \big(f(X_{t-1})\big)_{t=1}^T,
\qquad
(\mathcal K_2^c)^* g = \big(g(X_{t-1},Y_{t-1})\big)_{t=1}^T.
\end{equation}
In particular, the Gram matrices $K_1$ and $K_2^c$ defined in \eqref{eq:K1_gram_matrix_def} and \eqref{eq:mathbfK2grammatrixdef}, respectively,   admit the representations $K_1 = \mathcal K_1^* \mathcal K_1$ and $K_2^c = (\mathcal K_2^c)^* \mathcal K_2^c$. 

We now define the population operators $L_{11} : \clh_1 \to \clh_1$, $L_{12} : \clh_2 \to \clh_1$, $L_{21}^c : \clh_1 \to \clh_2$, and $L_{22}^c : \clh_2 \to \clh_2$, together with their empirical counterparts. For $f \in \clh_1$, $g  \in \clh_2$, and evaluation points $x,y \in \RR$, define 
\begin{align}\label{eq:defineL12L21L11L22}
    (L_{11}f )(x) &= \EE_{\pi_X} [ K_1(X_0, x) f(X_0)] = \int K_1( x', x) f(x') \pi_X(dx')\\
    (L_{12}g )(x) &= \EE_{\pi} [ K_1(X_0, x) g(X_0, Y_0)] = \int K_1( x', x)  g(x', y') \pi( dx', dy')\\
    (L_{22}^c g)(x, y) &=  \EE_{\pi} [ K_2^c( (X_0, Y_0), (x,y)) g(X_0,Y_0)] = \int K_2^c ((x',y'), (x,y)) g(x',y') \pi(dx',dy')\\
    (L_{21}^c f) (x,y) &= \EE_{\pi} [ K_2^c ((X_0, Y_0), (x, y)) f(X_0)] = \int K_2^c((x', y'), (x, y)) f(x') \pi(dx' , dy').
\end{align}
The empirical analogues are
\begin{align}\label{eq:sampleLhatijoperator}
    \widehat L_{11}
    = \frac1T\mathcal K_1\mathcal K_1^*,
    \quad
    \widehat L_{12}
    = \frac1T\mathcal K_1(\mathcal K_2^c)^*,
    \quad
    \widehat L_{21}^c
    = \frac1T\mathcal K_2^c\mathcal K_1^*,
    \quad
    \widehat L_{22}^c
    = \frac1T\mathcal K_2^c(\mathcal K_2^c)^* .
\end{align}
Note that, for example,
\begin{equation} \label{eq:Lhat12def}
(\widehat L_{12} g)(x)
=
\frac{1}{T} \sum_{t=1}^T K_1(x, X_{t-1})  g(X_{t-1}, Y_{t-1}), \qquad g \in \clh_2, \; \; x \in \RR.
\end{equation}

Under suitable conditions, these empirical operators converge to their population counterparts in Hilbert--Schmidt norm at rate $1/\sqrt{T}$, as established in Proposition~\ref{prop:aux_operator_convergence} and Corollary \ref{cor:aux_operator_conv_prob}. The cross-operators $L_{12}$ and $L_{21}^c$ play an important role in the asymptotic analysis. In particular, under Assumption \ref{ass:cond_center}, the operator $L_{21}^c$ vanishes on functions of $X$. 
More precisely, for $f \in \clH_1$,
\begin{equation}\label{eq:L21orthogonality}
(L_{21}^c f)(x,y)
=
\EE_\pi\!\left[K_2^c((X_0,Y_0),(x,y)) f(X_0)\right]
=
\EE_\pi\!\left[\EE\!\left[K_2^c((X_0,Y_0),(x,y)) \mid X_0\right] f(X_0)\right]
= 0.
\end{equation}
This can be viewed as the operator-level analogue of the conditional centering condition in Assumption \ref{ass:cond_center}. Similarly, for each $g \in \clH_2$,
\begin{equation}
(L_{12} g)(x)
=
\EE_\pi\!\left[K_1(X_0,x) g(X_0,Y_0)\right]
=
\EE_\pi\!\left[K_1(X_0,x) \EE[g(X_0,Y_0)\mid X_0]\right]
= 0.
\end{equation}
These conditions are naturally aligned with the decomposition \eqref{eq:obs_decomposed} introduced in Section \ref{sec:modelandproblemformulation}, where $f_X \in L^2(\pi_X)$ and $g_{X,0} \in L_{X,0}^2(\pi)$. 

With the above operators properly defined, Lemma~\ref{lem:commute} implies that the first-stage KRR estimator $\widehat f$ \eqref{eq:fhat_repr} admits both a Gram-matrix representation and an
operator representation
\begin{equation}\label{eq:fhat_operator}
\widehat f
=
\mathcal K_1 \left(K_1 + \lambda_{T,f} T I_T\right)^{-1} X
=
(\widehat L_{11} + \lambda_{T,f} I_{\clh_1})^{-1}
\left(
\frac{1}{T} \mathcal K_1 X
\right),
\end{equation}
where  $X \doteq (X_1,\dots,X_T)^\top$. Observe that $X$ admits the representation
\begin{equation}
X = \mathcal K_1^* f_X + (\mathcal K_2^c)^* g_{X,0} + \boldsymbol{\varepsilon},
\label{eq:X_operator}
\end{equation}
where $\boldsymbol{\varepsilon} \doteq (\varepsilon_{X,1},\dots,\varepsilon_{X,T})^\top$. The residual vector $\widehat e$ defined in \eqref{eq:residuals} then admits the representation
\begin{equation}\label{eq:residuals_operatorform}
\widehat e = X - \mathcal K_1^* \widehat{f}
=
\mathcal K_1^*(f_X - \widehat f)
+
(\mathcal K_2^c)^* g_{X,0}
+
\boldsymbol{\varepsilon}.
\end{equation}
The $\clh_2$-valued element $Z_T$  defined in \eqref{eq:ZT_def} can be written in operator form as
\begin{equation}\label{eq:ZT_operator}
Z_T
=
\frac{1}{\sqrt{T}} \mathcal K_2^c \widehat e
=
\sqrt{T}\left[
\widehat L_{21}^c(f_X - \widehat f)
+
\widehat L_{22}^c g_{X,0}
+
\frac{1}{T}\mathcal K_2^c \boldsymbol{\varepsilon}
\right].
\end{equation}

{
Note that, under the alternative hypothesis, the first-stage estimator
$\widehat f$ is constructed from observations containing both $f_X$ and
$g_{X,0}$ through \eqref{eq:X_operator}. To isolate the estimation error associated with estimating $f_X$,
define
\begin{equation}\label{eq:tildefdef1}
\widetilde f
\doteq
(\widehat L_{11}+\lambda_{T,f}I_{\clh_1})^{-1}
\left(
\widehat L_{11}f_X
+
\frac{1}{T}\mathcal K_1\boldsymbol{\varepsilon}
\right),
\end{equation}
which corresponds to the first-stage KRR estimator that would be obtained if
the component $g_{X,0}$ were absent from the model
in \eqref{eq:obs_decomposed}, namely, if the null hypothesis in \eqref{eq:H0_reformulated} held. Replacing $X$ in  \eqref{eq:fhat_operator} by the decomposition in
\eqref{eq:X_operator}, we obtain
\begin{equation}\label{eq:tildeffhat}
\widehat f
=
\widetilde f
+
(\widehat L_{11}+\lambda_{T,f}I_{\clh_1})^{-1}
\widehat L_{12}g_{X,0}.
\end{equation}
Note in particular that, under the null hypothesis  in \eqref{eq:H0_reformulated}, we have
$\widehat f=\widetilde f$. Substituting the representation \eqref{eq:tildeffhat} into \eqref{eq:ZT_operator} yields
\begin{equation}\label{eq:ZT_operator_alternative}
Z_T
=
\sqrt{T}\left[
\widehat L_{21}^c(f_X-\widetilde f)
+
\widehat L_{22}^c g_{X,0}
-
\widehat L_{21}^c
(\widehat L_{11}+\lambda_{T,f}I_{\clh_1})^{-1}
\widehat L_{12}g_{X,0}
+
\frac{1}{T}\mathcal K_2^c\boldsymbol{\varepsilon}
\right].
\end{equation}
}
Under the null hypothesis, this simplifies to
\begin{equation}\label{eq:ZT_operator_null}
Z_T
=
\sqrt{T}\left[
\widehat L_{21}^c(f_X - \widehat f)
+
\frac{1}{T}\mathcal K_2^c \boldsymbol{\varepsilon}
\right].
\end{equation}
{The first term in \eqref{eq:ZT_operator_null} arises from the estimation error of the first-stage KRR estimator and reflects the remaining influence of the component $f_X$ on the test statistic. Removing the influence of $f_X$ from the test statistic is one of the main motivations for the residualization and conditional-centering steps introduced earlier. Under the assumptions adopted in this paper, residualization reduces the contribution of $f_X$ to the estimation-error term $\sqrt{T} \widehat L_{21}^c(f_X-\widehat f)$, while the conditional-centering condition provides a convenient mechanism for rendering this term asymptotically negligible. Alternative approaches for controlling the same term are discussed in Section~\ref{subsec:weaken_conditional_centering}. Consequently, the leading stochastic fluctuation is determined by
$
\frac{1}{\sqrt T}\mathcal K_2^c\boldsymbol{\varepsilon},
$
which, as shown in Proposition \ref{prop:aux_FCLT}, converges weakly to a centered Gaussian element in $\clh_2$. The decomposition \eqref{eq:ZT_operator_alternative} also provides intuition for
the behavior of the test under fixed alternatives. In particular, the term
$
\widehat L_{22}^c g_{X,0}
$
is the dominant term under fixed alternatives and, as shown in Theorem~\ref{thm:ZT_alt}, drives the divergence of the test statistic.}

\subsection{Motivation for the Test Statistic}\label{sec:motivationforteststat}
{The purpose of this subsection is to motivate the proposed test statistic. The objects introduced here are used only for exposition and do not enter the final test procedure or the main asymptotic results.}
{The decomposition in \eqref{eq:obs_decomposed} suggests an estimator-based approach to
testing the hypothesis in \eqref{eq:H0_reformulated}.  If  $f_X $ were known, then
\begin{equation}
X_t-f_X(X_{t-1})
=
g_{X,0}(X_{t-1},Y_{t-1})+\varepsilon_{X,t},
\end{equation}
so the problem would reduce to testing whether the regression function in this
second-stage equation is identically zero.  In this case, we could define the idealized residuals 
\begin{equation}
\widetilde{e}_t \doteq X_t- f_X(X_{t-1}),
\end{equation}
and then test the null hypothesis by estimating the second-stage regression function
$g_{X,0}$.  In particular, one could define the
idealized second-stage KRR estimator
\begin{align}
\widehat g_\lambda
=
\argmin_{g\in\clh_2}
\left\{
\frac{1}{T}\sum_{t=1}^T
\left(\widetilde{e}_t-g(X_{t-1},Y_{t-1})\right)^2
+
\lambda\|g\|_{\clh_2}^2
\right\}.
\label{eq:glambdaremarkmotivation}
\end{align}
A natural estimator-based test would reject  $H_0 $ for sufficiently large values of
 $\|\widehat g_\lambda\|_{\clh_2}^2 $. The population target of $\widehat{g}_{\lambda}$ is the regularized
version of  $g_{X,0} $ defined by
\begin{align}
g_{X,0,\lambda}
=
\argmin_{g\in\clh_2}
\left\{
\EE_{\pi}\left[
\left(
g_{X,0}(X_0,Y_0)-g(X_0,Y_0)
\right)^2
\right]
+
\lambda\|g\|_{\clh_2}^2
\right\}.
\label{eq:proximategxzerolambdaeqn1}
\end{align}
The corresponding normal equation is 
\begin{align}
(L_{22}^c+\lambda I_{\clh_2})g_{X,0,\lambda}
=
L_{22}^c g_{X,0},
\label{eq:populationproximatenormaleqn}
\end{align}
where $L_{22}^c$ is defined in \eqref{eq:defineL12L21L11L22}. 
By Lemma \ref{lem:injectivitycovarianceoperator}, for every $g\in\mathcal H_2$,
$L_{22}^c g=0$ in $\mathcal H_2$
if and only if 
$g=0$ in $L^2(\pi)$.
Consequently, under Assumption \ref{ass:rkhs_dynamics}, the null hypothesis \eqref{eq:H0_reformulated} is equivalent to $L_{22}^c g_{X,0}=0$.
The empirical analogue of the normal equation in
\eqref{eq:populationproximatenormaleqn} is
\begin{align}
(\widehat L_{22}^c+\lambda I_{\clh_2})\widehat g_\lambda
=
\frac{1}{T}\sum_{t=1}^T
\widetilde e_t K_2^c((X_{t-1},Y_{t-1}),\cdot).
\label{eq:empiricalnormaleqn}
\end{align}
The right-hand side of \eqref{eq:empiricalnormaleqn} is the empirical
$\mathcal{H}_2$-valued moment whose population analogue is
\begin{align}
\EE_{\pi}\left[
\widetilde e_1 K_2^c((X_0,Y_0),\cdot)
\right]
&=
\EE_{\pi}\left[
g_{X,0}(X_0,Y_0)K_2^c((X_0,Y_0),\cdot)
\right]    = L_{22}^c g_{X,0}.
\label{eq:populationmomentL22gx0}
\end{align}
Therefore, if the goal is to test whether $L_{22}^c g_{X,0}$ vanishes, it is natural
to use a suitably rescaled version of the right-hand side of
\eqref{eq:empiricalnormaleqn} directly, rather than first applying the
regularized inverse $(\widehat L_{22}^c+\lambda I_{\clh_2})^{-1}$ and then carrying out the test with the $\clh_2$-norm of the resulting KRR estimator. This motivates the idealized statistics
\begin{equation}
\widetilde Z_T
\doteq
\frac{1}{\sqrt T}\sum_{t=1}^T
\widetilde e_t K_2^c((X_{t-1},Y_{t-1}),\cdot),
\qquad
\widetilde Q_T
\doteq
\|\widetilde Z_T\|_{\clh_2}^2.
\end{equation}
As mentioned above, under the null hypothesis, $L_{22}^c g_{X,0} = 0$, so
$\widetilde Z_T$ is centered at zero. Consequently, large values of $\widetilde Q_T$ provide evidence against the null. The statistic in \eqref{eq:QT_def} is the analogue of this idealized
statistic, namely it replaces the unobservable idealized residuals $\widetilde{e}_t$
by the first-stage residual 
\begin{align}
\widehat e_t=X_t-\widehat f(X_{t-1}).
\label{eq:firststagetrueresiduals}
\end{align}
The kernel centering condition in \eqref{eq:cond_center} simply ensures that the first-stage estimation error in \eqref{eq:firststagetrueresiduals} does not affect the limiting distribution
of $Z_T$.
\begin{remark}
\label{remark:empiricalkernelsufficientfortesting}
A finite-dimensional analogue appears in  linear regression.  Consider
the model 
\begin{equation}
Y=X\beta+\epsilon,
\qquad
X\in\RR^{n\times p},
\qquad
\epsilon\sim \cln_n(0,\sigma^2 I_n),
\end{equation}
where $X$ is fixed and $X^\top X$ is invertible.  The least-squares
estimator is given by $
\widehat\beta=(X^\top X)^{-1}X^\top Y$, 
so $X^\top X\widehat\beta=X^\top Y$.
Since $X^\top X$ is injective, $\widehat\beta=0$ if and only if $
X^\top X\widehat\beta=X^\top Y=0$. Thus, the hypothesis $H_0:\beta=0$ may be tested either through the estimator
$\widehat\beta$ or through $X^\top Y$.
The latter approach is directly analogous to the statistic introduced in \eqref{eq:ZT_def} (or $\widetilde Z_T$). In particular,  $
\EE[X^\top Y]=X^\top X\beta,$
so, under $H_0$, $X^\top Y$ is centered at zero, while under the
alternative it is centered at $X^\top X\beta$.  Therefore, the suitably-rescaled empirical moment $X^\top Y/\sqrt{n}$ can be used to test $H_0$ without
first applying the inverse $(X^\top X)^{-1}$.  In this analogy, $X^\top X$ plays the role
of the regularized empirical covariance operator $ \widehat{L}_{22} + \lambda I_{\clh_2}$, while $X^\top Y/\sqrt{n}$ plays the role of the rescaled empirical
$\clh_2$-valued moment $ \widetilde{Z}_T$, and $\widehat\beta$ plays the role of $\widehat{g}_{\lambda}$.
\end{remark}
}
\subsection{Empirical Conditional Centering}
\label{sec:empirical_conditional_centering}

The statistic in \eqref{eq:QT_def} is defined with a population kernel
$K_2^c$ satisfying Assumption~\ref{ass:cond_center}. Since the stationary
distribution is unknown, this kernel is generally unavailable. We therefore
start from a bounded uncentered kernel
$\widetilde K_2:\RR^2\times\RR^2\to\RR$, with separable RKHS
$\widetilde\clh_2$, and estimate the conditional mean of its feature map.

Define the $\widetilde\clh_2$-valued conditional feature mean
\begin{equation}
m(x)
\doteq
\EE_\pi\!\left[
\widetilde K_2((X_0,Y_0),\cdot)\mid X_0=x
\right],
\qquad \pi_X\text{-a.e. }x\in\RR.
\label{eq:full_m_target}
\end{equation}
We fix a measurable version of $m$ on $\RR$. Boundedness of
$\widetilde K_2$ ensures that this Bochner conditional expectation is well
defined. The associated population-centered scalar kernel is
\begin{equation}
\begin{split}
K_2^c((x,y),(x',y'))
\doteq
\big\langle
\widetilde K_2((x,y),\cdot)-m(x),\widetilde K_2((x',y'),\cdot)-m(x')
\big\rangle_{\widetilde\clh_2}.
\end{split}
\label{eq:full_oracle_centered_kernel}
\end{equation}
By the definition of conditional expectation,
\begin{equation}
\EE_\pi\!\left[
\widetilde K_2((X_0,Y_0),\cdot)-m(X_0)\mid X_0
\right]
=0
\quad\text{a.s. in }\widetilde\clh_2.
\label{eq:full_population_centering_identity}
\end{equation}
Consequently, the kernel in \eqref{eq:full_oracle_centered_kernel} satisfies
Assumption~\ref{ass:cond_center}.

The map defined on kernel sections by
\begin{equation*}
K_2^c((x,y),\cdot)
\longmapsto
\widetilde K_2((x,y),\cdot)-m(x)
\end{equation*}
extends uniquely to a surjective linear isometry from the RKHS
$\clh_2$ induced by $K_2^c$ onto the closure in
$\widetilde\clh_2$ of
\begin{equation*}
\operatorname{span}\{
\widetilde K_2((x,y),\cdot)-m(x):(x,y)\in\RR^2
\}.
\end{equation*}
We use this isometry to identify $\clh_2$ with that closed subspace of
$\widetilde\clh_2$.

Let $\calG$ be the $\widetilde\clh_2$-valued RKHS (see e.g., \cite{micchelli2005learning}) induced by the
operator-valued kernel 
\begin{equation}
\Gamma(x,x')
=
K_1(x,x')I_{\widetilde\clh_2},
\qquad x,x'\in\RR.
\label{eq:full_operator_kernel}
\end{equation}
 For the empirical-centering tuning parameter $\lambda_{T,K}>0$, define
\begin{equation}
\widehat m_T
\in
\argmin_{g\in\calG}
\left\{
\frac1T\sum_{t=1}^T
\left\|
\widetilde K_2((X_{t-1},Y_{t-1}),\cdot)-g(X_{t-1})
\right\|_{\widetilde\clh_2}^2
+
\lambda_{T,K}\|g\|_{\calG}^2
\right\}.
\label{eq:full_mhat_def}
\end{equation}
For the operator representation, recall $L_{11}$ and $\widehat L_{11}$ from
\eqref{eq:defineL12L21L11L22} and \eqref{eq:sampleLhatijoperator}. Define the
uncentered cross-covariance operators by
\begin{align}
\widetilde L_{21}h
&\doteq
\EE_\pi\!\left[
\widetilde K_2((X_0,Y_0),\cdot)h(X_0)
\right],
\quad
\widehat{\widetilde L}_{21,T}h
\doteq
\frac1T\sum_{t=1}^T
\widetilde K_2((X_{t-1},Y_{t-1}),\cdot)h(X_{t-1}),
\quad
h\in\clh_1.
\label{eq:full_raw_L21hat}
\end{align}
For $\lambda>0$, set
\begin{equation}
m_\lambda(x)
\doteq
A_\lambda K_1(x,\cdot)
\quad
\text{ with }
\quad
A_\lambda
\doteq
\widetilde L_{21}(L_{11}+\lambda I_{\clh_1})^{-1}.
\label{eq:full_population_krr_centering}
\end{equation}
The normal equation for \eqref{eq:full_mhat_def} gives
\begin{equation}
\widehat m_T(x)
=
\widehat A_{\lambda_{T,K},T}K_1(x,\cdot),
\qquad x\in\RR,
\quad 
\text{ with }
\quad
\widehat A_{\lambda,T}
\doteq
\widehat{\widetilde L}_{21,T}
(\widehat L_{11}+\lambda I_{\clh_1})^{-1},
\label{eq:full_empirical_krr_operator}
\end{equation}
which  follows  from the representer
theorem for vector-valued RKHS (see e.g., Theorem 4 of \cite{micchelli2005learning}). Let
$\widetilde K_2\in\RR^{T\times T}$ denote the uncentered second-stage Gram matrix
with entries
\begin{equation}
(\widetilde K_2)_{ij}
=
\widetilde K_2((X_{i-1},Y_{i-1}),(X_{j-1},Y_{j-1})),
\qquad i,j=1,\ldots,T,
\label{eq:full_raw_gram}
\end{equation}
and define
\begin{equation}
S_K
\doteq
K_1(K_1+T\lambda_{T,K}I_T)^{-1}.
\label{eq:SK_residualizer_def}
\end{equation}
At the observed inputs, the
RKHS-valued KRR fitted values satisfy
\begin{equation*}
\widehat m_T(X_{i-1})
=
\sum_{j=1}^T (S_K)_{ij}
\widetilde K_2((X_{j-1},Y_{j-1}),\cdot),
\qquad i=1,\ldots,T.
\end{equation*}
Since $S_K$ is symmetric, the Gram matrix of the residual feature sections
\[
\widetilde K_2((X_{i-1},Y_{i-1}),\cdot)
-
\widehat m_T(X_{i-1}),
\qquad i=1,\ldots,T,
\]
is therefore
\begin{equation}
\widehat K_2^c
\doteq
(I_T-S_K)\widetilde K_2(I_T-S_K).
\label{eq:K2hatc_residualized_def}
\end{equation}

Define the population- and empirically centered kernel
sections
\begin{equation}
k_t^c
\doteq
\widetilde K_2((X_{t-1},Y_{t-1}),\cdot)-m(X_{t-1}),
\qquad
\widehat k_{t,T}^c
\doteq
\widetilde K_2((X_{t-1},Y_{t-1}),\cdot)-\widehat m_T(X_{t-1}).
\label{eq:full_centered_sections}
\end{equation}
Using the first-stage residuals from \eqref{eq:residuals}, the
feasible residual embedding and quadratic statistic are
\begin{align}
\widehat Z_T^{\mathrm{res}}
&\doteq
\frac1{\sqrt T}\sum_{t=1}^T\widehat e_t\widehat k_{t,T}^c,
\label{eq:full_residual_embedding}\\
\widehat Q_T^{\mathrm{res}}
&\doteq
\left\|\widehat Z_T^{\mathrm{res}}\right\|_{\widetilde\clh_2}^2
=
\frac1T\widehat e^\top\widehat K_2^c\widehat e.
\label{eq:full_residual_quadratic}
\end{align}
This is the procedure used throughout the simulations and applications. Its
asymptotic null validity is established in
Section~\ref{sec:full_sample_empirical_centering}.
The residual embedding is an RKHS-valued empirical moment that combines the
first-stage residuals with the empirically centered lagged-state features
$\widehat k_{t,T}^c$. The natural test statistic is its squared RKHS norm. At the population level, under the null hypothesis, the ideal
residual is the innovation and its population moment with $k_t^c$ is zero.
Moreover, since $\EE_\pi[k_t^c\mid X_{t-1}]=0$, conditional centering removes
the population contribution of first-stage estimation error that depends only
on $X_{t-1}$. Large values of $\widehat Q_T^{\mathrm{res}}$ therefore provide
evidence against the null.

\section{Theoretical Results}\label{sec:theoreticalresults}

The previous section developed the proposed test construction, including its
oracle formulation, operator representation, and feasible empirical-centering
implementation.
This section develops the corresponding
asymptotic theory in two stages. We first analyze the oracle statistic based
on a fixed population conditionally centered kernel satisfying
\eqref{eq:cond_center}, establishing the limiting behavior of the residual
embedding $Z_T$ and the test statistic $Q_T$, defined respectively in
\eqref{eq:ZT_def} and \eqref{eq:QT_def}. We then analyze the empirical conditional centering procedure introduced in
Section~\ref{sec:empirical_conditional_centering}. The principal difficulty arises from the dependence between the estimated centering map and the residual embedding. An exact regularized-operator factorization controls this
dependence and establishes the same asymptotic null distribution as for the
corresponding oracle statistic.

We begin by stating the assumptions used throughout the section.

\begin{assumption}
\label{assump:bounded_dynamics}
Recall the model in \eqref{eq:nonlinearVAR}. There exists a constant $M_g \in (0,\infty)$ such that
\begin{equation}
\sup_{(x,y)\in \RR^2} |g_X(x,y)| \le M_g, \qquad \sup_{(x,y)\in \RR^2} |g_Y(x,y)| \le M_g.
\end{equation}
\end{assumption}

\begin{assumption}
\label{assump:bounded_kernel}
{The kernels $K_1$ and $K_2^c$ are Borel measurable, bounded, symmetric positive definite kernels, and their RKHSs $\clh_1$ and $\clh_2$ are separable. That is, there exists $\kappa\in(0,\infty)$ such that}
\begin{equation}
\sup_{x\in\RR} K_1(x,x) \le \kappa^2,
\qquad
\sup_{(x,y)\in\RR^2} K_2^c((x,y),(x,y)) \le \kappa^2.
\end{equation}
{Whenever the empirical-centering construction is used, the uncentered kernel $\widetilde K_2$ is also assumed Borel measurable and bounded, its RKHS $\widetilde\clh_2$ is separable, and there exists $\kappa_2\in(0,\infty)$ such that}
\begin{equation}\sup_{(x,y)\in\RR^2}\widetilde K_2((x,y),(x,y))\le \kappa_2^2.
\end{equation}
\end{assumption}

\begin{assumption}
\label{assump:gaussian}
The innovation sequences
$\{\varepsilon_{X,t}\}_{t\in\NN}$ and
$\{\varepsilon_{Y,t}\}_{t\in\NN}$
are independent and identically distributed Gaussian random variables with mean zero and variance
$\sigma^2\in(0,\infty)$. {Moreover, for each $t \in \NN$, $(\varepsilon_{X,t}, \varepsilon_{Y,t})$ is independent of $\mathcal{F}_{t-1}$.}
\end{assumption}

The Gaussian innovation assumption in Assumption \ref{assump:gaussian} is adopted primarily for clarity of presentation and
to simplify the theoretical development. Most of the arguments are expected to extend to
more general independent innovations with mean zero and variance $\sigma^2$ with minimal changes, provided
appropriate regularity conditions are imposed to ensure the geometric ergodicity and the moment conditions required in the asymptotic analysis. Supporting this expectation, the
simulation study in Section~\ref{subsection:heavytail} demonstrates that the proposed
procedure remains robust under a range of heavy-tailed innovation distributions.

\begin{assumption}
\label{assump:first_stage_consistency}
The estimator $\widetilde f$ defined in \eqref{eq:tildefdef1} satisfies
\begin{equation}
\|\widetilde f-f_X\|_{\clh_1}
=
o_{\PP}(1).
\end{equation}
\end{assumption}

Assumption~\ref{assump:first_stage_consistency} ensures that the oracle first-stage estimator
consistently recovers the component $f_X$. Under the null hypothesis,
$\widetilde f=\widehat f$, and the model \eqref{eq:obs_decomposed} reduces to the
nonlinear autoregressive model
\begin{equation}
X_t=f_X(X_{t-1})+\varepsilon_{X,t}.
\end{equation}
Under Assumption~\ref{ass:rkhs_dynamics},
Assumptions~\ref{assump:bounded_dynamics}--\ref{assump:gaussian},
and the conditions on $K_1$ required by Theorem~3.1 of
\citet{duker2025kernel}, that theorem gives high-probability
supremum-norm  and RKHS-norm bounds for the first-stage KRR estimator.
For example, in  the fixed-dimensional setting with a Gaussian first-stage kernel $K_1$,
Theorem~3.1 of \citet{duker2025kernel} implies that, for some
$\zeta_f\in(0,\infty)$, the stochastic component of the supremum-norm
estimation error is of order
\begin{equation}
\frac{(\log T)^{\zeta_f}}{\lambda_{T,f}\sqrt T}.
\end{equation}
Thus, whenever
$\lambda_{T,f}^{-1}$ is at most polylogarithmic and the deterministic regularization bias is negligible at
that scale, the first-stage supremum-norm error can be
$O_{\PP}((\log T)^{\zeta_f}/\sqrt T)$ after enlarging $\zeta_f$ appropriately.  This is the
rate regime used explicitly in the empirical-centering discussion
below. The $\mathcal H_1$-consistency component of
Assumption~\ref{assump:first_stage_consistency} is retained as part of the
assumption unless it is obtained from a separate RKHS-norm consistency result.
Stating Assumption~\ref{assump:first_stage_consistency} also allows the
flexibility of using other estimation procedures. As an example, we discuss a
decomposition into linear and nonlinear functionals which allows using a linear 
regression estimator; see Section \ref{subsec:weaken_conditional_centering}.

The assumptions above separate the probabilistic requirements on the nonlinear vector
autoregressive model from the analytic requirements imposed by the RKHS construction.
The boundedness assumptions keep the kernel terms well behaved.
Assumption~\ref{assump:first_stage_consistency} controls the error from estimating $f_X$,
and conditional centering prevents this component from affecting the limit.

The subsequent asymptotic analysis relies on the geometric ergodicity of
$\{(X_t,Y_t)\}_{t\in\NN_0}$. Specifically, for $h \in \NN_0$, let
\begin{equation}
P^h((x,y),A)
=
\PP\big((X_h,Y_h)\in A \mid (X_0,Y_0)=(x,y)\big),
\qquad
A\in\mathcal B(\mathbb R^2), \quad (x,y) \in \RR^2,
\end{equation}
denote the $h$-step transition probability of the process
$\{(X_t,Y_t)\}_{t\in\NN_0}$. Throughout the proofs, we use the geometric convergence bound
\begin{equation}
\label{eq:geom_ergodicity}
\|P^h((x,y),\cdot)-\pi(\cdot)\|_{\TV}
\le
\rho^hJ(x,y),
\qquad
h\in\NN_0, \quad (x,y) \in \RR^2,
\end{equation}
for some $\rho\in(0,1)$ and some $\pi$-integrable 
$J \in \clb(\mathbb R^2 : [0,\infty))$.
Rather than imposing geometric ergodicity as a separate assumption, we note that the geometric convergence bound in \eqref{eq:geom_ergodicity} follows from Assumptions~\ref{assump:bounded_dynamics} and~\ref{assump:gaussian}; see Section~4.2 of \citet{duker2025kernel} for a detailed discussion.

\subsection{Asymptotic Distribution and Consistency of the Test Statistic}
\label{subsec:asymptotic_distribution_consistency}

This subsection treats the oracle setting in which the population
conditionally centered kernel $K_2^c$ is assumed to be known. Because $K_2^c$
depends on the unknown stationary distribution $\pi$, this setting is
generally infeasible in practice. Section~\ref{sec:full_sample_empirical_centering} develops a feasible empirical-centering procedure and establishes its asymptotic equivalence to the oracle construction.

In Theorem \ref{thm:ZT_limit} and Corollary \ref{cor:QT_limit}, we state the main asymptotic consequences for the residual embedding
$Z_T$ and the quadratic statistic $Q_T$ defined in
\eqref{eq:ZT_def} and \eqref{eq:QT_def}, respectively.
The proofs rely on two auxiliary
RKHS limit results: an $\clh_2$-valued martingale central limit theorem for
the stochastic term $T^{-1/2}\mathcal K_2^c\boldsymbol{\varepsilon}$, and
Hilbert--Schmidt convergence of the empirical RKHS covariance operators.
These auxiliary results are stated in Appendix \ref{appendix:proofs_auxiliary}.

Our first main result gives the weak limit of the residual embedding under the
null hypothesis.

\begin{theorem}
\label{thm:ZT_limit}
Suppose Assumptions \ref{ass:cond_center},  \ref{assump:bounded_dynamics}--\ref{assump:first_stage_consistency}, and recall $Z_T$ from
\eqref{eq:ZT_def}. Then, under the null hypothesis~\eqref{eq:H0_reformulated},
\begin{equation}
Z_T
\xrightarrow{\distr}
\mathcal G_2
\qquad
\text{in } \clh_2,
\end{equation}
where $\mathcal G_2$ is a centered Gaussian element in $\clh_2$ with
covariance operator $\sigma^2 L_{22}^c$.
\end{theorem}

Theorem~\ref{thm:ZT_limit} shows that, under the null hypothesis, the
first-stage residualization and the conditional-centering step remove the
own-history component from the limiting distribution. Under the null, we have 
the decomposition~\eqref{eq:ZT_operator_null}.
The second term converges to $\mathcal G_2$ by the auxiliary $\clh_2$-valued
central limit theorem, while the first term is asymptotically negligible
because conditional centering gives $L_{21}^c h=0$ for functions $h$ of
$X$ and the empirical operator $\widehat L_{21}^c$ converges to $L_{21}^c$
at the required rate. The proof is given in
Appendix \ref{app:appendixasymptotic_distribution_consistency}.

The corresponding scalar limit follows by applying the continuous mapping
theorem to the squared RKHS norm.

\begin{corollary}
\label{cor:QT_limit}
Suppose Assumptions \ref{ass:cond_center},  \ref{assump:bounded_dynamics}--\ref{assump:first_stage_consistency}, and recall the test
statistic $Q_T$ from \eqref{eq:QT_def}. Then, under the null hypothesis \eqref{eq:H0_reformulated},
\begin{equation}
Q_T
\xrightarrow{\distr}
\|\mathcal G_2\|_{\clh_2}^2.
\end{equation}
Moreover, if $\{\mu_j\}_{j=1}^\infty$ are the eigenvalues of $L_{22}^c$,
then
\begin{equation}
\|\mathcal G_2\|_{\clh_2}^2
\overset{\distr}{=}
\sigma^2
\sum_{j=1}^\infty
\mu_j N_j^2,
\end{equation}
where $\{N_j\}_{j=1}^\infty$ are independent standard normal random variables.
\end{corollary}

Corollary~\ref{cor:QT_limit} provides the null distribution used to
calibrate the test. Since the population eigenvalues $\{\mu_j\}_{j=1}^{\infty}$ are
unknown, they are approximated in practice by the eigenvalues of the
empirical centered kernel matrix, as described in Section~\ref{sec:criticalvalue}.
The validity of this spectral approximation is supported by the empirical
operator convergence result stated in
Proposition~\ref{prop:aux_operator_convergence}.

The next theorem records the behavior of the statistic under fixed
alternatives. It shows that the statistic diverges at rate $T$ whenever
the $Y$-dependent component is nonzero in the second-stage RKHS.

\begin{theorem}
\label{thm:ZT_alt}
Suppose Assumptions \ref{ass:cond_center}, \ref{ass:rkhs_dynamics},  
\ref{assump:bounded_dynamics}--\ref{assump:first_stage_consistency},  and $\lambda_{T,f}\sqrt T \to \infty$.
Then,
\begin{equation}
Z_T-\sqrt T\,\widehat L_{22}^c g_{X,0}
\xrightarrow{\distr}
\mathcal G_2
\qquad
\text{in } \clh_2 .
\end{equation}
Consequently, if $L_{22}^c g_{X,0}\neq 0$, then
\begin{equation}
\frac{Q_T}{T}
\xrightarrow{\PP}
\|L_{22}^c g_{X,0}\|_{\clh_2}^2,
\end{equation}
and hence $Q_T$ diverges in probability at the linear rate $T$.
\end{theorem}
By Lemma \ref{lem:injectivitycovarianceoperator}, every fixed alternative
$g_{X,0}\in\mathcal H_2$ with
$\|g_{X,0}\|_{L^2(\pi)}>0$ satisfies
$L_{22}^c g_{X,0}\neq 0$.
Therefore, under the assumptions of Theorem \ref{thm:ZT_alt}, the proposed oracle test is consistent against all
such fixed alternatives.

Theorem~\ref{thm:ZT_alt} also clarifies the role of the second-stage kernel
in finite samples. Although every such fixed alternative has a nonzero
population moment, the strength of the deterministic signal is $\|L_{22}^c g_{X,0}\|_{\mathcal H_2}$.
Thus, alternatives that are represented more strongly by the second-stage
RKHS $\clh_2$ are expected to yield larger finite-sample power. The proof also
requires control of the term
\begin{equation}
\widehat L_{21}^c
(\widehat L_{11}+\lambda_{T,f}I_{\clh_1})^{-1}
\widehat L_{12}g_{X,0},
\end{equation}
which appears in the alternative decomposition in \eqref{eq:ZT_operator_alternative}.
The conditional-centering condition makes this term asymptotically
negligible. Details are provided in
Appendix \ref{app:appendixasymptotic_distribution_consistency}.

\subsection{Asymptotics after Empirical Conditional Centering}
\label{sec:full_sample_empirical_centering}

We establish the null theory for the statistics in
\eqref{eq:full_residual_embedding} and
\eqref{eq:full_residual_quadratic} that employ empirical centering. The main idea is to show that the
residual embedding $\widehat Z_T^{\mathrm{res}}$ has the same
asymptotic limit as its oracle counterpart $Z_T$ in \eqref{eq:ZT_def}.
As shown in \eqref{eq:full_app_three_remainders}, the proof involves three
remainder terms. These terms contain different combinations of the
innovations, the first-stage KRR estimation error, and the empirical
conditional-centering error. Their interaction is the main source of
difficulty, as the resulting terms involve both the temporal dependence of
the process and the dependence induced by estimating the relevant functions
and operators from the data.

The proof separates these sources of error through a sequence of
decompositions and operator factorizations and controls the resulting terms
individually. We first introduce the assumptions and intermediate results
for the first-stage KRR error and the empirical conditional-centering error,
before combining them in the main theorem.

The first source of error comes from estimating the own-history component
$f_X$.

\begin{assumption}
\label{ass:fhattau-f}
For the estimator $\widehat f$ in \eqref{eq:fhat_repr}, suppose that there is
a deterministic sequence $a_T\downarrow0$ such that
\begin{equation}
\|\widehat f-f_X\|_\infty
=O_\PP(a_T),
\qquad
\|\widehat f-f_X\|_{\clh_1}
=O_\PP(a_T).
\label{eq:full_first_stage_rate}
\end{equation}
\end{assumption}

Under the null, $\widehat f$ coincides with the oracle first-stage estimator
in Assumption~\ref{assump:first_stage_consistency}. The quantitative rate in
\eqref{eq:full_first_stage_rate} is needed because the first-stage estimation
error interacts with the empirical conditional-centering error in the
feasible residual embedding. Rates of this form can be obtained under the
null from the KRR theory developed in \cite{duker2025kernel} under suitable
regularity and tuning conditions. Moreover, by the reproducing property,
\[
\|h\|_\infty\leq\kappa\|h\|_{\clh_1},
\]
so an $\clh_1$-norm rate also implies the corresponding supremum-norm rate.

We next quantify the error introduced by estimating the conditional-centering
map. Recall $m_\lambda$ from
\eqref{eq:full_population_krr_centering} and define the squared population
centering bias
\begin{equation}
b_{\lambda,c}^2
\doteq
\int
\|m_\lambda(x)-m(x)\|_{\widetilde\clh_2}^2\,\pi_X(dx).
\label{eq:full_centering_bias}
\end{equation}
To control the regularized operator fluctuations arising from empirical
centering, we use the following regularized feature map. For
$\lambda\in(0,1)$, let $\psi_{1,\lambda}:\RR\to\clh_1$ be given by
\begin{equation}
\psi_{1,\lambda}(x)
\doteq
(L_{11}+\lambda I_{\clh_1})^{-1/2}K_1(x,\cdot).
\label{eq:full_regularized_feature}
\end{equation}
We assume that, for some $\gamma_\psi\in(0,1]$ and
$C_\psi\in(0,\infty)$,
\begin{equation}
\|\psi_{1,\lambda}(x)\|_{\clh_1}
\leq
C_\psi\lambda^{-\gamma_\psi/2}
\quad\text{for every }\lambda\in(0,1)
\quad\text{and }\pi_X\text{-a.e. }x.
\label{eq:full_embedding_bound}
\end{equation}
The inequality in \eqref{eq:full_embedding_bound} always holds with
$\gamma_\psi=1$ and $C_\psi=\kappa$, since
\[
(L_{11}+\lambda I_{\clh_1})^{-1}
\preceq
\lambda^{-1}I_{\clh_1}.
\]

Define the empirical mean squared centering error
\begin{equation}
\widehat R_{c,T}
\doteq
\frac1T\sum_{t=1}^T
\|\widehat m_T(X_{t-1})-m(X_{t-1})\|_{\widetilde\clh_2}^2.
\label{eq:full_centering_risk}
\end{equation}
A useful decomposition for controlling this quantity is given in
\eqref{eq:full_app_centering_error_decomposition}. It separates the empirical
centering error into a stochastic estimation component, arising from the
difference between the empirical and population regularized KRR operators,
and the deterministic regularization bias $m_\lambda-m$. These two
components are controlled separately throughout the proof.

The following assumption specifies how the first-stage rate, the centering
bias, and the regularization parameter must be balanced.

\begin{assumption}
\label{ass:full_empirical_centering_tuning}
Let $a_T$, $b_{\lambda,c}$, and $\gamma_\psi$ be as in
\eqref{eq:full_first_stage_rate}, \eqref{eq:full_centering_bias}, and
\eqref{eq:full_embedding_bound}. Assume that
$\lambda_{T,K}\downarrow0$, $Ta_T^2\to\infty$, and
\begin{equation}
b_{\lambda_{T,K},c}^2
=o\!\left((Ta_T^2)^{-1}\right),
\qquad
a_T^2\lambda_{T,K}^{-(\gamma_\psi+1)}\to0,
\qquad
T\lambda_{T,K}^{2\gamma_\psi+1}\to\infty.
\label{eq:full_tuning_conditions}
\end{equation}
\end{assumption}

The conditions in \eqref{eq:full_tuning_conditions} reflect the different
interactions that must be controlled to recover the oracle limit. The first
condition requires the deterministic centering bias to be sufficiently small
relative to the first-stage KRR error. In particular, a slower first-stage
rate, corresponding to a larger $a_T$, requires a smaller centering bias.
The second condition controls the interaction between the first-stage KRR
error and the stochastic component of empirical centering. The final
condition prevents the centering regularization parameter from decreasing
too quickly and is primarily needed to control the interaction between the
stochastic centering error and the innovations. The requirement
$Ta_T^2\to\infty$ determines the scale relative to which the centering error
must be negligible.

The next proposition makes the stochastic and deterministic contributions to
the empirical centering error explicit.

\begin{proposition}
\label{prop:full_centering_risk_rate}
Under Assumptions~\ref{assump:bounded_dynamics},
\ref{assump:bounded_kernel}, \ref{assump:gaussian}, and
\ref{ass:full_empirical_centering_tuning},
\begin{equation}
\widehat R_{c,T}
=
O_\PP\!\left(
T^{-1}\lambda_{T,K}^{-(\gamma_\psi+1)}
+b_{\lambda_{T,K},c}^2
\right)
=
o_\PP\!\left((Ta_T^2)^{-1}\right).
\label{eq:full_centering_risk_rate}
\end{equation}
\end{proposition}

The first term in \eqref{eq:full_centering_risk_rate} is the stochastic
estimation contribution, while the second is the deterministic
regularization bias. Under Assumption~\ref{ass:full_empirical_centering_tuning},
their combined rate is sufficiently small to control the remainder term
involving the product of the first-stage KRR error and the empirical
centering error as shown in \eqref{eq:full_app_third_remainder}. In particular, the second equality in
\eqref{eq:full_centering_risk_rate} gives precisely the scale needed after
combining the first-stage rate with the centering error through
Cauchy--Schwarz.

It is important that the appearance of $a_T$ in the second equality should
not be interpreted as an intrinsic relationship between the statistical
accuracy of the first-stage estimator and that of the centering estimator.
The first equality describes the centering error itself and does not involve
$a_T$. The dependence on $a_T$ enters through
Assumption~\ref{ass:full_empirical_centering_tuning}, which imposes the
relative rates required for the two estimation errors to interact
negligibly in the feasible residual embedding.

We can now state the asymptotic null distribution of the feasible procedure.

\begin{theorem}
\label{thm:full_oracle_equivalence}
Suppose Assumptions \ref{assump:bounded_dynamics}-\ref{assump:gaussian},
\ref{ass:fhattau-f},
\ref{ass:full_empirical_centering_tuning}, and recall
$\widehat Z_T^{\mathrm{res}}$ from
\eqref{eq:full_residual_embedding}. Under the null hypothesis
\eqref{eq:H0_reformulated},
\begin{equation}
\widehat Z_T^{\mathrm{res}}
\xrightarrow{\distr}
\calG_c
\quad\text{in }\widetilde\clh_2,
\label{eq:full_embedding_limit}
\end{equation}
where $\calG_c$ is centered Gaussian with covariance operator $\sigma^2L_c$,
with
\begin{equation}
L_ch
\doteq
\EE_\pi\!\left[
\langle h,k_1^c\rangle_{\widetilde\clh_2}k_1^c
\right],
\qquad h\in\widetilde\clh_2,
\label{eq:full_Lc_def}
\end{equation}
and $k_1^c$ as in \eqref{eq:full_centered_sections}. Recall
$\widehat Q_T^{\mathrm{res}}$ from
\eqref{eq:full_residual_quadratic}. Then,
\begin{equation}
\widehat Q_T^{\mathrm{res}}
\xrightarrow{\distr}
\|\calG_c\|_{\widetilde\clh_2}^2
\overset{\distr}{=}
\sigma^2\sum_{j=1}^\infty\mu_j^cN_j^2,
\label{eq:full_weighted_chisq}
\end{equation}
where $\{\mu_j^c\}_{j\geq1}$ are the eigenvalues of $L_c$ and
$\{N_j\}_{j\geq1}$ are independent standard normal random variables.
\end{theorem}

The proof of Theorem~\ref{thm:full_oracle_equivalence} in
Appendix~\ref{app:full_sample_empirical_centering} gives the skeleton of the
argument. The three remainder terms in \eqref{eq:full_app_three_remainders} correspond to the innovation-weighted centering error, the first-stage KRR error weighted by the population-centered features, and the interaction between the first-stage and empirical-centering errors. Each term is then controlled
using the bounds developed for the corresponding source of error.

Operator decompositions and factorizations play a central role in these
arguments. Empirical operator fluctuations are written as averages of
rank-one operators and compared with their population counterparts. Kernel
boundedness and the reproducing property provide the required envelope
bounds, while geometric ergodicity controls the temporal covariances in
these empirical averages. This general argument is used repeatedly for
different operators; representative examples are given in
Proposition~\ref{prop:aux_operator_convergence} and
Lemma~\ref{lem:full_fluctuation_rates} and their proofs.

The empirical-centering operator requires an additional step. Its
regularized KRR operator error cannot be controlled directly in the desired
form, and the exact factorization in
\eqref{eq:full_app_exact_factorization} separates it into operator
fluctuations that can be handled by the preceding argument. This
factorization is the key step for the innovation-weighted stochastic
centering error.

The innovation-weighted terms are then controlled using their martingale
structure. The relevant feature terms are measurable with respect to the
past, while the innovations have conditional mean zero. Consequently,
cross-time inner products vanish after conditioning, as illustrated in
\eqref{eq:full_app_W_cross_zero} and
\eqref{eq:expectation_H2_norm_squared_r_lambda_bias-2}. Combining these
martingale bounds with the operator fluctuation bounds, the centering-risk
bound in Proposition~\ref{prop:full_centering_risk_rate}, and the first-stage
rate in Assumption~\ref{ass:fhattau-f} shows that all three remainder terms
are $o_\PP(1)$. The feasible residual embedding therefore has the same null
limit as the oracle innovation embedding.

For the spectral calibration of critical values and $p$-values introduced in
Section~\ref{sec:criticalvalue}, define the feasible empirical covariance
operator
\begin{equation}
\widehat L_{c,T}h
\doteq
\frac1T\sum_{t=1}^T
\langle h,\widehat k_{t,T}^c\rangle_{\widetilde\clh_2}
\widehat k_{t,T}^c.
\label{eq:full_feasible_covariance}
\end{equation}

\begin{proposition}
\label{prop:full_covariance_convergence}
Under the assumptions of Proposition~\ref{prop:full_centering_risk_rate},
\begin{equation}
\|\widehat L_{c,T}-L_c\|_{\cls_2}
=o_\PP(1),
\qquad
\operatorname{tr}(\widehat L_{c,T})
\xrightarrow{\PP}
\operatorname{tr}(L_c).
\label{eq:full_covariance_convergence}
\end{equation}
The nonzero eigenvalues of $\widehat L_{c,T}$ are the
eigenvalues of $\widehat K_2^c/T$.
\end{proposition}

The proposition justifies using the spectrum of the normalized Gram matrix
to approximate the weights in \eqref{eq:full_weighted_chisq}.

\begin{proposition}
\label{prop:full_tuning_choices}
Suppose
\begin{equation}
a_T^2=\frac{(\log T)^{2\zeta_f}}{T},
\qquad
b_{\lambda,c}^2=O(\lambda^{\beta_c})
\label{eq:full_source_regime}
\end{equation}
for some $\zeta_f,\beta_c>0$. Each of the following choices satisfies
Assumption~\ref{ass:full_empirical_centering_tuning}:
\begin{enumerate}[label=(\roman*)]
\item $\lambda_{T,K}=(\log T)^{-r}$ with $r\beta_c>2\zeta_f$;
\item $\lambda_{T,K}=T^{-q}$ with
$0<q<(2\gamma_\psi+1)^{-1}$.
\end{enumerate}
\end{proposition}

The rate for $a_T$ in \eqref{eq:full_source_regime} corresponds to the
first-stage KRR regime discussed after
Assumption~\ref{assump:first_stage_consistency}. The condition
$b_{\lambda,c}^2=O(\lambda^{\beta_c})$ assumes polynomial decay of the
regularization bias in the conditional-centering problem. Under these
representative rate regimes, Proposition~\ref{prop:full_tuning_choices}
shows that the tuning conditions in
Assumption~\ref{ass:full_empirical_centering_tuning} are achievable with
both logarithmically and polynomially decreasing regularization sequences.

The finite-sample choices $\lambda_{T,K}=1/T$ and
$1/(T\log T)$ used in some simulations are not covered by
Proposition~\ref{prop:full_tuning_choices}; they should be interpreted as
empirical implementation rules rather than admissible asymptotic sequences.

\section{Practical Considerations}
\label{se:practical}

This section summarizes the implementation and the spectral
approximation used for critical values and $p$-values under empirical centering.

\subsection{Empirical Centering Construction}
\label{sec:kernelcentering}

We first,
compute the KRR estimator $\widehat f$ in \eqref{eq:fhat_repr} and the
residual vector $\widehat e$ in \eqref{eq:residuals}. Second, form the uncentered
second-stage Gram matrix $\widetilde K_2$ in \eqref{eq:full_raw_gram} and the
matrix $S_K$ in \eqref{eq:SK_residualizer_def}. Third, compute
\begin{equation}
\widehat K_2^c=(I_T-S_K)\widetilde K_2(I_T-S_K)
\end{equation}
and the statistic $\widehat Q_T^{\mathrm{res}}$ from \eqref{eq:full_residual_quadratic}
The two regularization parameters play distinct roles:
$\lambda_{T,f}$ controls the first-stage regression, while
$\lambda_{T,K}$ controls the conditional-centering regression. The matrix
$I_T-S_K$ residualizes the uncentered second-stage feature map with respect to
functions of $X$ represented by $K_1$. Applying it on both sides of
$\widetilde K_2$ gives the Gram matrix of the residualized feature sections.

Theorem~\ref{thm:full_oracle_equivalence} establishes the asymptotic null
validity of this construction. 

\subsection{Approximation of the Null Distribution}
\label{sec:criticalvalue}

Let $\widehat\mu_1,\ldots,\widehat\mu_T$ be the eigenvalues of
$\widehat K_2^c/T$. By Proposition~\ref{prop:full_covariance_convergence},
these are the nonzero empirical covariance-operator eigenvalues associated
with the centered features. The weighted chi-square limit in
Theorem~\ref{thm:full_oracle_equivalence} is approximated by
\begin{equation}
\widehat Q_T^{\mathrm{res}}
\overset{\distr}{\approx}
\widehat\sigma^2
\sum_{j=1}^T\widehat\mu_jN_j^2,
\label{eq:empirical_weighted_chisq_approx}
\end{equation}
where $N_1,\ldots,N_T$ are independent standard normal random variables and
$\widehat\sigma^2$ is estimated from the first-stage residuals.
This approximation can be used directly by simulating the weighted sum.

For computational efficiency, we also use the Satterthwaite approximation
\citep{welch1938significance,satterthwaite1946approximate}. Put
$w_j=\widehat\sigma^2\widehat\mu_j$ and define
\begin{equation}
m_1=\sum_{j=1}^T w_j,
\qquad
m_2=2\sum_{j=1}^T w_j^2,
\qquad
\nu=\frac{2m_1^2}{m_2},
\qquad
c=\frac{m_2}{2m_1}.
\label{eq:satterthwaite_parameters}
\end{equation}
Moment matching gives
$\sum_{j=1}^T w_jN_j^2\overset{\distr}{\approx}c\chi_\nu^2$, and the
approximate $p$-value is
$1-F_{\chi_\nu^2}\!\left(
\widehat Q_T^{\mathrm{res}}/c
\right)$.
Section~\ref{sec:simulation_QTdist} examines this approximation in finite
samples.

\section{When Conditional Centering Can Be Weakened}
\label{subsec:weaken_conditional_centering}

In the preceding sections, the second-stage kernel $K_2^c$ was conditionally
centered in order to eliminate the first-stage estimation error from the null
limit. We now consider a generic bounded second-stage kernel $K_2^\circ$ for
which this orthogonality need not hold. To distinguish this construction from
the conditionally centered one, all second-stage objects introduced in this
section carry a superscript $\circ$.

Let $\clh_2^\circ$ denote the RKHS induced by $K_2^\circ$. For
$v=(v_1,\ldots,v_T)^\top\in\RR^T$, define the empirical evaluation operator
$\mathcal K_2^\circ:\RR^T\to\clh_2^\circ$ by
\begin{equation}
(\mathcal K_2^\circ v)(x,y)
\doteq
\sum_{t=1}^T
v_t K_2^\circ\bigl((X_{t-1},Y_{t-1}),(x,y)\bigr).
\label{eq:circ_eval_operator}
\end{equation}
Define the population operators
\begin{align}
L_{21}^\circ f
&\doteq
\EE_\pi\!\left[
K_2^\circ\bigl((X_0,Y_0),\cdot\bigr)f(X_0)
\right],
\quad f\in\clh_1,
\qquad
L_{22}^\circ g
\doteq
\EE_\pi\!\left[
K_2^\circ\bigl((X_0,Y_0),\cdot\bigr)g(X_0,Y_0)
\right],
\quad g\in\clh_2^\circ,
\end{align}
and their empirical analogues $\widehat L_{21}^\circ
\doteq
\frac1T\mathcal K_2^\circ\mathcal K_1^*$,
$\widehat L_{22}^\circ
\doteq
\frac1T\mathcal K_2^\circ(\mathcal K_2^\circ)^*$.
Finally, define the innovation embedding, residual embedding, and quadratic
statistic by
\begin{align}
G_T^\circ
\doteq
\frac1{\sqrt T}\mathcal K_2^\circ\boldsymbol{\varepsilon},
\qquad
Z_T^\circ
\doteq
\frac1{\sqrt T}\mathcal K_2^\circ\widehat e,
\qquad
Q_T^\circ
\doteq
\|Z_T^\circ\|_{\clh_2^\circ}^2.
\label{eq:circ_GT_def}
\end{align}

Under the null hypothesis~\eqref{eq:H0_reformulated}, the same residual
decomposition as in~\eqref{eq:ZT_operator_null}, now with $K_2^\circ$ in
place of $K_2^c$, gives
\begin{equation}
Z_T^\circ
=
G_T^\circ
+
\sqrt T\,\widehat L_{21}^\circ(f_X-\widehat f).
\label{eq:circ_null_decomposition}
\end{equation}
Adding and subtracting the population operator yields
\begin{equation}
\begin{split}
Z_T^\circ
=
G_T^\circ
&+
\sqrt T\,L_{21}^\circ(f_X-\widehat f)
+
\sqrt T\,(\widehat L_{21}^\circ-L_{21}^\circ)
(f_X-\widehat f).
\end{split}
\label{eq:circ_population_decomposition}
\end{equation}
This decomposition makes the role of conditional centering explicit. For the
conditionally centered kernel used in the preceding sections,
Assumption~\ref{ass:cond_center} implies $L_{21}^c f=0$ for every
$f\in\clh_1$; see~\eqref{eq:L21orthogonality}. For a generic kernel $K_2^\circ$,
there is no reason for $L_{21}^\circ$ to vanish. What matters for the null
limit is instead the transformed first-stage error
\begin{equation}
\sqrt T\,L_{21}^\circ(f_X-\widehat f).
\label{eq:circ_transformed_first_stage_error}
\end{equation}
Indeed, the proof of Corollary~\ref{cor:aux_operator_conv_prob} applies in the
same way to $K_2^\circ$ and gives
\begin{equation}
\sqrt T\,
\|\widehat L_{21}^\circ-L_{21}^\circ\|_{\op}
=
O_\PP(1)
\label{eq:circ_L21_operator_rate}
\end{equation}
under boundedness of $K_1$ and $K_2^\circ$ and the geometric-ergodicity
condition~\eqref{eq:geom_ergodicity}. Moreover, under the null,
$\widehat f=\widetilde f$ by~\eqref{eq:tildeffhat}, so
Assumption~\ref{assump:first_stage_consistency} gives
$\|\widehat f-f_X\|_{\clh_1}=o_\PP(1)$. Hence the final term in
\eqref{eq:circ_population_decomposition} is $o_\PP(1)$. Conditional
centering is therefore one sufficient way to remove the first-stage
contribution, but it is stronger than necessary: the null limit can also be
obtained by characterizing the weak limit of
\eqref{eq:circ_transformed_first_stage_error} and incorporating that limit
into the asymptotic distribution.

The innovation term itself is unaffected by this change in viewpoint. Under
the same boundedness, innovation, and geometric-ergodicity conditions, the
argument of Proposition~\ref{prop:aux_FCLT}, with $K_2^c$ and $\clh_2$
replaced by $K_2^\circ$ and $\clh_2^\circ$, gives
\begin{equation}
G_T^\circ
\xrightarrow{\distr}
\mathcal G_2^\circ
\qquad
\text{in }\clh_2^\circ,
\label{eq:circ_GT_limit}
\end{equation}
where $\mathcal G_2^\circ$ is centered Gaussian with covariance operator
$\sigma^2L_{22}^\circ$. The remaining issue is therefore the joint behavior
of this innovation term and the transformed first-stage error.

\begin{theorem}
\label{thm:ZT_limit_weak_centering}
Suppose Assumptions~\ref{assump:bounded_kernel},
\ref{assump:gaussian}, and
\ref{assump:first_stage_consistency} hold, with the requirements on
$K_2^c$ and $\clh_2$ in Assumption~\ref{assump:bounded_kernel} imposed instead
on $K_2^\circ$ and $\clh_2^\circ$. Suppose also that the process
$\{(X_t,Y_t)\}_{t\in\NN_0}$ is geometrically ergodic in the sense of
\eqref{eq:geom_ergodicity}. The kernel $K_2^\circ$ need not satisfy
Assumption~\ref{ass:cond_center}. Under the null hypothesis
\eqref{eq:H0_reformulated}, assume that, in
$\clh_2^\circ\times\clh_2^\circ$,
\begin{equation}
\left(
G_T^\circ,
\sqrt T\,L_{21}^\circ(f_X-\widehat f)
\right)
\xrightarrow{\distr}
\left(
\mathcal G_2^\circ,
\mathcal R^\circ
\right),
\label{eq:first_stage_functional_clt}
\end{equation}
for some $\clh_2^\circ$-valued random element $\mathcal R^\circ$. Then,
\begin{equation}
Z_T^\circ
\xrightarrow{\distr}
\mathcal G_2^\circ+\mathcal R^\circ
\quad
\text{in }\clh_2^\circ,
\quad
\text{ and }
\quad
Q_T^\circ
\xrightarrow{\distr}
\|\mathcal G_2^\circ+\mathcal R^\circ\|_{\clh_2^\circ}^2.
\label{eq:circ_Q_limit}
\end{equation}
\end{theorem}

Theorem~\ref{thm:ZT_limit_weak_centering} shows precisely what is lost when
conditional centering is removed. For $K_2^c$, the population transformed
first-stage error is identically zero and the theorem reduces to the null
limit in Theorem~\ref{thm:ZT_limit}. For a generic $K_2^\circ$, the additional
term $\mathcal R^\circ$ need not vanish and must be included in the null
distribution.

Condition~\eqref{eq:first_stage_functional_clt} is a requirement on a
particular linear transformation of the first-stage estimation error, not a
full $\sqrt T$-rate weak limit for $\widehat f-f_X$ in $\clh_1$. This
distinction is important in light of \citet{mas2007weak}. In the Hilbertian
autoregressive model, a growing-rank pseudo-inverse plays a role analogous to
a vanishing regularization parameter: Theorem~3.2 of \citet{mas2007weak}
gives a CLT for prediction error, whereas Theorem~3.1 shows that the full
operator estimator cannot have a non-degenerate weak limit in operator norm.
Thus, when $\lambda_{T,f}\downarrow0$, a full Hilbert-space CLT for
$\widehat f-f_X$ may be substantially stronger than necessary. Linear or
smoother functionals can have tractable asymptotics even when the full
estimator does not. Results on linear-functional asymptotics for KRR, such as
\citet{tuo2024asymptotic}, provide a natural template for studying
$L_{21}^\circ(f_X-\widehat f)$, although the joint
$\clh_2^\circ$-valued limit in~\eqref{eq:first_stage_functional_clt} must still
be established for the dependent-data setting considered here.

The same decomposition can also be used with first-stage estimators other than
KRR, provided the corresponding consistency and joint-limit conditions hold.
A particularly transparent example is inference against departures from a
linear autoregressive baseline. Consider
\begin{equation}
X_t
=
A X_{t-1}
+
r_0(X_{t-1},Y_{t-1})
+
\varepsilon_{X,t},
\qquad
H_0^{\mathrm{lin}}:r_0=0.
\label{eq:linear_baseline_section6}
\end{equation}
Let $\widehat A$ be a first-stage estimator of $A$ and set
$\widehat f_{\mathrm{lin}}(x)=\widehat A x$. Under the null,
$f_X(x)=Ax$, and hence
\begin{equation}
\sqrt T\,
L_{21}^\circ(f_X-\widehat f_{\mathrm{lin}})
=
\sqrt T\,(A-\widehat A)L_{21}^\circ\iota,
\qquad
\iota(x)\doteq x,
\label{eq:linear_baseline_functional}
\end{equation}
whenever the displayed quantities are well defined. The additional first-stage
term is therefore a fixed linear transformation of the finite-dimensional
estimation error $\sqrt T(\widehat A-A)$. A joint CLT for this estimation error
and $G_T^\circ$ can be incorporated exactly as in
Theorem~\ref{thm:ZT_limit_weak_centering}. The resulting residual-kernel
statistic targets departures $r_0\neq0$ from the linear autoregressive
baseline, including nonlinear own-history or cross-variable effects that are
represented by $K_2^\circ$. Conditional centering is not required for this
formulation, but the effect of estimating $A$ generally remains in the null
distribution unless an additional orthogonality condition makes
$L_{21}^\circ\iota=0$.

A different simplification arises when factorization of the
stationary law allows conditional centering to be obtained
by marginal centering. Suppose
that, under the null hypothesis, the stationary law factorizes as
$\pi=\pi_X\otimes\pi_Y$. This occurs, for example, in the autonomous system
\begin{equation}
X_t=f_X(X_{t-1})+\varepsilon_{X,t},
\qquad
Y_t=f_Y(Y_{t-1})+\varepsilon_{Y,t},
\label{eq:autonomous_product_example}
\end{equation}
when the innovations are independent and the stationary initial distribution
factorizes as a product. Let the second-stage kernel have the product form
\begin{equation}
K_2^\circ((x,y),(x',y'))
=
\bar K_X(x,x')\bar K_Y(y,y'),
\label{eq:circ_product_kernel}
\end{equation}
where $\bar K_X$ and $\bar K_Y$ are bounded kernels and $\bar K_Y$ is
marginally centered,
\begin{equation}
\int \bar K_Y(y,y')\,\pi_Y(dy)=0
\qquad
\text{for every }y'\in\RR.
\label{eq:marginal_y_centering}
\end{equation}
Then $L_{21}^\circ h=0$ for every admissible function $h$ of $X$. Thus, under
this product-form stationary null, the full conditional-centering condition
\eqref{eq:cond_center} can be replaced by the simpler marginal-centering
condition~\eqref{eq:marginal_y_centering}. In finite samples, the analogous
construction is to center the Gram matrix associated with the $Y$-factor
before forming the product-kernel Gram matrix.

\begin{proposition}
\label{prop:autonomous_null}
Assume that the null stationary distribution satisfies
$\pi=\pi_X\otimes\pi_Y$, and suppose that $K_2^\circ$ has the product form
\eqref{eq:circ_product_kernel}, with $\bar K_Y$ satisfying
\eqref{eq:marginal_y_centering}. Then
\begin{equation}
L_{21}^\circ h=0
\qquad
\text{for every }h\in L^2(\pi_X)
\end{equation}
for which the expectation defining $L_{21}^\circ h$ is well defined. If, in
addition, Assumptions~\ref{assump:bounded_kernel},
\ref{assump:gaussian}, and \ref{assump:first_stage_consistency} hold, with
$K_2^c$ and $\clh_2$ in Assumption~\ref{assump:bounded_kernel} replaced by
$K_2^\circ$ and $\clh_2^\circ$, and the process is geometrically ergodic in
the sense of~\eqref{eq:geom_ergodicity}, then under the null hypothesis
\begin{equation}
Z_T^\circ
\xrightarrow{\distr}
\mathcal G_2^\circ
\qquad
\text{in }\clh_2^\circ,
\qquad
Q_T^\circ
\xrightarrow{\distr}
\|\mathcal G_2^\circ\|_{\clh_2^\circ}^2.
\label{eq:product_null_circ_limit}
\end{equation}
\end{proposition}

Proposition~\ref{prop:autonomous_null} gives a setting in which
conditional centering can be achieved by marginal centering.
Indeed, when $\pi=\pi_X\otimes\pi_Y$ and $\bar K_Y$ is
marginally centered, the product kernel satisfies, for every
$(x',y')$,
\[
\mathbb E_\pi\!\left[
K_2^\circ((X_0,Y_0),(x',y')) \mid X_0
\right]
=
\bar K_X(X_0,x')
\int \bar K_Y(y,y')\,\pi_Y(dy)
=0
\quad\text{a.s.}
\]
Thus $K_2^\circ$ is conditionally centered under the factorized
stationary law, and
$L_{21}^\circ(f_X-\widehat f)=0$ identically.
Under the Granger non-causality null $g_{X,0}=0$ in $L^2(\pi)$
and the remaining assumptions of the proposition, the empirical
first-stage contribution is asymptotically negligible, so no
additional joint-limit assumption for the first-stage estimation
error is needed.
Here, factorization of the stationary law is an additional
assumption used to simplify centering under the Granger
non-causality null.

The discussion above concerns the null distribution of the residual moment.
Other inferential formulations are also possible. One option is to keep a
fixed amount of second-stage smoothing. Suppose $g_{X,0}\in\clh_2^\circ$ and,
for $\lambda>0$, define
\begin{equation}
g_\lambda^\circ
\doteq
(L_{22}^\circ+\lambda I_{\clh_2^\circ})^{-1}
L_{22}^\circ g_{X,0}.
\label{eq:circ_smoothed_target}
\end{equation}
When $\lambda\to\lambda_0>0$, the inverse remains bounded, so inference for
$g_\lambda^\circ$ avoids the instability associated with sending the
regularization parameter to zero. Under $g_{X,0}=0$, the smoothing bias is
zero for every $\lambda$. Under alternatives, however, the target is the
smoothed component $g_\lambda^\circ$ rather than $g_{X,0}$ itself, and
recovering the latter as $\lambda\downarrow0$ requires separate approximation
or source conditions. This modification changes the inferential target; it
does not by itself remove the first-stage contribution in
\eqref{eq:circ_population_decomposition}.

A second option is to work directly with the scalar quadratic form. Under the
null, the innovation-only component satisfies
\begin{equation}
\|G_T^\circ\|_{\clh_2^\circ}^2
=
\frac1T
\sum_{s=1}^T\sum_{t=1}^T
\varepsilon_{X,s}\varepsilon_{X,t}
K_2^\circ\bigl(
(X_{s-1},Y_{s-1}),(X_{t-1},Y_{t-1})
\bigr).
\label{eq:circ_direct_quadratic}
\end{equation}
Thus, under alternative regularity regimes, one may seek a direct scalar limit
theorem for this quadratic form rather than first proving an
$\clh_2^\circ$-valued weak limit. Such an argument changes the route to the
null distribution, not the role of the first-stage error: the contribution in
\eqref{eq:circ_transformed_first_stage_error} must still either be shown to
vanish or be incorporated into the limiting quadratic statistic.

\section{Simulation Studies}\label{sec:simulation}

In this section, the finite-sample performance of the proposed nonlinear
Granger causality test is investigated through several simulation studies. We
examine the finite-sample distribution of the test statistic and the
Satterthwaite approximation in Section~\ref{sec:simulation_QTdist}, its
sensitivity to the regularization parameters $\lambda_{T,f}$ and
$\lambda_{T,K}$ in Section~\ref{sec:sensitivity}, and its robustness under
heavy-tailed innovations in Section~\ref{subsection:heavytail}. A
finite-sample comparison with the nonparametric causality test of
\citet{nishiyama2011consistent} is provided in
Section~\ref{se:comparison_nishiyama}. These simulation studies, as well as
the real-data applications, use the empirical-centering procedure
introduced in Section~\ref{sec:kernelcentering}.

Throughout all simulations, the Gaussian kernel
\begin{equation}
K(x,y)
=
\exp\left(
-\frac{\|x-y\|^2}{\tau^2}
\right),
\end{equation}
is used for the first-stage kernel $K_1$ and as the uncentered second-stage kernel $\widetilde K_2$; the second-stage Gram matrix is then centered using \eqref{eq:K2hatc_residualized_def}. The default bandwidth is $\tau=\sqrt{2}$. As discussed after Theorem~\ref{thm:ZT_alt}, the proposed procedure relies on the $Y$-dependent component $g_{X,0}$ being well approximated within $\clh_2$. The Gaussian kernel induces a very rich reproducing kernel Hilbert space and is therefore well suited for detecting general nonlinear Granger-causal relationships when no prior structural information is available.

\begin{remark}
The choice of kernel bandwidth becomes increasingly important as the dimension of the lagged state vector grows. In higher-dimensional settings, pairwise Euclidean distances tend to increase with the number of variables and lags, which may substantially alter the effective behavior of the Gaussian kernel. Consequently, larger bandwidths may be preferable when many variables or large lag orders are considered. In practice, when no prior bandwidth information is available, a common heuristic is to choose the bandwidth based on the median pairwise distance between observations. In the low-dimensional settings considered here, especially when the component
series have variances close to 1 or are standardized separately, the default
choice $\tau=\sqrt{2}$ performs well empirically.
\end{remark}

Unless otherwise specified, the regularization parameters are chosen as
\begin{equation}
\lambda_{T,f}
=
\frac{\log(T)}{\sqrt{T}},
\qquad
\lambda_{T,K}
=
\frac{1}{T\log(T)}.
\end{equation}
The choice of $\lambda_{T,f}$ is motivated by the first-stage consistency
condition in Assumption~\ref{assump:first_stage_consistency} and the
asymptotic theory of Theorem~\ref{thm:ZT_alt}. In particular,
$\lambda_{T,f}\to0$ and
$\lambda_{T,f}\sqrt{T}=\log(T)\to\infty$, so the explicit regularization
condition in Theorem~\ref{thm:ZT_alt} is satisfied. First-stage consistency
is imposed separately through Assumption~\ref{assump:first_stage_consistency};
Theorem~3.1 of \citet{duker2025kernel} provides sufficient conditions under
which this assumption can be verified.

The choice of $\lambda_{T,K}$ is a practical tuning rule for the empirical-centering procedure. Since the matrix inverse in
\eqref{eq:SK_residualizer_def} depends on $T\lambda_{T,K}$, the choice $\lambda_{T,K}=1/(T\log(T))$ gives
$T\lambda_{T,K}=1/\log(T)\to0$. Thus, the centering fit becomes progressively
less regularized as $T$ increases, while retaining a small amount of
finite-sample stabilization. This choice is not covered by the asymptotic
tuning conditions in Proposition~\ref{prop:full_tuning_choices}; its
finite-sample performance is examined in Section~\ref{sec:sensitivity}.

Throughout the simulations, the bivariate lag-1 nonlinear VAR model~\eqref{eq:nonlinearVAR} is considered. The model is specified in Table \ref{tab:gx}. 
\begin{table}[h]
\centering
\begin{tabular}{ll}
\toprule
Model & $g_X(x,y)$ \\
\midrule
Null & $\sin(x)$ \\
Alternative 1 & $\sin(x)+\cos(y)$ \\
Alternative 2 & $\sin(x)\cos(y)$ \\
\bottomrule
\end{tabular}
\caption{Model specifications. In all cases,
$g_Y(x,y)=\cos(x)+\tanh(y)$.}
\label{tab:gx}
\end{table}
Unless otherwise mentioned, the innovations are generated from independent standard normal distributions. {The noise variance $\sigma^2$ is estimated from the residuals of a first-stage kernel ridge regression fit based only on $X_{t-1}$, using the fixed regularization parameter $\lambda=1/T$. This choice is adopted throughout the simulations to isolate the effect of the first-stage tuning parameter $\lambda_{T,f}$ on the proposed test statistic and to avoid confounding it with the quality of the variance estimation. Since variance estimation is not the primary focus of the present work, other consistent estimators of $\sigma^2$ could also be used.}

Unless otherwise specified, rejection rates are reported as the proportion of
1000 Monte Carlo repetitions in which the null hypothesis of no Granger
causality from $Y$ to $X$ is rejected. Under the null model, the rejection rate should be close to $0.05$, while under the alternative models it should be close to $1$.

\subsection{Finite-Sample Distribution of \texorpdfstring{$Q_T$}{QT} and Satterthwaite Approximation}\label{sec:simulation_QTdist}

In this subsection, the finite-sample distribution of the test statistic $Q_T$, together with the corresponding Satterthwaite approximation introduced in Section~\ref{sec:criticalvalue}, is investigated. The null model and the two nonlinear alternatives introduced previously are considered, and the empirical distribution of the test statistic is examined under different sample sizes.

For each setting, we plot histograms of $Q_T$ based on 1000 Monte Carlo repetitions. The black curve shown in each panel corresponds to the Satterthwaite approximation
$
c\chi_\nu^2,
$
where $c$ and $\nu$ are obtained by averaging the Satterthwaite parameters over all repetitions. Hence, the curve represents the average finite-sample approximation of the distribution of the test statistic. The averaged approximation is used only for visualization, while the rejection decisions are computed using the Satterthwaite approximation obtained separately within each repetition.

\begin{figure}[!ht]
    \centering
    \IfFileExists{Figure_1.png}{
    \IfFileExists{Figure_1.png}
    {\includegraphics[width=0.8\textwidth]{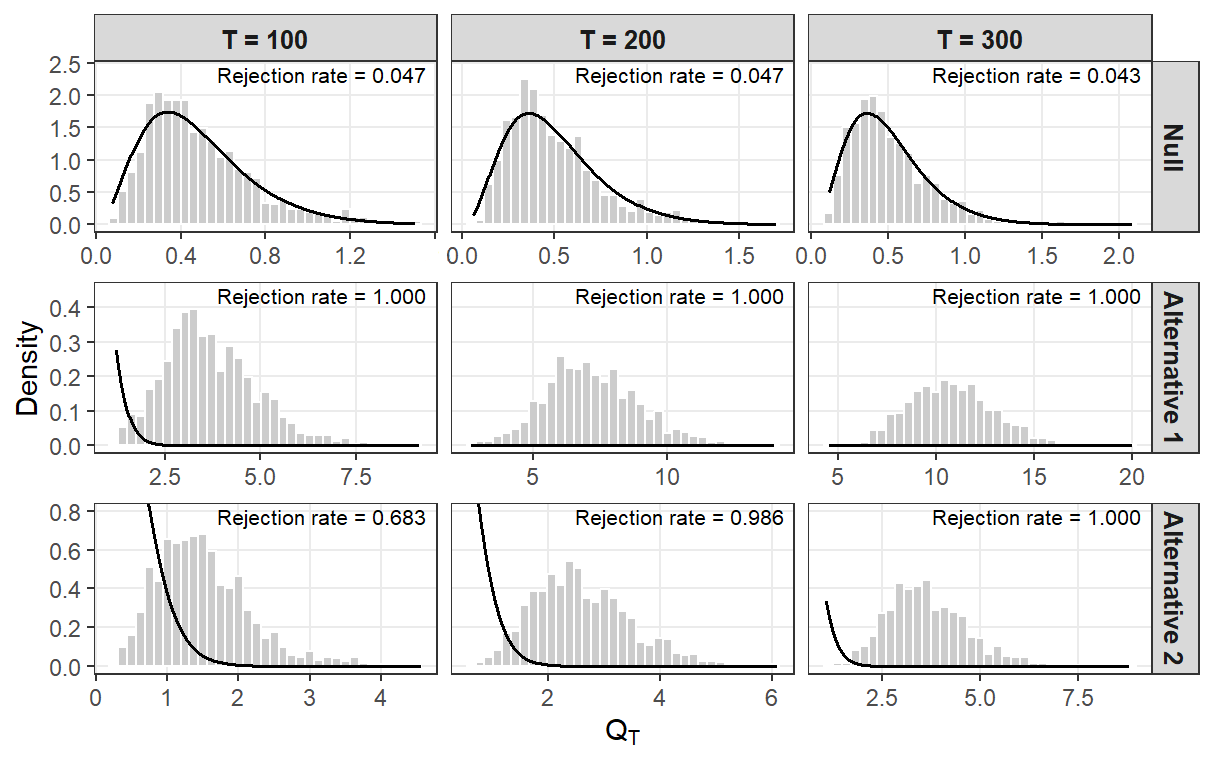}}
    {\fbox{\parbox[c][0.24\textheight][c]{0.78\textwidth}{\centering
    Simulation figure unavailable in the supplied source files.}}}
    }{
    \fbox{\parbox[c][0.28\textheight][c]{0.8\textwidth}{\centering
    Simulation figure omitted because \texttt{arXiv/plots/dist\_QT.png} was not found.}}
    }
    \caption{Empirical distributions of the proposed test statistic $Q_T$ under the null model and two nonlinear alternatives for different sample sizes. The black curves correspond to the averaged Satterthwaite approximations $c\chi_\nu^2$. Rejection rates are computed at significance level $\alpha=0.05$.}
\label{fig:simulation_distribution}
\end{figure}

Figure~\ref{fig:simulation_distribution} suggests that under the null model, the empirical distribution of the proposed statistic aligns well with the Satterthwaite approximation even for relatively small sample sizes. Moreover, the empirical rejection rates remain close to the nominal significance level $0.05$, indicating good finite-sample calibration of the procedure. Under both alternative models, the test statistic shifts increasingly to the
right of the corresponding Satterthwaite reference distribution as the sample
size grows, which is consistent with the divergence behavior established in
the theoretical analysis. At the same time, the Satterthwaite reference
distribution remains concentrated at substantially smaller values than the
empirical statistic. It therefore continues to provide a useful reference for
critical-value selection and $p$-value computation in the presence of
nonlinear Granger-causal effects.

\subsection{Sensitivity to the Regularization Parameters \texorpdfstring{$\lambda_{T,f}$}{lambda Tf} and \texorpdfstring{$\lambda_{T,K}$}{lambda TK}}\label{sec:sensitivity}

In this subsection, the sensitivity of our test to the regularization parameters $\lambda_{T,f}$ and $\lambda_{T,K}$ is investigated. The sample size is fixed at $T=100$, and only Null and Alternative 2 of Table \ref{tab:gx} are considered. The parameter grids are chosen as
\begin{equation}
\lambda_{T,f}
\in
\left\{
\frac{\log(T)}{\sqrt{T}},
\frac{1}{\sqrt{T}},
\frac{1}{T},
\frac{1}{T^2}
\right\},
\qquad
\lambda_{T,K}
\in
\left\{
\frac{1}{\sqrt{T}},
\frac{1}{T},
\frac{1}{T\log(T)},
\frac{1}{T^2}
\right\}.
\end{equation}

\begin{table}[!ht]
\centering
\begin{minipage}{0.48\textwidth}
\centering
\caption{Rejection rates under Null.}
\label{tab:lambda_null}
\begin{tabular}{c|cccc}
\hline
$\lambda_{T,f}\backslash \lambda_{T,K}$
& $1/\sqrt{T}$ & $1/T$ & $1/(T\log T)$ & $1/T^2$ \\
\hline
$\log(T)/\sqrt{T}$ & 0.136 & 0.046 & 0.036 & 0.049 \\
$1/\sqrt{T}$       & 0.056 & 0.052 & 0.041 & 0.063 \\
$1/T$              & 0.033 & 0.040 & 0.051 & 0.044 \\
$1/T^2$            & 0.030 & 0.035 & 0.029 & 0.042 \\
\hline
\end{tabular}
\end{minipage}
\hfill
\begin{minipage}{0.48\textwidth}
\centering
\caption{Rejection rates under Alternative 2.}
\label{tab:lambda_alt2}
\begin{tabular}{c|cccc}
\hline
$\lambda_{T,f}\backslash \lambda_{T,K}$
& $1/\sqrt{T}$ & $1/T$ & $1/(T\log T)$ & $1/T^2$ \\
\hline
$\log(T)/\sqrt{T}$ & 0.728 & 0.699 & 0.697 & 0.684 \\
$1/\sqrt{T}$       & 0.696 & 0.706 & 0.676 & 0.677 \\
$1/T$              & 0.642 & 0.694 & 0.661 & 0.664 \\
$1/T^2$            & 0.574 & 0.621 & 0.647 & 0.668 \\
\hline
\end{tabular}
\end{minipage}
\end{table}

Tables~\ref{tab:lambda_null} and~\ref{tab:lambda_alt2} report the empirical rejection rates under Null and Alternative 2, respectively. Overall, the test demonstrates relatively stable performance across a broad range of regularization choices, provided that the regularization parameters are not chosen excessively large simultaneously. Under Null, the empirical rejection rates remain reasonably close to the nominal significance level, while under Alternative 2 the procedure consistently achieves substantial power. The results also suggest that the effect of $\lambda_{T,f}$ is more pronounced than that of $\lambda_{T,K}$. In particular, overly small values of $\lambda_{T,f}$ tend to make the procedure slightly undersized and slightly reduce power. A more extensive sensitivity analysis, including the effects of bandwidth selection, signal strength, and noise level, is provided in Appendix \ref{sec:appendix_bandwidth}. The implications of these results for the choice of kernel, bandwidth, and regularization parameters are discussed in Section~\ref{sec:tuning_discussion}.

\subsection{Robustness Under Heavy-Tailed Innovations}\label{subsection:heavytail}
In this subsection, the robustness of the test to deviations from Gaussianity is investigated. Only Null and Alternative 2 of Table \ref{tab:gx} are considered. To assess robustness under different tail behaviors, the innovation terms are generated from a symmetrized Weibull distribution. Specifically, each innovation is obtained by drawing a Weibull random variable and randomly assigning it a positive or negative sign with equal probability, resulting in a symmetric distribution centered at zero. The scale parameter is chosen so that the innovations have unit variance. Three values of the Weibull shape parameter are considered. A shape parameter of $2$ corresponds to sub-Gaussian noise, a shape parameter of $1$ produces sub-exponential tails, and a shape parameter of $0.5$ yields heavy-tailed noise.

\begin{table}[!ht]
\centering
\begin{minipage}{0.48\textwidth}
\centering
\caption{Rejection rates under Null.}
\label{tab:tail_null}
\begin{tabular}{c|ccc}
\hline
Shape & $T=100$ & $T=200$ & $T=300$ \\
\hline
0.5 & 0.058 & 0.063 & 0.050 \\
1   & 0.046 & 0.046 & 0.049 \\
2   & 0.046 & 0.049 & 0.049 \\
\hline
\end{tabular}
\end{minipage}
\hfill
\begin{minipage}{0.48\textwidth}
\centering
\caption{Rejection rates under Alternative 2.}
\label{tab:tail_alt2}
\begin{tabular}{c|ccc}
\hline
Shape & $T=100$ & $T=200$ & $T=300$ \\
\hline
0.5 & 0.117 & 0.172 & 0.251 \\
1   & 0.450 & 0.857 & 0.988 \\
2   & 0.857 & 0.999 & 1.000 \\
\hline
\end{tabular}
\end{minipage}
\end{table}

Tables~\ref{tab:tail_null} and~\ref{tab:tail_alt2} report the empirical rejection rates under Null and Alternative 2, respectively. Under Null, the empirical rejection rates remain close to the nominal significance level across all tail settings, indicating that the test maintains good finite-sample calibration even under heavy-tailed innovations. Under Alternative 2, the procedure continues to achieve substantial power under sub-Gaussian and sub-exponential noise, with power approaching one as the sample size increases. As expected, the presence of heavy-tailed noise with shape parameter $0.5$ substantially reduces power due to the increased variability of the innovations. Nevertheless, the rejection rates still increase with the sample size, demonstrating a degree of robustness of the proposed procedure under heavy-tailed innovations.

\subsection{Comparison}
\label{se:comparison_nishiyama}

This subsection compares the RKHS test with the nonparametric causality test of \citet{nishiyama2011consistent}, which we refer to as ``NHKJ-test'', alluding to the authors' last names. Their procedure
formulates nonlinear non-causality in conditional mean as a collection of
unconditional moment restrictions. Specifically, after estimating the null regression of $X_t$ on $X_{t-1}$, the residualized component is tested
against a user-specified collection of basis functions of $(X_{t-1},Y_{t-1})$. The resulting test statistic is a weighted sum of squared standardized sample moments,
\begin{equation}
    \widehat S_T = \sum_{i=1}^{k_T} w_i \widehat a_i^2 .
\end{equation}
The basis functions and the weights determine the directions of departure from the null to which the test is most sensitive. In particular, basis
functions appearing earlier in the ordering receive larger weights when the decreasing sequence $w_i=0.9^i$ is used. Thus, prior knowledge about
the likely alternative can be incorporated by placing the corresponding basis directions earlier in the expansion.

Following the simulation design of \citet{nishiyama2011consistent}, the
competitor is implemented using eight trigonometric basis functions and weights $w_i=0.9^i$, together with their reported asymptotic 5\% critical value. The basis functions are ordered as
\begin{equation}
\cos(y),\quad
\sin(y),\quad
\sin(y)\sin(x),\quad
\sin(y)\cos(x),\quad
\cos(y)\sin(x),\quad
\cos(y)\cos(x),\quad
\sin(2y),\quad
\cos(2y).
\end{equation}
This ordering is deliberately favorable to the competitor under Alternative 1, since the Granger-causal component $\cos(y)$ is exactly the first
basis direction and receives the largest weight. Alternative 2 is also
partly favorable to the competitor, since the interaction
$\sin(x)\cos(y)$ appears among the chosen basis functions, although with a
smaller weight. To further examine the effect of using a finite,
user-specified basis, an additional interaction alternative is considered
only in this subsection:
\begin{equation}
    g_X(x,y)=\sin(x)+0.3xy,
    \qquad
    g_Y(x,y)=\cos(x)+\tanh(y).
\end{equation}
This interaction is not directly represented by the finite trigonometric
basis above.

The bandwidth constant in the NHKJ-test is another important tuning parameter. Their simulation study uses a bandwidth of the form
$h=C T^{-0.3}$ and states that $C$ is approximately 7, with a different
value used for $T=100$, but the exact value is not reported. In finite samples, the empirical size of their test is sensitive to the choice of bandwidth. To make the comparison size-controlled and favorable to
the competitor, an oracle calibration is used: for each sample size, $C$
is selected from a grid under the null model so that the empirical
rejection rate is closest to the nominal level $0.05$. The selected values
are $C=7.00$ for $T=100$ and $C=6.25$ for $T=200$. The same selected value
of $C$ is then used for all alternatives at the corresponding sample size.

\begin{table}[h]
\centering
\begin{tabular}{lcccc}
\toprule
& \multicolumn{2}{c}{$T=100$} & \multicolumn{2}{c}{$T=200$} \\
\cmidrule(lr){2-3}\cmidrule(lr){4-5}
Model & NHKJ-test & RKHS & NHKJ-test & RKHS \\
\midrule
Null                    & 0.054 & 0.049 & 0.056 & 0.049 \\
Alternative 1           & 0.909 & 0.997 & 1.000 & 1.000 \\
Alternative 2           & 0.185 & 0.696 & 0.826 & 0.982 \\
Interaction alternative & 0.287 & 0.760 & 0.606 & 0.994 \\
\bottomrule
\end{tabular}
\caption{Rejection rates for the RKHS test and the test of
\citet{nishiyama2011consistent}. The NHKJ-test is reported using
oracle size-calibrated bandwidth constants, with $C=7.00$ for $T=100$
and $C=6.25$ for $T=200$. The interaction alternative corresponds to
$g_X(x,y)=\sin(x)+0.3xy$.}
\label{tab:nishiyama_comparison}
\end{table}

After bandwidth calibration, the two methods have comparable empirical size. It is worth noting that NHKJ's method is reported using oracle bandwidth calibration, whereas our method is implemented using the default tuning parameters described at the beginning of this section, without any additional calibration. Under Alternative 1, the Granger-causal component is exactly the leading basis function used by NHKJ's method and therefore represents a favorable setting for that procedure. Nevertheless, our method still achieves slightly higher power when $T=100$ and identical power when $T=200$. Under Alternative 2, the Granger-causal component is still represented by the chosen basis functions but receives a smaller weight, leading to a noticeable loss of power for NHKJ's method, whereas our method continues to perform well. Under the interaction alternative, the Granger-causal component is not directly represented by the finite trigonometric basis used by NHKJ's method, and the power difference becomes even more pronounced. Overall, these results demonstrate that our method provides consistently high power across a range of nonlinear alternatives without requiring the nonlinear Granger-causal component to be specified in advance through a finite collection of basis functions.

\section{Real Data Analysis}
\label{se:real_data}

This section illustrates the nonlinear Granger causality test using two real datasets. The examples represent distinct application domains and are used to demonstrate the ability of the test to detect nonlinear predictive relationships in real-world time series.

\subsection{Cardiorespiratory Data from a Sleep Apnea Subject}

The benchmark cardiorespiratory dataset consists of two sequential recordings, \texttt{b1} and \texttt{b2}, each containing 17,000 observations of heart rate, chest volume (respiration force), and blood oxygen concentration collected at 2 Hz from a subject with sleep apnea \citep{dr1994multichannel}. The data are available through PhysioNet \citep{goldberger2000physiobank}. The present analysis uses only the heart rate (HR) and breathing (B) signals from \texttt{b1}. The dataset has been widely used as a benchmark for Granger causality analysis in physiological time series.

The physiological relationship between B and HR is commonly associated with respiratory sinus arrhythmia (RSA), a phenomenon in which heart rate tends to accelerate during inhalation and decelerate during exhalation due to autonomic regulation \citep{saul1989transfer}. From a predictive perspective, this mechanism suggests that past breathing patterns may contain information about future HR dynamics, making cardiorespiratory recordings a natural setting for studying Granger-causal relationships. Furthermore, previous studies have suggested that cardiorespiratory interactions exhibit substantial nonlinear behavior, motivating the use of nonlinear models for their analysis \citep{jo2007nonlinear}.

Despite its widespread use, previous studies have reported differing conclusions regarding the dominant direction of interaction in this dataset. A summary of these findings is provided by \citet{bahrami2023investigating}. As discussed therein, some analyses were based on selected portions of the record, whereas others used relatively long segments. Since the recording spans several hours and exhibits slow nonstationary dynamics, inferred directional relationships may depend on the segment under investigation. To address this issue, \citet{bahrami2023investigating} analyzed the entire record using a rolling-window approach and concluded that (B $\rightarrow$ HR) is the dominant direction of interaction.

Following \citet{bahrami2023investigating}, the \texttt{b1} record is divided into consecutive non-overlapping windows of length 200, corresponding to 1 minute and 40 seconds of observations. Within each window, the two series are standardized separately, and the proposed nonlinear Granger causality test is applied in both directions, (B $\rightarrow$ HR) and (HR $\rightarrow$ B), using the default tuning-parameter choices described in Section~\ref{sec:simulation}.

\begin{figure}[!ht]
\centering
\IfFileExists{Figure_2.png}{
\IfFileExists{Figure_2.png}
{\includegraphics[width=0.8\textwidth]{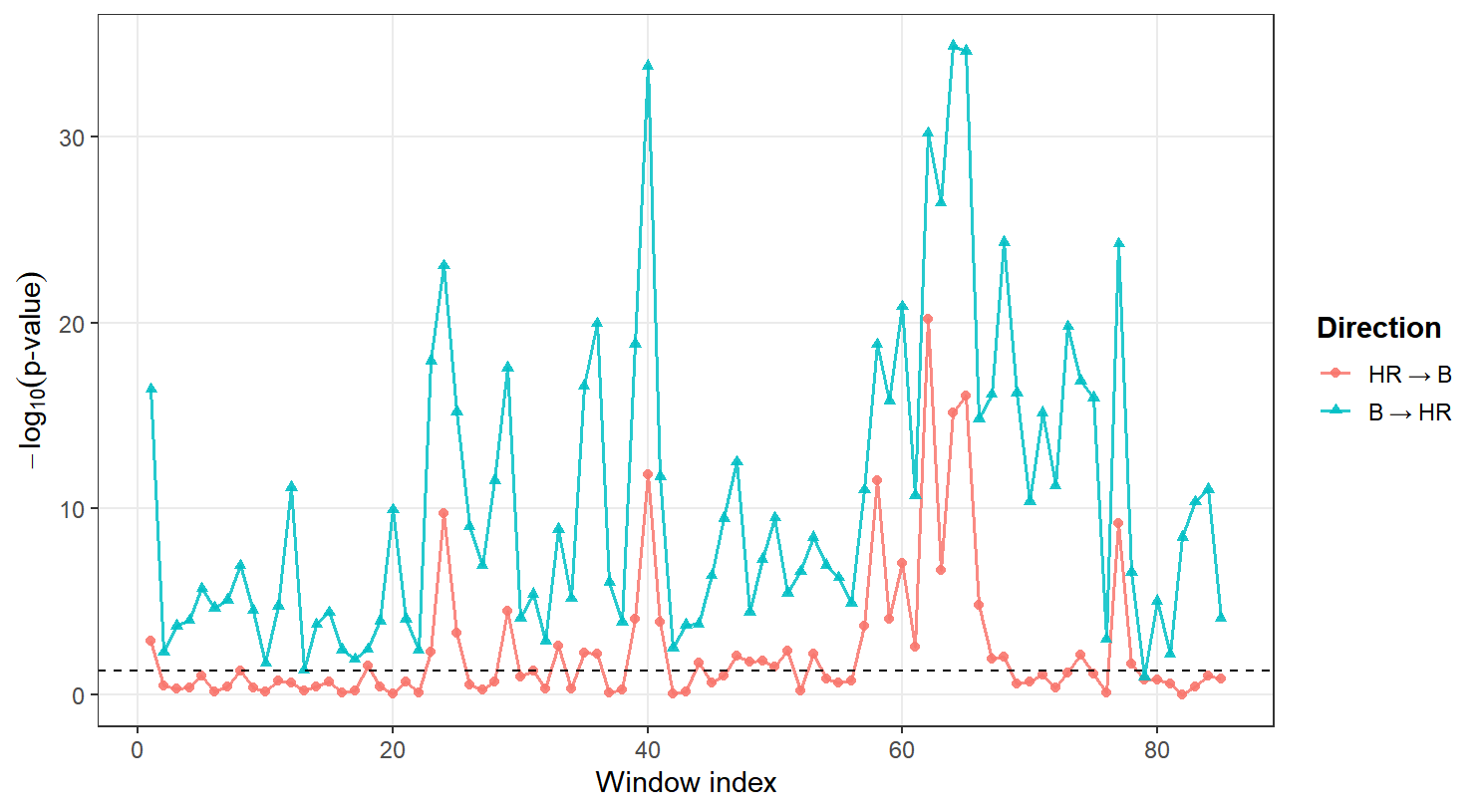}}
{\fbox{\parbox[c][0.24\textheight][c]{0.78\textwidth}{\centering
Cardiorespiratory figure unavailable in the supplied source files.}}}
}{
\fbox{\parbox[c][0.28\textheight][c]{0.8\textwidth}{\centering
Real-data figure omitted because \texttt{Figure_2.png} was not found.}}
}
\caption{Rolling-window nonlinear Granger causality analysis of the \texttt{b1} cardiorespiratory dataset. The dashed horizontal line corresponds to the significance threshold ($\alpha = 0.05$).}
\label{fig}
\end{figure}

Figure~\ref{fig} shows the resulting $-\log_{10}(p\text{-value})$ values across all windows. Evidence of Granger causality from B to HR is generally stronger and more persistent than evidence of Granger causality from HR to B, with the latter reaching significance only intermittently. This overall pattern is consistent with the physiological interpretation of RSA and agrees with the findings of \citet{bahrami2023investigating}.

The strength of the detected Granger-causal relationships nevertheless varies considerably throughout the recording, indicating that the evidence for directional Granger causality is not uniform over time. Such temporal variation is plausible for a long recording from a sleep apnea subject, where changes in physiological state and sleep-apnea-related mechanisms may be associated with changes in cardiorespiratory Granger causality. It may also help explain differing conclusions reported in the literature. By examining the full recording through rolling windows rather than relying on a single segment, the analysis provides a more localized view of temporal variation in Granger causality while retaining an overall stronger pattern from B to HR.

\subsection{Solar and Geomagnetic Activity}

The data are obtained from the dataset compiled and maintained by the GFZ Helmholtz Centre for Geosciences \citep{Matzka2021Kp}. The dataset contains several geomagnetic and solar activity indices, including the planetary geomagnetic indices Kp, ap, and Ap, as well as the sunspot number (SN). In this study, the daily Ap index and daily sunspot number are considered over the period from 1932-01-01 to 2019-12-31, corresponding to the end of the last complete solar cycle available at the time of analysis.

The sunspot number is a standard measure of solar magnetic activity, while the Ap index provides a linear-scale measure of daily global geomagnetic activity derived from the planetary Kp index \citep{matzka2021geomagnetic}. Solar activity is a primary driver of geomagnetic disturbances through interactions between the solar wind and the Earth's magnetosphere. Consequently, variations in solar activity, as reflected by the sunspot number, are often associated with subsequent changes in geomagnetic activity measured by the Ap index. From a predictive perspective, this relationship suggests that past values of SN may contain information useful for forecasting future values of Ap, making this pair of series a natural setting for investigating Granger-causal relationships. Both series also exhibit pronounced long-term variability associated with the approximately 11-year solar cycle. To reduce the influence of long-term trend and seasonal effects, the daily observations are first aggregated into monthly averages and then decomposed using STL decomposition \citep{cleveland1990stl}. The remainder components are subsequently used for Granger causality analysis.

After preprocessing, the final dataset contains 1056 monthly observations, comprising eight consecutive 11-year segments. Both the nonlinear Granger causality test and the classical linear Granger causality test are applied to the entire dataset as well as separately to each 11-year segment. The nonlinear Granger causality test is implemented using the same default settings as those used in the simulation studies described in Section~\ref{sec:simulation}. In order to use the default bandwidth for the Gaussian kernel, the monthly remainder series are standardized separately prior to applying both the nonlinear and linear Granger causality procedures.

\begin{table}[!ht]
\centering
\caption{$p$-values for testing whether SN Granger-causes Ap using the proposed nonlinear method and the classical linear Granger causality test.}
\label{tab:realdata_gc}
\begin{tabular}{c|c|c|cc}
\hline
Segment & Period & $T$ & Nonlinear & Linear \\
\hline
0 & 1932--2019 & 1056 & 0.0001 & 0.0196 \\
1 & 1932--1942 & 132 & 0.0003 & 0.0547 \\
2 & 1943--1953 & 132 & 0.0757 & 0.6983 \\
3 & 1954--1964 & 132 & 0.0111 & 0.0246 \\
4 & 1965--1975 & 132 & 0.0320 & 0.0549 \\
5 & 1976--1986 & 132 & 0.5809 & 0.1799 \\
6 & 1987--1997 & 132 & 0.1046 & 0.1560 \\
7 & 1998--2008 & 132 & 0.0218 & 0.7295 \\
8 & 2009--2019 & 132 & 0.4698 & 0.5459 \\
\hline
\end{tabular}
\end{table}

Table~\ref{tab:realdata_gc} reports the resulting $p$-values for testing whether SN Granger-causes Ap. For the full dataset, both methods detect significant predictive relationships from SN to Ap, with the proposed nonlinear method producing a substantially smaller $p$-value. When the analysis is performed separately over individual 11-year segments, the nonlinear method detects significant relationships in several periods where the linear method does not. In particular, the nonlinear test rejects the null in Segments 1, 4, and 7, whereas the corresponding linear tests do not at the $0.05$ level.

These results suggest that the predictive relationship between solar activity and geomagnetic activity may contain nonlinear components that are not fully captured by linear autoregressive models. At the same time, the variability across 11-year segments also indicates that the strength and form of the relationship may evolve over time, which is consistent with the known phase-dependent behavior of geomagnetic activity throughout the solar cycle.

\section{Conclusion}

This paper considers nonlinear Granger causality within a nonlinear vector
autoregressive framework. By introducing an orthogonal decomposition of the
conditional mean, the nonlinear Granger-causal effect is isolated as the
component that cannot be explained by the own past of the response variable,
thereby reducing nonlinear Granger causality to the problem of testing whether
an unknown nonlinear function is identically zero. Based on this formulation,
an RKHS-based testing procedure was developed by first removing the own-history
component through kernel ridge regression and then embedding the resulting
residuals into an RKHS generated by a conditionally centered kernel. This
construction leads naturally to a $\clh_2$-valued stochastic process together
with a computationally simple quadratic test statistic based on its RKHS norm.

For the oracle construction based on a population conditionally centered
kernel, we established the asymptotic behavior of both the
$\clh_2$-valued residual embedding and the corresponding quadratic test
statistic under the null hypothesis and fixed alternatives. Since the
population conditional-centering map is generally unknown, we then developed
a feasible procedure that estimates this map from the data and uses the
resulting centered feature representations in the test statistic. The main theoretical result for the feasible procedure is that empirical
centering does not alter the asymptotic null distribution. The proof relies
on decomposing the discrepancy from the oracle construction into several
remainder terms and controlling the contributions from first-stage KRR
estimation, empirical conditional centering, and the innovations. A key step
is an exact regularized-operator factorization, which isolates the difficult
innovation-weighted centering contribution and allows it to be controlled
through operator bounds and martingale arguments. Consequently, the feasible
residual embedding retains the oracle null limit, and the same weighted
chi-square limit applies to the feasible test statistic. For practical calibration, the relevant population covariance spectrum is
approximated using the spectrum of the empirical centered Gram matrix, which
provides the weights for computing critical values and $p$-values.

Beyond the proposed Granger causality test, the theoretical developments are
of independent interest. In particular, a central limit theorem was established for a class of
RKHS-valued stochastic processes arising from nonlinear vector
autoregressive models, with the $\clh_2$-valued residual embedding considered
in the testing procedure as a special case. The
convergence of empirical covariance operators was also established, providing
an infinite-dimensional analogue of the convergence of sample covariance
matrices in classical linear vector autoregressive analysis. Together, these
results furnish RKHS counterparts of the fundamental asymptotic tools
underlying inference for linear vector autoregressive models and provide a
theoretical foundation for RKHS-based asymptotic inference in nonlinear vector
autoregressive systems. More broadly, the orthogonal decomposition framework
developed in this paper offers a principled approach for separating own-history
and cross-variable components and may prove useful for other kernel-based
inference problems involving nonlinear time series.

\appendix

\section{Notation Tables}
\label{app:notation_tables}

This appendix collects the main notation used throughout the paper. The tables are
intended as a guide to recurring objects and do not include every auxiliary
quantity introduced only locally inside proofs.

\begingroup
\small
\renewcommand{\arraystretch}{1.15}

\begin{longtable}{@{}p{0.1\textwidth}p{0.75\textwidth}@{}}
\caption{Operator notation used in the population and asymptotic analysis.}
\label{tab:notation_operators}\\
\toprule
Symbol & Description \\
\midrule
\endfirsthead
\toprule
Symbol & Description \\
\midrule
\endhead
\bottomrule
\endfoot
$\clh_1,K_1$ & RKHS and kernel used for the first-stage regression on $X_{t-1}$. \\
$\clh_2,K_2^c$ & Conditionally centered RKHS and kernel used for detecting the $Y$-dependent component. \\
$\calK_1,\calK_2^c$ & Evaluation operators associated with $K_1$ and $K_2^c$. \\
$L_{11},L_{12}$ & Population covariance and cross-covariance operators involving the first-stage RKHS, \eqref{eq:defineL12L21L11L22}. \\
$L_{21}^c,L_{22}^c$ & Population cross-covariance and covariance operators involving the centered second-stage RKHS, \eqref{eq:defineL12L21L11L22}. \\
$\widehat L_{11},\widehat L_{12}$ & Empirical covariance and cross-covariance operators involving the first-stage RKHS, \eqref{eq:sampleLhatijoperator}. \\
$\widehat L_{21}^c,\widehat L_{22}^c$ & Empirical cross-covariance and covariance operators involving the centered second-stage RKHS, \eqref{eq:sampleLhatijoperator}. \\
\end{longtable}

\begin{longtable}{@{}p{0.1\textwidth}p{0.75\textwidth}@{}}
\caption{Notation for empirical conditional centering.}
\label{tab:notation_empirical_centering}\\
\toprule
Symbol & Description \\
\midrule
\endfirsthead
\toprule
Symbol & Description \\
\midrule
\endhead
\bottomrule
\endfoot
$\widetilde\clh_2,\widetilde K_2$ & Uncentered second-stage kernel and its associated RKHS.\\
$m,m_\lambda$ & Population conditional feature mean and its regularized target, \eqref{eq:full_m_target} and \eqref{eq:full_population_krr_centering}. \\
$\widetilde L_{21}$ & Population uncentered cross-covariance operator, \eqref{eq:full_raw_L21hat}. \\
$\widehat{\widetilde L}_{21,T}$ & Uncentered empirical cross-covariance operator, \eqref{eq:full_raw_L21hat}.  \\
$A_\lambda,\widehat A_{\lambda,T}$ & Population and empirical regularized KRR operators defining $m_\lambda$ and $\widehat m_T$, \eqref{eq:full_population_krr_centering} and \eqref{eq:full_empirical_krr_operator}. \\
$\calG,\Gamma$ & $\widetilde\clh_2$-valued RKHS and operator-valued kernel used to estimate $m$, \eqref{eq:full_operator_kernel}. \\
$\widehat m_T$ & Vector-valued KRR centering estimator, \eqref{eq:full_mhat_def}. \\
$S_K$ & KRR centering smoother, \eqref{eq:SK_residualizer_def}. \\
$k_t^c,\widehat k_{t,T}^c$ & Population- and empirically centered feature sections, \eqref{eq:full_centered_sections}. \\
$\widehat K_2^c$ & Empirically centered second-stage Gram matrix, \eqref{eq:K2hatc_residualized_def}. \\
$\widehat Z_T^{\mathrm{res}}$ & Feasible residual embedding, \eqref{eq:full_residual_embedding}. \\
$\widehat Q_T^{\mathrm{res}}$ & Feasible quadratic statistic, \eqref{eq:full_residual_quadratic}. \\
\end{longtable}

\begin{longtable}{@{}p{0.1\textwidth}p{0.75\textwidth}@{}}
\caption{Notation for empirical-centering rates and limits.}
\label{tab:notation_empirical_centering_rates}\\
\toprule
Symbol & Description \\
\midrule
\endfirsthead
\toprule
Symbol & Description \\
\midrule
\endhead
\bottomrule
\endfoot
$a_T$ & Rate of the first-stage estimator, \eqref{eq:full_first_stage_rate}. \\
$b_{\lambda,c}^2$ & Squared population centering bias, \eqref{eq:full_centering_bias}. \\
$\psi_{1,\lambda}$ & Regularized kernel section, \eqref{eq:full_regularized_feature}. \\
$\widehat R_{c,T}$ & Empirical mean squared centering error of $\widehat m_T$, \eqref{eq:full_centering_risk}. \\
$L_c,\widehat L_{c,T}$ & Population and feasible centered covariance operators, \eqref{eq:full_Lc_def} and \eqref{eq:full_feasible_covariance}. \\
\end{longtable}

\endgroup

\section{Proofs of main results in Section \ref{subsec:asymptotic_distribution_consistency}}
\label{app:appendixasymptotic_distribution_consistency}

In this section, we present the proofs regarding the asymptotic distribution and consistency of our oracle test statistic.

\begin{proof}[Proof of Theorem~\ref{thm:ZT_limit}]

Under the null hypothesis $g_{X,0} = 0$, the decomposition \eqref{eq:ZT_operator} reduces to \eqref{eq:ZT_operator_null}, i.e.,
\begin{equation} \label{eq:proof:thm:ZT:eq1}
Z_T
=
G_T
+
\sqrt{T}\widehat L_{21}^c(f_X-\widehat f).
\end{equation}
We consider the two summands on the right-hand side of \eqref{eq:proof:thm:ZT:eq1} separately.  For $G_T$, by Proposition \ref{prop:aux_FCLT}, we have
\begin{equation}\label{eq:convGTtodGinH2}
G_T
\;\xrightarrow[]{\distr}\;
\mathcal G_2
\qquad
\text{in }\clh_2,
\end{equation}
so it just remains to show that the second summand in \eqref{eq:proof:thm:ZT:eq1} is negligible, i.e., 
\begin{equation}\label{eq:negligibilitylhat21fxminusfhat}
\left\|
\sqrt{T}\widehat L_{21}^c(f_X-\widehat f)
\right\|_{\clh_2}
=
o_{\PP}(1).
\end{equation}
Note that, with explanations given below, 
\begin{align}
    \left\|
    \sqrt{T}\widehat L_{21}^c(f_X-\widehat f)
    \right\|_{\clh_2}
    &=
    \left\|
    \sqrt{T}
    (\widehat L_{21}^c-L_{21}^c) (f_X-\widehat f)
    +
    \sqrt{T}
    L_{21}^c (f_X-\widehat f)
    \right\|_{\clh_2}
    \label{eq:proof:thm:ZT:eq2}
    \\&=
    \left\|
    \sqrt{T}
    (\widehat L_{21}^c-L_{21}^c) (f_X-\widehat f)
    \right\|_{\clh_2}
    \label{eq:proof:thm:ZT:eq3}
    \\&\leq
    \sqrt{T}
    \| \widehat L_{21}^c-L_{21}^c \|_{\op}
    \left\|
    f_X-\widehat f
    \right\|_{\clh_1},
    \label{eq:proof:thm:ZT:eq4}
    \\&=
    O_{\PP}(1) o_{\PP}(1)=o_{\PP}(1),
    \label{eq:proof:thm:ZT:eq5}
\end{align}
where \eqref{eq:proof:thm:ZT:eq3}
is due to the orthogonality identity 
\eqref{eq:L21orthogonality} with $f = f_X - \widehat f$, i.e., $L_{21}^c(f_X-\widehat f)=0.$
The inequality in \eqref{eq:proof:thm:ZT:eq4} follows from  the definition of the operator norm, and the first identity in \eqref{eq:proof:thm:ZT:eq5} follows from Corollary \ref{cor:aux_operator_conv_prob} and Assumption \ref{assump:first_stage_consistency}.
The result follows from \eqref{eq:convGTtodGinH2}, \eqref{eq:negligibilitylhat21fxminusfhat}, and  Slutsky's theorem.
\end{proof}

\begin{proof}[Proof of Theorem~\ref{thm:ZT_alt}]
Starting from the decomposition \eqref{eq:ZT_operator_alternative}, subtracting
$\sqrt T\,\widehat L_{22}^c g_{X,0}$ gives
\begin{align}
Z_T-\sqrt T\,\widehat L_{22}^c g_{X,0}
&=
G_T
+
\sqrt T\,\widehat L_{21}^c(f_X-\widetilde f)
-
\sqrt T\,
\widehat L_{21}^c
  (\widehat L_{11}+\lambda_{T,f}I_{\clh_1})^{-1}
\widehat L_{12}g_{X,0}.
\label{eq:proof:thm:ZT_alt:decomp}
\end{align}
By Proposition \ref{prop:aux_FCLT}, $G_T \xrightarrow{\distr} \mathcal G_2$ in $\clh_2$.
It remains to show that the last two terms in
\eqref{eq:proof:thm:ZT_alt:decomp} are $o_{\PP}(1)$ in $\clh_2$.

For the term $\sqrt T\,\widehat L_{21}^c(f_X-\widetilde f)$, under Assumption \ref{assump:first_stage_consistency}, the same arguments as in
the proof of Theorem~\ref{thm:ZT_limit} give
\begin{equation}
\left\|
\sqrt T\,\widehat L_{21}^c(f_X-\widetilde f)
\right\|_{\clh_2}
=
o_{\PP}(1),
\end{equation}
where $\widetilde f$ here plays the role of $\widehat f$ in that proof.

It remains to control the second summand in \eqref{eq:proof:thm:ZT_alt:decomp}.
We show that this term is $o_{\PP}(1)$ in $\clh_2$. Note that, with explanations given below,
\begin{align}
&\left\|
\sqrt T\,
\widehat L_{21}^c
  (\widehat L_{11}+\lambda_{T,f}I_{\clh_1})^{-1}
\widehat L_{12}g_{X,0}
\right\|_{\clh_2}
\notag\\
&=
\left\|
\sqrt T\,
(\widehat L_{21}^c-L_{21}^c)
  (\widehat L_{11}+\lambda_{T,f}I_{\clh_1})^{-1}
(\widehat L_{12}-L_{12})g_{X,0}
\right\|_{\clh_2}
\label{eq:proof:thm:ZT_alt:eq1}
\\
&\le
\sqrt T\,
\left\|\widehat L_{21}^c-L_{21}^c\right\|_{\op}
\left\|
  (\widehat L_{11}+\lambda_{T,f}I_{\clh_1})^{-1}
\right\|_{\op}
\left\|\widehat L_{12}-L_{12}\right\|_{\op}\left\|g_{X,0}\right\|_{\clh_2}
\label{eq:proof:thm:ZT_alt:eq2}
\\
&\le
\sqrt T\,
\|\widehat L_{21}^c-L_{21}^c\|_{\op}
\lambda_{T,f}^{-1}
\|\widehat L_{12}-L_{12}\|_{\op}
\|g_{X,0}\|_{\clh_2}
\label{eq:proof:thm:ZT_alt:eq3}
\\
&=
O_{\PP}
\left(
\frac{1}{\lambda_{T,f}\sqrt T}
\right)
=
o_{\PP}(1).
\label{eq:proof:thm:ZT_alt:eq4}
\end{align}
The equality \eqref{eq:proof:thm:ZT_alt:eq1} follows from two orthogonality relations. First, the conditional-centering condition \eqref{eq:cond_center} implies
\begin{equation}
L_{21}^c h=0,
\qquad h\in\clh_1.
\end{equation}
Second, since $g_{X,0}\in\clh_2\subseteq L^2_{X,0}(\pi)$, we have $L_{12}g_{X,0}=0$.
The inequality
\eqref{eq:proof:thm:ZT_alt:eq2} follows from the definition of the operator
norm. The inequality \eqref{eq:proof:thm:ZT_alt:eq3} uses
\begin{equation}
\left\|
  (\widehat L_{11}+\lambda_{T,f}I_{\clh_1})^{-1}
\right\|_{\op}
\le
\lambda_{T,f}^{-1},
\end{equation}
which holds since $\widehat L_{11}$ is self-adjoint and positive semidefinite.
Finally, \eqref{eq:proof:thm:ZT_alt:eq4} follows from
Corollary \ref{cor:aux_operator_conv_prob}, which gives
\begin{align}
\|\widehat L_{21}^c-L_{21}^c\|_{\op}
&=
O_{\PP}(T^{-1/2}),
\qquad
\text{as an operator from } \clh_1 \text{ to } \clh_2,
\\
\|\widehat L_{12}-L_{12}\|_{\op}
&=
O_{\PP}(T^{-1/2}),
\qquad
\text{as an operator from } \clh_2 \text{ to } \clh_1,
\end{align}
together with the condition $\lambda_{T,f}\sqrt T\to\infty$.

Combining the decomposition \eqref{eq:proof:thm:ZT_alt:decomp} with the two
negligible terms above gives
\begin{equation}
Z_T-\sqrt T\,\widehat L_{22}^c g_{X,0}
=
G_T+o_{\PP}(1)
\qquad
\text{in } \clh_2.
\end{equation}
By Proposition \ref{prop:aux_FCLT} and Slutsky's theorem,
\begin{equation} \label{eq:ZT-sqrt}
Z_T-\sqrt T\,\widehat L_{22}^c g_{X,0}
\xrightarrow{\distr}
\mathcal G_2
\qquad
\text{in } \clh_2.
\end{equation}

We next prove the convergence of the rescaled quadratic statistic. By \eqref{eq:ZT-sqrt}, $Z_T-\sqrt T \widehat L_{22}^c g_{X,0}
=
O_{\PP}(1)$ in $\clh_2$, since convergence in distribution implies stochastic boundedness. Therefore,
\begin{align}
\left\|
\frac{Z_T}{\sqrt T}-L_{22}^c g_{X,0}
\right\|_{\clh_2}
&\leq
\frac{1}{\sqrt T}
\left\|
Z_T-\sqrt T \widehat L_{22}^c g_{X,0}
\right\|_{\clh_2}
+
\left\|
\widehat L_{22}^c-L_{22}^c
\right\|_{\op}
\|g_{X,0}\|_{\clh_2}
=
o_{\PP}(1).
\label{eq:ineqprop1-2}
\end{align}
Indeed, the first term on the right-hand side of
\eqref{eq:ineqprop1-2} is $o_{\PP}(1)$ by
\eqref{eq:ZT-sqrt}. Moreover, by Corollary
\ref{cor:aux_operator_conv_prob} 
the second term is also $o_{\PP}(1)$, because
$g_{X,0}\in\clh_2$ is fixed. Consequently,
\begin{equation}
\frac{Z_T}{\sqrt T}
\xrightarrow{\PP}
L_{22}^c g_{X,0}
\qquad
\text{in } \clh_2.
\end{equation}

Since $Q_T=\|Z_T\|_{\clh_2}^2$, the continuous mapping theorem gives
\begin{equation}
\frac{Q_T}{T}
=
\left\|
\frac{Z_T}{\sqrt T}
\right\|_{\clh_2}^2
\xrightarrow{\PP}
\|L_{22}^c g_{X,0}\|_{\clh_2}^2.
\end{equation}
If $L_{22}^c g_{X,0}\neq0$, the limit is strictly positive. Hence, $\frac{Q_T}{T}
\xrightarrow{\PP}
c >0,$
which implies that $Q_T\to\infty$ in probability at the linear rate $T$.
\end{proof}

\section{Auxiliary RKHS Limit Results and their Proofs}
\label{appendix:proofs_auxiliary}

This section collects two auxiliary RKHS limit results used in the
proofs of Theorems~\ref{thm:ZT_limit} and \ref{thm:ZT_alt}. They are stated
separately because they are not specific to the Granger-causality null
hypothesis. In particular, they do not require the conditional-centering
condition~\eqref{eq:cond_center}; conditional centering is used only when
the first-stage residualization error is removed from the limiting
distribution of the test statistic.

\begin{proposition}
\label{prop:aux_FCLT}
Under Assumptions~\ref{assump:bounded_dynamics},
\ref{assump:bounded_kernel}, and \ref{assump:gaussian}, as $T \to \infty$,
\begin{equation}
\label{eq:aux_FCLT}
G_T
\doteq
\frac{1}{\sqrt{T}}\mathcal K_2^c\boldsymbol{\varepsilon}
\xrightarrow{\distr}
\mathcal G_2
\qquad
\text{in } \clh_2,
\end{equation}
where $\mathcal G_2$ is a centered Gaussian element in $\clh_2$ with
covariance operator $\sigma^2 L_{22}^c$.
\end{proposition}

\begin{remark}
The same argument applies to
$\frac{1}{\sqrt{T}}\mathcal K_1\boldsymbol{\varepsilon}$ in $\clh_1$, with
covariance operator $\sigma^2 L_{11}$.
\end{remark}

\begin{proof}[Proof of Proposition \ref{prop:aux_FCLT}]
We use the standard characterization of weak convergence in
separable Hilbert spaces: convergence of all one-dimensional projections together with tightness
implies weak convergence \citep{ledoux1991probability}. By Lemma \ref{lem:fdd_convergence} below, we have convergence of all one-dimensional projections. That is,  
for every $h\in\clh_2$,
\begin{equation}
\langle G_T,h\rangle_{\clh_2}
\xrightarrow{\distr}
\langle \mathcal G_2,h\rangle_{\clh_2}.
\end{equation}
It remains to show that $\{G_T:T\in\NN\}$ is tight in $\clh_2$. We employ Lemma 7.1 of \citet{panaretos2013fourier} to prove tightness. The criterion was first introduced in the proof of Theorem 2.7 of \cite{bosq2000linear} and states: Let $\{e_j\}_{j\ge 1}$ be an orthonormal basis of $\clh_2$ such that
\begin{enumerate}[label=\textit{(\roman*)}]
    \item  $\EE\left[ \langle G_T,e_j\rangle_{\clh_2}^2 \right] 
    \le a_j
    \qquad
    j\in\NN,\ T\in\NN,$
    \label{item1.panaretos}
    \item $\sum_{j=1}^\infty a_j < \infty$,
    \label{item2.panaretos}
\end{enumerate}
then $\{ G_T \}$ is tight.

For \ref{item1.panaretos}, we get, with explanations given below, 
\begin{align}
\EE\left[
\left\langle G_T,e_j\right\rangle_{\clh_2}^2
\right]
&=
\EE\left[
\left(
\frac{1}{\sqrt T}
\sum_{t=1}^T
\varepsilon_{X,t}
\left\langle
K_2^c((X_{t-1},Y_{t-1}),\cdot),e_j
\right\rangle_{\clh_2}
\right)^2
\right]                                                        \label{eq:coord-second-moment-a}
\\
&=
\EE\left[
\left(
\frac{1}{\sqrt T}
\sum_{t=1}^T
\varepsilon_{X,t}e_j(X_{t-1},Y_{t-1})
\right)^2
\right]                                                        \label{eq:coord-second-moment-b}
\\
&=
\frac{1}{T}
\sum_{s=1}^T\sum_{t=1}^T
\EE\left[
\varepsilon_{X,s}\varepsilon_{X,t}
e_j(X_{s-1},Y_{s-1})e_j(X_{t-1},Y_{t-1})
\right]                                                        \label{eq:coord-second-moment-c}
\\
&=
\frac{1}{T}
\sum_{t=1}^T
\EE\left[
\varepsilon_{X,t}^2e_j(X_{t-1},Y_{t-1})^2
\right]                                                        \label{eq:coord-second-moment-d}
\\
&=
\frac{1}{T}
\sum_{t=1}^T
\EE\left[
(e_j(X_{t-1},Y_{t-1}))^2
\EE\left(\varepsilon_{X,t}^2\mid \mathcal F_{t-1}\right)
\right]                                                        \label{eq:coord-second-moment-e}
\\
&=
\frac{\sigma^2}{T}
\sum_{t=1}^T
\EE\left[
(e_j(X_{t-1},Y_{t-1}))^2
\right]                                                        =
\sigma^2\EE_\pi\left[ (
e_j(X_0,Y_0))^2
\right],                         \label{eq:coord-second-moment-g}
\end{align}
where the equality in \eqref{eq:coord-second-moment-a} uses the definition of $G_T$ in \eqref{eq:aux_FCLT}, \eqref{eq:coord-second-moment-b} is due to the reproducing property, applicable since $e_j \in \clh_2$.
The identity in \eqref{eq:coord-second-moment-d} uses that the off-diagonal terms vanish. Indeed, for $s<t$,
\begin{align}
&\EE\left[
\varepsilon_{X,s}\varepsilon_{X,t}
e_j(X_{s-1},Y_{s-1})e_j(X_{t-1},Y_{t-1})
\right]  \notag \\
&\qquad =
\EE\left[
\varepsilon_{X,s}
e_j(X_{s-1},Y_{s-1})e_j(X_{t-1},Y_{t-1})
\EE\left(\varepsilon_{X,t}\mid \mathcal F_{t-1}\right)
\right]
=0.
\end{align}
The identity in \eqref{eq:coord-second-moment-e} uses the tower property.
It also uses that $(X_{t-1},Y_{t-1})$ is
$\mathcal F_{t-1}$-measurable. It follows that \ref{item1.panaretos} holds with
\begin{align}\label{eq:defaj}
a_j \doteq \sigma^2  \EE_{\pi} [( e_j (X_0,Y_0))^2], \qquad j \in \NN.
\end{align}

For \ref{item2.panaretos}, recall $a_j$ from \eqref{eq:defaj}. Then, by Parseval's identity and the reproducing property,
\begin{equation}
\sum_{j=1}^\infty a_j
=
\sigma^2
\EE_\pi\left[
\sum_{j=1}^\infty (e_j(X_0,Y_0))^2
\right]
=
\sigma^2
\EE_\pi \left[ \|K_2^c((X_0,Y_0),\cdot)\|_{\clh_2}^2 \right]
=
\sigma^2
\EE_\pi \left[ K_2^c((X_0,Y_0),(X_0,Y_0)) \right]
\le
\sigma^2\kappa^2,
\end{equation}
where the last inequality is due to Assumption \ref{assump:bounded_kernel}.
The same calculation shows that $L_{22}^c : \clh_2 \to \clh_2$ is trace class, and hence the centered Gaussian
element with covariance operator $\sigma^2L_{22}^c$ is well defined.
Thus
\begin{equation}
\EE\left[
\langle G_T,e_j\rangle_{\clh_2}^2
\right]
\le a_j
\qquad
j\in\NN,\ T\in\NN,
\end{equation}
with $\sum_{j=1}^\infty a_j<\infty$. The Hilbert space tightness criterion based on
coordinate second moments therefore implies tightness of $\{G_T:T\in\NN\}$
\citep[see, e.g., Theorem~2.7 of][]{bosq2000linear}; see also \citet{panaretos2013fourier}.

Combining tightness with the convergence of all one-dimensional projections yields $G_T \xrightarrow{\distr}\mathcal G_2$ in $\clh_2$.
\end{proof}

\begin{lemma}
\label{lem:fdd_convergence}
Under Assumptions~\ref{assump:bounded_dynamics}, \ref{assump:bounded_kernel}, and \ref{assump:gaussian}, for every $h \in\clh_2$,
\begin{equation}
\langle G_T,h\rangle_{\clh_2}
\xrightarrow{\distr}
\mathcal N\big(0,\sigma^2 \EE_\pi[h(X_0,Y_0)^2]\big),
\end{equation}
where $\EE_\pi[h(X_0,Y_0)^2]
=
\langle L_{22}^c h,h\rangle_{\clh_2}.$
\end{lemma}

\begin{proof}
Since $h \in \clh_2$, by the reproducing property, 
\begin{equation}
S_T
\doteq
\langle G_T,h\rangle_{\clh_2}
=
\left\langle 
\frac{1}{\sqrt{T}} \mathcal K_2^c \boldsymbol{\varepsilon} ,h \right\rangle_{\clh_2}
=
\frac{1}{\sqrt{T}}
\sum_{t=1}^T
\varepsilon_{X,t}h(X_{t-1},Y_{t-1})
=
\sum_{t=1}^T\xi_{t,T},
\end{equation}
where we define
\begin{equation}
\xi_{t,T}
\doteq
w_{t,T}\varepsilon_{X,t},
\qquad
w_{t,T}\doteq\frac{1}{\sqrt T}h(X_{t-1},Y_{t-1}).
\end{equation}
We verify the conditions of the martingale central limit
theorem given in Theorem 5.3.4 of \citet{fuller1995introduction}. 
Let $\mathcal{F}_{t-1}$ denote the sigma algebra generated by $\{ (X_s,Y_s):0\le s\le t-1\}$. 
For $S_T$ to converge in distribution, we need to show that
\begin{enumerate}[label=\textit{(\roman*)}]
    \item $\EE[\xi_{t,T}\mid\calF_{t-1}]=0$ for $1 \leq t \leq T$,
    \label{item1:Fuller-cond}
    \item $V_T^2 \xrightarrow{\PP} \sigma^2 \EE_\pi[h(X_0,Y_0)^2]$ with $V_T^2\doteq\sum_{t=1}^T \delta_{t,T}^2$ and $\delta_{t,T}^2
\doteq
\EE[\xi_{t,T}^2\mid \calF_{t-1}]$,
    \label{item2:Fuller-cond}
    \item $\sum_{t=1}^T
    \EE\left[ \xi_{t,T}^2 1_{\{|\xi_{t,T}|\ge \varepsilon s_T\}} \mid \calF_{t-1}
    \right] \xrightarrow{\PP} 0$ for all $\varepsilon > 0$.
    \label{item3:Fuller-cond}
\end{enumerate}
We prove the three conditions separately.

\textit{Proof of \ref{item1:Fuller-cond}:}
By Assumption~\ref{assump:gaussian}, $\{\xi_{t,T},\calF_t\}_{t=1}^T$ is a martingale
difference triangular array, since $\EE[\xi_{t,T}\mid\calF_{t-1}]=0.$
Moreover,
\begin{equation}
\delta_{t,T}^2
\doteq
\EE[\xi_{t,T}^2\mid \calF_{t-1}]
=
\sigma^2 w_{t,T}^2.
\end{equation}

\textit{Proof of \ref{item2:Fuller-cond}:}
Let
\begin{equation}
V_T^2\doteq\sum_{t=1}^T \delta_{t,T}^2
=
\frac{\sigma^2}{T}\sum_{t=1}^T h(X_{t-1},Y_{t-1})^2,
\qquad
s_T^2\doteq\EE[V_T^2].
\end{equation}
Since the process is initialized in stationarity,
$s_T^2
=
\sigma^2 \EE_\pi[h(X_0,Y_0)^2].$
If $\EE_\pi[h(X_0,Y_0)^2]=0$, then $S_T=0$ almost surely for every $T$, and the claim is
immediate. We therefore assume that $s_T^2>0$. 
Since $K_2^c((x,y),(x',y'))$ is bounded, the reproducing property gives
\begin{equation}
\|h\|_\infty\le c_1\|h\|_{\clh_2}<\infty,
\end{equation}
and therefore that $\| h^2\|_{\infty} = \| h\|^2_{\infty} \le c_1^2 \| h \|^2_{\clh_2}$. Thus, we can apply Theorem~1 of
\citet{jensen2007law} and infer
\begin{equation}
V_T^2
\xrightarrow{\PP}
\sigma^2 \EE_\pi[h(X_0,Y_0)^2].
\end{equation}
Thus
$
s_T^{-2}V_T^2
\xrightarrow{\PP}1.
$
This completes the proof of \ref{item2:Fuller-cond}.

\textit{Proof of \ref{item3:Fuller-cond}:}
It remains to verify the Lindeberg condition. Let $\varepsilon>0$. Since $s_T^2$ is constant
and positive, it suffices to show that
\begin{equation}\label{eq:lindebergcond1}
\sum_{t=1}^T
\EE\left[
\xi_{t,T}^2
1_{\{|\xi_{t,T}|\ge \varepsilon s_T\}}
\mid \calF_{t-1}
\right]
\xrightarrow{\PP}0.
\end{equation}
Furthermore, $|w_{t,T}|\le \frac{c_2}{\sqrt T}.$
Thus,
on the event $\{|\xi_{t,T}|\ge \varepsilon s_T\}$, either $w_{t,T}=0$ and the event is empty,
or $|\varepsilon_{X,t}|\ge c_3\sqrt T$. Lemma~\ref{lem:subgaussian-truncated-second-moment}
therefore gives constants $c_4,c_5\in(0,\infty)$ such that
\begin{equation}
\EE\left[
\xi_{t,T}^2
1_{\{|\xi_{t,T}|\ge \varepsilon s_T\}}
\mid \calF_{t-1}
\right]
\le
w_{t,T}^2
\EE\!\left[
\varepsilon_{X,t}^2
\mathbf{1}_{\{|\varepsilon_{X,t}|\ge c_3\sqrt{T}\}}
\right]
\le
c_4 w_{t,T}^2\exp(-c_5T).
\end{equation}
Consequently,
\begin{equation}
\sum_{t=1}^T
\EE\left[
\xi_{t,T}^2
1_{\{|\xi_{t,T}|\ge \varepsilon s_T\}}
\mid \calF_{t-1}
\right]
\le
c_4\exp(-c_5T)
\sum_{t=1}^T w_{t,T}^2
\le
c_6\exp(-c_5T)
\xrightarrow{}0,
\end{equation}
so \eqref{eq:lindebergcond1} holds. This completes the proof of \ref{item3:Fuller-cond}.

The martingale central limit theorem yields
\begin{equation}
S_T
\xrightarrow{\distr}
\mathcal N\big(0,\sigma^2 \EE_\pi[h(X_0,Y_0)^2]\big).
\end{equation}
It remains to identify the covariance. Since $\mathcal G_2$ is a centered Gaussian with
covariance operator $\sigma^2L_{22}^c$,
\begin{equation}
\langle \mathcal G_2,h\rangle_{\clh_2}
\sim
\mathcal N\left(
0,
\sigma^2 \langle L_{22}^c h,h\rangle_{\clh_2}
\right).
\end{equation}
By the definition of $L_{22}^c$ and the reproducing property, $\langle L_{22}^c h,h\rangle_{\clh_2}
=
\EE_\pi[h(X_0,Y_0)^2].$
\end{proof}

\begin{proposition}
\label{prop:aux_operator_convergence}
Recall that $\|\cdot\|_{\op}$ and $\|\cdot\|_{\cls_2}$ denote the operator
norm and Hilbert--Schmidt norm, respectively. Under
Assumptions~\ref{assump:bounded_dynamics}, \ref{assump:bounded_kernel}, and
\ref{assump:gaussian}, for each pair
\begin{equation}
(\widehat L,L)\in
\{(\widehat L_{11},L_{11}), (\widehat L_{12},L_{12}),
(\widehat L_{21}^c,L_{21}^c), (\widehat L_{22}^c,L_{22}^c)\},
\end{equation}
we have
\begin{equation}
\EE_\pi
\left[
\|\widehat L-L\|_{\cls_2}
\right]
=
O(T^{-1/2}),
\qquad
\EE_\pi
\left[
\|\widehat L-L\|_{\op}
\right]
=
O(T^{-1/2}).
\end{equation}
\end{proposition}

\begin{proof}
It suffices to prove the statement for $\widehat L_{12}$, since the remaining cases are obtained by
replacing the corresponding kernel sections.  For $(x,y)\in\RR^2$, define the
operator $\Phi_{12}(x,y):\clh_2\to\clh_1$ by
\begin{equation}
\Phi_{12}(x,y)g
\doteq
\left\langle g,K_2^c((x,y),(\cdot,\cdot))\right\rangle_{\clh_2}
K_1(x,\cdot),
\qquad g\in\clh_2.
\end{equation}
For notational convenience, let $\cls_2 \doteq \cls_{2}(\clh_2,\clh_1)$ and, for each $(x,y) \in \RR^2$, note that $\Phi_{12}(x,y)\in \cls_2$. By the reproducing property, for $g\in\clh_2$,
\begin{equation}
\Phi_{12}(x,y)g
=
g(x,y)K_1(x,\cdot).
\end{equation}
Recall $\widehat L_{12}:\clh_2\to\clh_1$ from \eqref{eq:Lhat12def} and write
\begin{equation}
\widehat L_{12}
=
\frac{1}{T}\sum_{t=1}^T\Phi_{12}(X_{t-1},Y_{t-1}),
\qquad
L_{12}
=
\EE_\pi[\Phi_{12}(X,Y)],
\end{equation}
such that
\begin{equation} \label{eq:Lhat12minusL12}
\widehat L_{12}-L_{12}
=
\frac{1}{T}\sum_{t=1}^T \Phi_{12}(X_{t-1},Y_{t-1})-\EE_\pi[\Phi_{12}(X,Y)]
\doteq
\frac{1}{T}\sum_{t=1}^T\Delta_t.
\end{equation}
By the formula for the Hilbert--Schmidt norm of a rank-one operator and the reproducing
property, it follows that
\begin{align}
\sup_{(x,y)\in\RR^2}\|\Phi_{12}(x,y)\|_{\cls_2}
&=
\sup_{(x,y)\in\RR^2}\|K_1(x,\cdot)\|_{\clh_1}
\|K_2^c((x,y),(\cdot,\cdot))\|_{\clh_2}
\\&=
\sup_{(x,y)\in\RR^2}\sqrt{K_1(x,x)}\sqrt{K_2^c((x,y),(x,y))}
\leq \kappa^2
\end{align}
by Assumption \ref{assump:bounded_kernel}. In particular, 
\begin{equation} \label{eq:momentbounddelta}
    \EE_\pi [\|\Delta_t\|^2_{\cls_2}] 
    \leq
2 \EE_\pi \left[\|\Phi_{12}(X_{t-1},Y_{t-1})\|^2_{\cls_2}\right]
+
2 \|\EE_\pi[\Phi_{12}(X,Y)]\|^2_{\cls_2}
\leq
4\kappa^4.
\end{equation}

Based on \eqref{eq:Lhat12minusL12}, we obtain
\begin{equation}\label{eq:proofoftheorem42:eq4}
\EE_\pi \left[ \|\widehat L_{12}-L_{12}\|_{\cls_2}^2 \right]
=
\frac{1}{T^2}
\sum_{s=1}^T\sum_{t=1}^T
\EE_\pi
\left[ \langle \Delta_s,\Delta_t\rangle_{\cls_2} \right]
=
\frac{1}{T^2}
\sum_{t=1}^T
\EE_\pi [\|\Delta_t\|_{\cls_2}^2]
+
\frac{2}{T^2}
\sum_{1\le s<t\le T}
\EE_\pi
[\langle \Delta_s,\Delta_t\rangle_{\cls_2}].
\end{equation}
We consider the two summands in \eqref{eq:proofoftheorem42:eq4} separately. 

For the diagonal terms, by \eqref{eq:momentbounddelta}, 
\begin{align}
\frac{1}{T^2}
\sum_{t=1}^T
\EE_\pi[\|\Delta_t\|_{\cls_2}^2]
\le
\frac{4\kappa^4}{T}.
\label{eq:EpiDeltatS2momentdiagonal}
\end{align}

For the second summand in \eqref{eq:proofoftheorem42:eq4}, we obtain
\begin{align}
\EE_\pi \left[ \langle \Delta_s,\Delta_t\rangle_{\cls_2} \right]
&= \EE_\pi\left[
\left\langle
\Delta_s,
\EE[\Delta_t\mid X_{s-1},Y_{s-1}]
\right\rangle_{\cls_2}
\right]
\label{al:tower}\\
&\le \EE_{\pi} \left[ \| \Delta_s\|_{\cls_2} \| \EE( \Delta_t \mid X_{s-1}, Y_{s-1})\|_{\cls_2}\right]
\label{al:CS1}\\
&\leq (\EE_{\pi} [\| \Delta_s \|_{\cls_2}^2])^{1/2} ( \EE_{\pi} [\| \EE(\Delta_t \mid X_{s-1},Y_{s-1})\|_{\cls_2}^2])^{1/2}
\label{al:CS2}
\\&\le (4\kappa^4)^{1/2} ( \EE_{\pi} [ \| \EE(\Delta_t \mid X_{s-1},Y_{s-1})\|_{\cls_2}^2] )^{1/2}
\label{al:kappabound},
\end{align}
where we used the tower property of conditional expectation in \eqref{al:tower} and the Cauchy--Schwarz inequality in \eqref{al:CS1} and \eqref{al:CS2}. Finally, \eqref{al:kappabound} uses the bound in \eqref{eq:momentbounddelta}. To bound the second factor in \eqref{al:kappabound}, we use the geometric ergodicity of $\{(X_t,Y_t)\}_{t \in \NN_0}$. There exists a $\pi$-integrable function $J: \RR^2 \to [0,\infty)$ such that, with
\begin{equation}
C_{\pi} \doteq \int_{\RR^2} J(x,y) \pi(dx , dy),
\end{equation}
and $h \doteq t-s$, we have
\begin{align}
   & \EE_{\pi}  \left[ \| \EE(\Delta_t \mid X_{s-1},Y_{s-1})\|_{\cls_2}^2 \right]\\
    &=
\int_{\RR^2} \left\|
\int_{\RR^2}
\Phi_{12}(u,v)
\left(P^h((x,y), (du, dv))-\pi(du, dv)\right)
\right\|_{\cls_2}^2 \pi(dx,dy) \label{eq:EpiDeltatS2momentoffdiagonalfirstline}
\\&\leq
4 \kappa^4
\int_{\RR^2}
\|P^h((x,y),\cdot)-\pi(\cdot)\|_{\TV}^2 \pi(dx,dy)\\
&\leq
4 \kappa^4
\int_{\RR^2}
\|P^h((x,y),\cdot)-\pi(\cdot)\|_{\TV} \pi(dx,dy)\\
&\leq
4 \kappa^4 \rho^h
\int_{\RR^2}J(x,y) \pi(dx,dy)\\
&= 4 C_{\pi}\kappa^4  \rho^h.\label{eq:EpiDeltatS2momentoffdiagonal}
\end{align}
To obtain the identity in \eqref{eq:EpiDeltatS2momentoffdiagonalfirstline}, note that
\begin{equation}
\Delta_t
=
\Phi_{12}(X_{t-1},Y_{t-1})-\EE_\pi[\Phi_{12}(X_0,Y_0)].
\end{equation}
The Markov property implies that
\begin{equation}\label{eq:conditional_law_Zt_given_Zs}
\mathcal L((X_{t-1},Y_{t-1})\mid (X_{s-1},Y_{s-1})=(x,y))
=
P^h((x,y),\cdot),
\end{equation}
so
\begin{equation}\label{eq:conditional_mean_Delta_t}
\EE[\Delta_t\mid (X_{s-1},Y_{s-1})=(x,y)]
=
\int_{\RR^2}\Phi_{12}(u,v) ( P^h((x,y), (du, dv))-\pi( du, dv) )
\end{equation}
The Bochner integral against the signed measure
$P^h((x,y),\cdot)-\pi$ is justified because $\Phi_{12}$ is a bounded
strongly measurable $\cls_2$-valued map.
Since $(X_{s-1},Y_{s-1})\sim\pi$, it follows that 
\begin{align}
\EE_{\pi}  \left[ \| \EE(\Delta_t \mid X_{s-1},Y_{s-1})\|_{\cls_2}^2 \right]
&=
\int_{\RR^2} \left\|
\int_{\RR^2}
\Phi_{12}(u,v)
\left(P^h((x,y), (du, dv))-\pi(du, dv)\right)
\right\|_{\cls_2}^2 \pi(dx,dy),
\label{eq:conditional_mean_Delta_t_squared_integral}
\end{align}
which completes the proof of the identity in \eqref{eq:EpiDeltatS2momentoffdiagonalfirstline}.

Using \eqref{al:kappabound} and \eqref{eq:EpiDeltatS2momentoffdiagonal}, we see that, with $\widetilde{\rho} \doteq \rho^{1/2} \in (0,1)$, 
\begin{align}
    \frac{2}{T^2}
\sum_{1\le s<t\le T} \EE_\pi
[\langle \Delta_s,\Delta_t\rangle_{\cls_2}]
&\le  \frac{2}{T^2} ( 4 \kappa^4)^{1/2}  \sum_{1 \le s < t \le T}(4 C_{\pi} \kappa^4 \rho^{|s-t|})^{1/2}\\
&=  \frac{8 \kappa^4 C_{\pi}^{1/2}}{T} \frac{1}{T} \sum_{h=1}^{T-1} (T-h)\widetilde{\rho}^h
\label{eq:boundsumoffdiagonalDeltainnerproduct}
\\&=
\frac{8 \kappa^4 C_{\pi}^{1/2}}{T}  \frac{ \widetilde{\rho }}{1 - \widetilde{\rho}} + o(1),
\label{eq:boundsumoffdiagonalDeltainnerproduct+1}
\end{align}
as $ T \to \infty$, 
Combining \eqref{eq:proofoftheorem42:eq4}, \eqref{eq:EpiDeltatS2momentdiagonal}, \eqref{eq:boundsumoffdiagonalDeltainnerproduct+1}, we obtain
\begin{align}
\EE_\pi [ \|\widehat L_{12}-L_{12}\|_{\cls_2}^2 ]
&\le
\frac{4\kappa^4}{T}
+
\frac{8 \kappa^4 C_{\pi}^{1/2}}{T}
\left( \frac{\widetilde{\rho}}{1-\widetilde{\rho}}+o(1)\right)
=O(T^{-1}).
\end{align}
Jensen's inequality then gives
\begin{equation}
\EE_\pi [\|\widehat L_{12}-L_{12}\|_{\cls_2}]
=
O(T^{-1/2}).
\end{equation}
Since the operator norm is bounded above by the Hilbert--Schmidt norm,
\begin{equation}
\EE_\pi [\|\widehat L_{12}-L_{12}\|_{\op}]
\le
\EE_\pi [\|\widehat L_{12}-L_{12}\|_{\cls_2}]
=
O(T^{-1/2}).
\end{equation}
The arguments for $\widehat L_{11}$, $\widehat L_{21}^c$, and $\widehat L_{22}^c$ are similar.
\end{proof}

\begin{corollary}
\label{cor:aux_operator_conv_prob}
Suppose Assumptions \ref{assump:bounded_dynamics}--\ref{assump:gaussian} hold. For each pair
\begin{equation}
(\widehat L,L)\in
\{(\widehat L_{11},L_{11}), (\widehat L_{12},L_{12}),
(\widehat L_{21}^c,L_{21}^c), (\widehat L_{22}^c,L_{22}^c)\},
\end{equation}
as $T \to \infty$,
\begin{equation}
\|\widehat L-L\|_{\cls_2}
=
O_{\PP}(T^{-1/2}),
\qquad
\|\widehat L-L\|_{\op}
=
O_{\PP}(T^{-1/2}).
\end{equation}
\end{corollary}

\begin{lemma}
\label{lem:subgaussian-truncated-second-moment}
Let $\varepsilon$ be a Gaussian random variable. Then, there exist constants
$C_1,C_2>0$ such that, for all $x\ge 0$,
\begin{equation}
\EE\left[
\varepsilon^2 1_{\{|\varepsilon|\ge x\}}
\right]
\le
C_1 \exp(-C_2 x^2).
\end{equation}
\end{lemma}
\begin{proof}
By Cauchy--Schwarz, for some constants $C_1,C_2>0$, 
\begin{equation}
\EE\left[\varepsilon^2 1_{\{|\varepsilon|\ge x\}}\right]
\le
\left(\EE[\varepsilon^4]\right)^{1/2}
\left(\PP(|\varepsilon|\ge x)\right)^{1/2}
\le
C_1\exp(-C_2x^2),
\end{equation}
since $\varepsilon$ is Gaussian, $\EE[\varepsilon^4]<\infty$ and there exist constants
$C,c>0$ such that
$\PP(|\varepsilon|\ge x)\le C\exp(-cx^2)$, $x\ge0$.
\end{proof}

\section{Proofs of main results in Section \ref{sec:full_sample_empirical_centering}}
\label{app:full_sample_empirical_centering}

This appendix proves Proposition~\ref{prop:full_centering_risk_rate},
Theorem~\ref{thm:full_oracle_equivalence}, and
Propositions~\ref{prop:full_covariance_convergence} and
\ref{prop:full_tuning_choices}. We first establish the operator bound
\eqref{eq:full_app_master_rate} in Lemma~\ref{lem:full_exact_factorization}.
Combined with the empirical covariance comparison in
Lemma~\ref{lem:full_empirical_loewner}, it controls the centering
error through \eqref{eq:full_app_centering_estimation_bound}. Combined with
the rate in \eqref{eq:full_app_W_rate}, it controls the innovation-weighted
centering error through \eqref{eq:full_app_innovation_operator_product}.
These bounds control the three remainder terms in
\eqref{eq:full_app_three_remainders}.

Recall $A_\lambda$ and $\widehat A_{\lambda,T}$ from \eqref{eq:full_population_krr_centering} and \eqref{eq:full_empirical_krr_operator} to define the regularized KRR operator error
\begin{equation}
E_{\lambda,T}
\doteq
\widehat A_{\lambda,T}-A_\lambda.
\label{eq:full_app_E_def}
\end{equation}

\begin{proof}[Proof of Proposition~\ref{prop:full_centering_risk_rate}]

Define the transformed covariance error
\begin{equation}
B_{\lambda,T}
\doteq
(L_{11}+\lambda I_{\clh_1})^{-1/2}
(\widehat L_{11}-L_{11})
(L_{11}+\lambda I_{\clh_1})^{-1/2}.
\label{eq:full_app_B_def}
\end{equation}
We work on the event $\{\|B_{\lambda,T}\|_{\op}\leq1/2\}$ and separately bound the probability of its complement. By Lemma~\ref{lem:B-control}, there exists a constant \(C<\infty\),
independent of \(T\) and \(\lambda\), such that
\begin{equation}
\mathbb{P}\!\left(
    \|B_{\lambda,T}\|_{\op}>\frac12
\right)
\leq
\frac{C}{T\lambda^{2\gamma_\psi}}.
\end{equation}
Moreover, Assumption~\ref{ass:full_empirical_centering_tuning} implies
\(T\lambda^{2\gamma_\psi}\to\infty\). Hence
\(\mathbb{P}(\|B_{\lambda,T}\|_{\op}\leq1/2)\to1\), and it suffices to
establish the desired bound on this event.

Recall $E_{\lambda,T}$ from \eqref{eq:full_app_E_def} and define the map $r_{\lambda} : \RR \to \widetilde{\clh}_2$ by
\begin{align}
    r_\lambda(x) \doteq m_\lambda (x) -m (x), \qquad x \in \RR.
    \label{eq:r_lambda_def_1}
\end{align}
By \eqref{eq:full_empirical_krr_operator}, for every
$t = 1, \dots, T$,
\begin{equation}
\widehat m_T(X_{t-1})-m(X_{t-1})
=
E_{\lambda,T}K_1(X_{t-1},\cdot)+r_\lambda(X_{t-1}).
\label{eq:full_app_centering_error_decomposition}
\end{equation}
We consider the two terms in \eqref{eq:full_app_centering_error_decomposition} separately. 
For the first summand in \eqref{eq:full_app_centering_error_decomposition}, on the event $\{\|B_{\lambda,T}\|_{\op}\leq1/2\}$,
\begin{align}
\frac1T\sum_{t=1}^T
\|E_{\lambda,T}K_1(X_{t-1},\cdot)\|_{\widetilde\clh_2}^2
&=
\operatorname{tr}
\left(E_{\lambda,T}\widehat L_{11}E_{\lambda,T}^*\right)
\label{eq:full_app_centering_estimation_bound-2}\\
&\leq
\frac32
\|E_{\lambda,T}(L_{11}+\lambda I_{\clh_1})^{1/2}\|_{\cls_2}^2
\label{eq:full_app_centering_estimation_bound-1}\\
&=
O_\PP\!\left(T^{-1}\lambda^{-(\gamma_\psi+1)}\right),
\label{eq:full_app_centering_estimation_bound}
\end{align}
where \eqref{eq:full_app_centering_estimation_bound-2} is due to the definition of the $\widetilde\clh_2$-norm and the reproducing property. Since we work on the event $\{\|B_{\lambda,T}\|_{\op}\leq1/2\}$, \eqref{eq:full_app_centering_estimation_bound-1} follows from 
Lemma~\ref{lem:full_empirical_loewner}.
Finally, \eqref{eq:full_app_centering_estimation_bound} follows from
Lemma~\ref{lem:full_exact_factorization}. The final condition in
\eqref{eq:full_tuning_conditions} implies
$T\lambda^{2\gamma_\psi}\to\infty$, since $\lambda<1$ eventually, so the
lemma applies.

For the bias term in \eqref{eq:full_app_centering_error_decomposition}, stationarity and \eqref{eq:full_centering_bias} imply
\begin{equation}
\EE_\pi\!\left[
\frac1T\sum_{t=1}^T\|r_\lambda(X_{t-1})\|_{\widetilde\clh_2}^2
\right]
=b_{\lambda,c}^2.
\label{eq:identity_expectation_r_lambda_b_lambda_c_norm}
\end{equation}
Markov's inequality therefore gives
\begin{equation}
\frac1T\sum_{t=1}^T
\|r_\lambda(X_{t-1})\|_{\widetilde\clh_2}^2
=O_\PP(b_{\lambda,c}^2).
\label{eq:full_app_centering_bias_bound}
\end{equation}
Combining \eqref{eq:full_app_centering_error_decomposition}--\eqref{eq:full_app_centering_bias_bound} with
$\|u+v\|^2\leq2\|u\|^2+2\|v\|^2$ proves the first equality in
\eqref{eq:full_centering_risk_rate}. Finally,
\begin{equation}
Ta_T^2\left(
T^{-1}\lambda^{-(\gamma_\psi+1)}+b_{\lambda,c}^2
\right)
=
a_T^2\lambda^{-(\gamma_\psi+1)}+Ta_T^2b_{\lambda,c}^2
\longrightarrow0
\end{equation}
by Assumption~\ref{ass:full_empirical_centering_tuning}. This proves the
second equality in \eqref{eq:full_centering_risk_rate} and completes the proof.
\end{proof}

\begin{proof}[Proof of Theorem~\ref{thm:full_oracle_equivalence}]
We first show that
\begin{equation}
\widehat Z_T^{\mathrm{res}}
=
Z_T^{0}+o_\PP(1)
\qquad\text{in }\widetilde\clh_2,
\qquad
Z_T^{0}
\doteq
\frac1{\sqrt T}\sum_{t=1}^T\varepsilon_{X,t}k_t^c
\label{eq:full_oracle_equivalence}
\end{equation}
with $k_t^c$ as in \eqref{eq:full_centered_sections}.
Proposition~\ref{prop:aux_FCLT} then gives
\(Z_T^0\xrightarrow{\distr}\calG_c\).

Define
\begin{equation}
\Delta_{t,T}
\doteq
\widehat m_T(X_{t-1})-m(X_{t-1}),
\label{eq:Delta_t_T_definition}
\end{equation}
and recall $k_t^c$ and $\widehat k_{t,T}^c$ from
\eqref{eq:full_centered_sections}.
Under the null, $\widehat e_t=\varepsilon_{X,t}-( \widehat f-f_X )(X_{t-1})$ and
$\widehat k_{t,T}^c=k_t^c-\Delta_{t,T}$. Expanding the product gives the main
reduction
\begin{align}
\widehat Z_T^{\mathrm{res}}-Z_T^{0}
&=
-\frac1{\sqrt T}\sum_{t=1}^T\varepsilon_{X,t}\Delta_{t,T}
-\frac1{\sqrt T}\sum_{t=1}^T( \widehat f-f_X )(X_{t-1})k_t^c
+
\frac1{\sqrt T}\sum_{t=1}^T( \widehat f-f_X )(X_{t-1})\Delta_{t,T}.
\label{eq:full_app_three_remainders}
\end{align}
We consider the three summands separately.

The first summand is $o_\PP(1)$ by
Lemma~\ref{lem:full_innovation_weighted_error}. For the second summand in \eqref{eq:full_app_three_remainders}, the operator
definitions \eqref{eq:full_app_Cc_def}--\eqref{eq:full_app_Cchat_def} and
$\widetilde L_{21}^c=0$ give
\begin{align}
\left\|
\frac1{\sqrt T}\sum_{t=1}^T( \widehat f-f_X )(X_{t-1})k_t^c
\right\|_{\widetilde\clh_2}
&=
\|\sqrt T(\widehat{\widetilde L}_{21,T}^c-\widetilde L_{21}^c)( \widehat f-f_X )\|_{\widetilde\clh_2}
\notag\\
&\leq
\sqrt T\|\widehat{\widetilde L}_{21,T}^c-\widetilde L_{21}^c\|_{\op}
\| \widehat f-f_X \|_{\clh_1}
=o_\PP(1),
\label{eq:full_app_second_remainder}
\end{align}
where Lemma~\ref{lem:full_cross_operator_rate} and
Assumption~\ref{ass:fhattau-f} are used in the final equality.

For the third summand in \eqref{eq:full_app_three_remainders}, the triangle and Cauchy--Schwarz inequalities give
\begin{align}
\left\|
\frac1{\sqrt T}\sum_{t=1}^T
( \widehat f-f_X )(X_{t-1})\Delta_{t,T}
\right\|_{\widetilde\clh_2}
&\leq
\sqrt T\| \widehat f-f_X \|_\infty
(\widehat R_{c,T})^{1/2}
=o_\PP(1),
\label{eq:full_app_third_remainder}
\end{align}
by Assumption~\ref{ass:fhattau-f} and
Proposition~\ref{prop:full_centering_risk_rate}. Combining
\eqref{eq:full_app_three_remainders}--\eqref{eq:full_app_third_remainder}
proves \eqref{eq:full_oracle_equivalence}.

Under the identification induced by
\eqref{eq:full_oracle_centered_kernel}, Proposition~\ref{prop:aux_FCLT} applied to the full-sample construction gives
$Z_T^{0}\xrightarrow{\distr}\calG_c$ with covariance
$\sigma^2L_c$. Slutsky's theorem proves the second convergence in
\eqref{eq:full_embedding_limit}, and the continuous mapping theorem gives
the convergence of the squared norm.

It remains to identify the scalar limit. Since
$\|k_1^c\|_{\widetilde\clh_2}\leq2\kappa_2$,
\begin{equation}
\operatorname{tr}(L_c)
=
\EE_\pi\|k_1^c\|_{\widetilde\clh_2}^2
\leq4\kappa_2^2<\infty.
\end{equation}
Thus $L_c$ is positive, self-adjoint, and trace class. Let
$\{e_j^c\}$ be an orthonormal eigenbasis on the closure of its range, with
$L_ce_j^c=\mu_j^ce_j^c$. The Gaussian coordinates are independent and
satisfy
\begin{equation}
\langle\calG_c,e_j^c\rangle_{\widetilde\clh_2}
\overset{\distr}{=}
\sigma(\mu_j^c)^{1/2}N_j.
\end{equation}
The Gaussian element vanishes on $\ker(L_c)$ almost surely. Parseval's
identity therefore gives
\begin{equation}
\|\calG_c\|_{\widetilde\clh_2}^2
\overset{\distr}{=}
\sigma^2\sum_{j=1}^\infty\mu_j^cN_j^2.
\end{equation}
The series is finite almost surely because it is nonnegative and has finite
expectation $\sigma^2\operatorname{tr}(L_c)$. This proves
\eqref{eq:full_weighted_chisq} and completes the proof.
\end{proof}

\begin{proof}[Proof of Proposition~\ref{prop:full_covariance_convergence}]
Recall $k_t^c$ from \eqref{eq:full_centered_sections} and define the oracle empirical covariance operator
\begin{equation}
L_{c,T}^0 \doteq \frac{1}{T} \sum_{t=1}^{T} k_t^c \otimes k_t^c,
\quad
\text{ so that } 
\quad
L_{c,T}^0h
=
\frac1T\sum_{t=1}^T
\langle h,k_t^c\rangle_{\widetilde\clh_2}k_t^c, \qquad h \in \widetilde{\clh}_2.
\label{eq:full_app_oracle_empirical_covariance}
\end{equation}
An argument similar to the one used in the proof of 
Lemma~\ref{lem:full_cross_operator_rate}, applied to
$k_t^c\otimes k_t^c$, gives
\begin{equation}
\|L_{c,T}^0-L_c\|_{\cls_2}=O_\PP(T^{-1/2})
\label{eq:full_app_oracle_covariance_rate}
\end{equation}
where $L_c$ is defined in \eqref{eq:full_Lc_def}.

Using \eqref{eq:full_centered_sections} and \eqref{eq:Delta_t_T_definition}, note that we can write 
\begin{equation}
    \widehat k_{t,T}^c=k_t^c-\Delta_{t,T}.\label{eq:represent_k_hat_c_t_T_k_t_c_Delta_t_T}
\end{equation}
Observe that
\begin{align}
\|\widehat k_{t,T}^c\otimes\widehat k_{t,T}^c-k_t^c\otimes k_t^c\|_{\cls_2}
&=
\|
-\Delta_{t,T}\otimes\widehat k_{t,T}^c
-
k_t^c\otimes\Delta_{t,T}
\|_{\cls_2}
\\&\leq
\|\Delta_{t,T}\|_{\widetilde\clh_2}
\left(
\|\widehat k_{t,T}^c\|_{\widetilde\clh_2}
+\|k_t^c\|_{\widetilde\clh_2}
\right),
\label{eq:full_app_rank_one_difference}
\end{align}
where we added and subtracted \(k_t^c\otimes\widehat k_{t,T}^c\) and used the triangle inequality and the rank-one identity
\(\|u\otimes v\|_{\cls_2}=\|u\|_{\widetilde\clh_2}
\|v\|_{\widetilde\clh_2}\). By \eqref{eq:full_feasible_covariance} and
\eqref{eq:full_app_oracle_empirical_covariance},
\begin{align}
\|\widehat L_{c,T}-L_{c,T}^0\|_{\cls_2}
&\leq
\frac1T\sum_{t=1}^T
\|\widehat k_{t,T}^c\otimes\widehat k_{t,T}^c
-k_t^c\otimes k_t^c\|_{\cls_2}
\notag\\
&\leq
\frac1T\sum_{t=1}^T
\|\Delta_{t,T}\|_{\widetilde\clh_2}
\left(
\|\widehat k_{t,T}^c\|_{\widetilde\clh_2}
+\|k_t^c\|_{\widetilde\clh_2}
\right)
\notag\\
&\leq
(\widehat R_{c,T})^{1/2}
\left[
\frac1T\sum_{t=1}^T
\left(
\|\widehat k_{t,T}^c\|_{\widetilde\clh_2}
+\|k_t^c\|_{\widetilde\clh_2}
\right)^2
\right]^{1/2},
\label{eq:full_app_feasible_oracle_covariance_bound}
\end{align}
where the second inequality follows from
\eqref{eq:full_app_rank_one_difference}, and the final inequality follows
from the Cauchy--Schwarz inequality, 
\eqref{eq:full_centering_risk}, and \eqref{eq:Delta_t_T_definition}. Using Assumption~\ref{assump:bounded_kernel} and Jensen's inequality, we see that 
\[
\|m(X_{t-1})\|_{\widetilde\clh_2}\leq\kappa_2,
\qquad
\|k_t^c\|_{\widetilde\clh_2}\leq2\kappa_2.
\]
The previous display and  \eqref{eq:represent_k_hat_c_t_T_k_t_c_Delta_t_T} then imply that 
\begin{align}
\frac1T\sum_{t=1}^T
\left(
\|\widehat k_{t,T}^c\|_{\widetilde\clh_2}
+\|k_t^c\|_{\widetilde\clh_2}
\right)^2
&\leq
\frac1T\sum_{t=1}^T
\left(4\kappa_2+\|\Delta_{t,T}\|_{\widetilde\clh_2}\right)^2
\notag\\
&\leq
32\kappa_2^2
+\frac2T\sum_{t=1}^T
\|\Delta_{t,T}\|_{\widetilde\clh_2}^2
\notag\\
&=
32\kappa_2^2+2\widehat R_{c,T}.
\label{eq:full_app_centered_section_second_factor}
\end{align}
Note that  \eqref{eq:full_centering_risk_rate} from Proposition~\ref{prop:full_centering_risk_rate}  and the condition that 
$Ta_T^2\to\infty$ in
Assumption~\ref{ass:full_empirical_centering_tuning} imply that
$\widehat R_{c,T}=o_\PP(1)$. This, together with \eqref{eq:full_app_centered_section_second_factor}, shows that 
\begin{align}
    \frac1T\sum_{t=1}^T
\left(
\|\widehat k_{t,T}^c\|_{\widetilde\clh_2}
+\|k_t^c\|_{\widetilde\clh_2}
\right)^2 = O_{\PP}(1).
\end{align}
The previous display, together with 
\eqref{eq:full_app_feasible_oracle_covariance_bound} and the fact that $\widehat R_{c,T}=o_\PP(1)$ yields
\begin{equation}
\|\widehat L_{c,T}-L_{c,T}^0\|_{\cls_2}
=o_\PP(1).
\label{eq:full_app_feasible_oracle_covariance}
\end{equation}
Combining \eqref{eq:full_app_oracle_covariance_rate} and \eqref{eq:full_app_feasible_oracle_covariance}, we see that 
\begin{equation}
\|\widehat L_{c,T}-L_c\|_{\cls_2}
\leq
\|\widehat L_{c,T}-L_{c,T}^0\|_{\cls_2}
+
\|L_{c,T}^0-L_c\|_{\cls_2}
=o_\PP(1),
\end{equation}
which proves the first statement in
\eqref{eq:full_covariance_convergence}.

We prove the second statement in \eqref{eq:full_covariance_convergence}. First, recall that, for any
$u\in\widetilde\clh_2$, the operator
$u\otimes u$  satisfies
\begin{equation}
\operatorname{tr}(u\otimes u)
=
\|u\|_{\widetilde\clh_2}^2,
\label{eq:full_app_rank_one_trace}
\end{equation}
 Therefore, by \eqref{eq:full_feasible_covariance},
\eqref{eq:full_app_oracle_empirical_covariance}, and linearity of the trace,
\begin{equation}
\operatorname{tr}(\widehat L_{c,T})
=
\frac1T\sum_{t=1}^T
\|\widehat k_{t,T}^c\|_{\widetilde\clh_2}^2,
\qquad
\operatorname{tr}(L_{c,T}^0)
=
\frac1T\sum_{t=1}^T
\|k_t^c\|_{\widetilde\clh_2}^2.
\label{eq:full_app_empirical_trace_identities}
\end{equation}
Using \eqref{eq:represent_k_hat_c_t_T_k_t_c_Delta_t_T} and the inequality $\big|\|u\|^2-\|v\|^2\big|
\leq
\|u-v\|(\|u\|+\|v\|),$
we obtain
\begin{align}
\left|
\operatorname{tr}(\widehat L_{c,T})
-
\operatorname{tr}(L_{c,T}^0)
\right|
&\leq
\frac1T\sum_{t=1}^T
\|\Delta_{t,T}\|_{\widetilde\clh_2}
\left(
\|\widehat k_{t,T}^c\|_{\widetilde\clh_2}
+
\|k_t^c\|_{\widetilde\clh_2}
\right)
\notag\\
&\leq
(\widehat R_{c,T})^{1/2}
\left[
\frac1T\sum_{t=1}^T
\left(
\|\widehat k_{t,T}^c\|_{\widetilde\clh_2}
+
\|k_t^c\|_{\widetilde\clh_2}
\right)^2
\right]^{1/2}
\notag\\
&=
o_\PP(1),
\label{eq:full_app_trace_difference}
\end{align}
where the second inequality uses the Cauchy--Schwarz inequality and
\eqref{eq:full_centering_risk}, and the final identity  follows from
\eqref{eq:full_centering_risk_rate} and
\eqref{eq:full_app_centered_section_second_factor}. Since the map 
\begin{align}
    (x,y) \mapsto  \| \widetilde{K}_2 ( (x,y), \cdot) - m(x)\|_{\widetilde{\clh}_2}, \qquad (x,y) \in \RR^2,
\end{align}
is bounded and  measurable and $\{(X_t,Y_t)\}$ is ergodic, the ergodic theorem for Markov chains ensures that 
\begin{equation}
\frac1T\sum_{t=1}^T
\|k_t^c\|_{\widetilde\clh_2}^2
\rightarrow 
\EE_\pi\|k_1^c\|_{\widetilde\clh_2}^2, \quad \text{a.s.}
\label{eq:full_app_oracle_trace_limit}
\end{equation}
The second identity in \eqref{eq:full_app_empirical_trace_identities}, together with \eqref{eq:full_app_oracle_trace_limit}, then  implies that  
\begin{align}
\operatorname{tr}(L_{c,T}^0) \to \EE_{\pi} \|k_1^c\|_{\widetilde{\clh}_2}^2, \quad \text{a.s.}
\label{eq:trace_L_c_T_0_converge_ergodic}
\end{align}
Finally, \eqref{eq:full_Lc_def}, Tonelli's theorem, and
\eqref{eq:full_app_rank_one_trace} imply that
\begin{equation}
\operatorname{tr}(L_c)
=
\EE_\pi\!\left[
\operatorname{tr}(k_1^c\otimes k_1^c)
\right]
=
\EE_\pi\|k_1^c\|_{\widetilde\clh_2}^2.
\label{eq:full_app_population_trace_identity}
\end{equation}
Combining \eqref{eq:full_app_trace_difference}, \eqref{eq:trace_L_c_T_0_converge_ergodic}, and 
\eqref{eq:full_app_population_trace_identity} we see that 
$\operatorname{tr}(\widehat L_{c,T})
\xrightarrow{\PP}
\operatorname{tr}(L_c),$
which completes the proof of the second statement in \eqref{eq:full_covariance_convergence}.

We prove the final statement of the proposition, namely that the nonzero
eigenvalues of $\widehat L_{c,T}$ are the eigenvalues of
$\widehat K_2^c/T$. Define
$\Phi_T:\RR^T\to\widetilde\clh_2$ and its adjoint $\Phi_T^*:\widetilde\clh_2\to\RR^T$ by
\[
\Phi_Ta=\sum_{t=1}^T a_t\widehat k_{t,T}^c,
\quad a\in\RR^T,
\quad \text{ and }
\quad
\Phi_T^*h
=
\left(
\langle \widehat k_{1,T}^c,h\rangle_{\widetilde\clh_2},
\ldots,
\langle \widehat k_{T,T}^c,h\rangle_{\widetilde\clh_2}
\right)^\top,
\quad h\in\widetilde\clh_2.
\]
Then,
$\widehat K_2^c=\Phi_T^*\Phi_T$ and
$\widehat L_{c,T}=T^{-1}\Phi_T\Phi_T^*$. Since the nonzero
eigenvalues of $\Phi_T^*\Phi_T/T$ and $\Phi_T\Phi_T^*/T$ coincide, the
result follows.
\end{proof}

\begin{proof}[Proof of Proposition~\ref{prop:full_tuning_choices}]
We verify the three product conditions in
Assumption~\ref{ass:full_empirical_centering_tuning}.

For logarithmic tuning, let $\lambda_{T,K}=(\log T)^{-r}$. Then
$Ta_T^2=(\log T)^{2\zeta_f}\to\infty$, and
\begin{align}
Ta_T^2b_{\lambda_{T,K},c}^2
&=
O\!\left((\log T)^{2\zeta_f-r\beta_c}\right)
\longrightarrow0,\\
a_T^2\lambda_{T,K}^{-(\gamma_\psi+1)}
&=
\frac{(\log T)^{2\zeta_f+r(\gamma_\psi+1)}}{T}
\longrightarrow0,\\
T\lambda_{T,K}^{2\gamma_\psi+1}
&=
\frac{T}{(\log T)^{r(2\gamma_\psi+1)}}
\longrightarrow\infty.
\end{align}
The first convergence uses $r\beta_c>2\zeta_f$.

For polynomial tuning, let $\lambda_{T,K}=T^{-q}$ with
$0<q<(2\gamma_\psi+1)^{-1}$. Then
\begin{align}
Ta_T^2b_{\lambda_{T,K},c}^2
&=
O\!\left((\log T)^{2\zeta_f}T^{-q\beta_c}\right)
\longrightarrow0,\\
a_T^2\lambda_{T,K}^{-(\gamma_\psi+1)}
&=
(\log T)^{2\zeta_f}
T^{-1+q(\gamma_\psi+1)}
\longrightarrow0,\\
T\lambda_{T,K}^{2\gamma_\psi+1}
&=
T^{1-q(2\gamma_\psi+1)}
\longrightarrow\infty.
\end{align}
The imposed range for $q$ implies
$q<(\gamma_\psi+1)^{-1}$, so the exponent in the second line is negative.
All conditions in \eqref{eq:full_tuning_conditions} follow.
\end{proof}

\section{Auxiliary empirical centering results and their proofs}

For $(x,y)\in\RR^2$ and $h\in\clh_1$, define the rank-one operators
\begin{align}
[\Xi_\lambda(x,y)]h
&\doteq
\langle h,\psi_{1,\lambda}(x)\rangle_{\clh_1}
\bigl(\widetilde K_2((x,y),\cdot)-m_\lambda(x)\bigr),
\label{eq:full_app_Xi_def}\\
P_\lambda(x)h
&\doteq
\langle h,\psi_{1,\lambda}(x)\rangle_{\clh_1}
\psi_{1,\lambda}(x).
\label{eq:full_app_P_def}
\end{align}
Let
\begin{align}
D_{1,\lambda,T}
&\doteq
\left\|
\frac1T\sum_{t=1}^T
\left(
\Xi_\lambda(X_{t-1},Y_{t-1})
-\EE_\pi\Xi_\lambda(X_0,Y_0)
\right)
\right\|_{\cls_2},
\label{eq:full_app_D1_def}\\
D_{0,\lambda,T}
&\doteq
\left\|
\frac1T\sum_{t=1}^T
\left(
P_\lambda(X_{t-1})-\EE_\pi P_\lambda(X_0)
\right)
\right\|_{\op}.
\label{eq:full_app_D0_def}
\end{align}

\begin{lemma}
\label{lem:full_fluctuation_rates}
Under Assumptions~\ref{assump:bounded_dynamics}--\ref{assump:gaussian} and
the embedding bound \eqref{eq:full_embedding_bound},
\begin{align}
D_{1,\lambda,T}
&=
O_\PP\!\left(
T^{-1/2}\lambda^{-(\gamma_\psi+1)/2}
\right),
\label{eq:full_app_D1_rate}\\
D_{0,\lambda,T}
&=
O_\PP\!\left(
T^{-1/2}\lambda^{-\gamma_\psi}
\right).
\label{eq:full_app_D0_rate}
\end{align}
\end{lemma}

\begin{proof}
We first bound the two rank-one operators. Recall
that $m_\lambda$ is the minimizer of the population vector-valued KRR
criterion
\begin{equation}
g\longmapsto
\EE_\pi\!\left[
\|\widetilde K_2((X_0,Y_0),\cdot)-g(X_0)\|_{\widetilde\clh_2}^2
\right]
+\lambda\|g\|_{\calG}^2.
\label{eq:full_app_population_krr_objective}
\end{equation}
Evaluating \eqref{eq:full_app_population_krr_objective} at the zero function
and using Assumption~\ref{assump:bounded_kernel} gives
\begin{equation}
\lambda\|m_\lambda\|_{\calG}^2
\leq
\EE_\pi\|\widetilde K_2((X_0,Y_0),\cdot)\|_{\widetilde\clh_2}^2
\leq\kappa_2^2.
\end{equation}
The reproducing property for the operator-valued kernel in
\eqref{eq:full_operator_kernel} therefore yields
\begin{equation}
\|m_\lambda(x)\|_{\widetilde\clh_2}
\leq
\kappa\|m_\lambda\|_{\calG}
\leq
\kappa\kappa_2\lambda^{-1/2}.
\label{eq:full_app_mlambda_envelope}
\end{equation}
Since $\lambda<1$, it follows that
\begin{equation}
\|\widetilde K_2((x,y),\cdot)-m_\lambda(x)\|_{\widetilde\clh_2}
\leq
(1+\kappa)\kappa_2\lambda^{-1/2}
\doteq Q_\lambda.
\label{eq:full_app_Qlambda}
\end{equation}

The Hilbert--Schmidt norm of a rank-one operator is the product of the norms
of its two factors. Combining \eqref{eq:full_embedding_bound} and
\eqref{eq:full_app_Qlambda} gives
\begin{align}
\|\Xi_\lambda(x,y)\|_{\cls_2}
&\leq
C_\psi Q_\lambda\lambda^{-\gamma_\psi/2}
\leq
C_1\lambda^{-(\gamma_\psi+1)/2}
\doteq B_\lambda,
\label{eq:full_app_Xi_envelope}\\
\|P_\lambda(x)\|_{\cls_2}
&=
\|\psi_{1,\lambda}(x)\|_{\clh_1}^2
\leq
C_\psi^2\lambda^{-\gamma_\psi}
\doteq H_\lambda.
\label{eq:full_app_P_envelope}
\end{align}

The embedding bound is required only $\pi_X$-almost everywhere. Choose a
common $\pi_X$-full set on which it holds for the countable collection of
tuning values $\{\lambda_{T,K}:T\in\NN\}$, and set the two kernels to zero
outside that set. Stationarity and a finite union bound show that this
modification does not change any empirical operator almost surely.

We then prove \eqref{eq:full_app_D1_rate}. Define
\begin{equation}
\zeta_t
\doteq
\Xi_\lambda(X_{t-1},Y_{t-1})
-\EE_\pi\Xi_\lambda(X_0,Y_0),
\end{equation}
and let
\begin{equation}
C_\pi
\doteq
\EE_\pi[J(X_0,Y_0)]<\infty.
\end{equation}
The proof is based on the following second-moment calculation
\begin{align}
\EE_\pi D_{1,\lambda,T}^2
&=
\EE_\pi\left\|
\frac1T\sum_{t=1}^T\zeta_t
\right\|_{\cls_2}^2
\notag\\
&=
\frac1{T^2}\sum_{t=1}^T
\EE_\pi\|\zeta_t\|_{\cls_2}^2
+
\frac{2}{T^2}
\sum_{h=1}^{T-1}\sum_{t=1}^{T-h}
\EE_\pi
\langle\zeta_t,\zeta_{t+h}\rangle_{\cls_2}
\notag\\
&\leq
\frac{4B_\lambda^2}{T}
+
\frac{16B_\lambda^2C_\pi}{T^2}
\sum_{h=1}^{T-1}(T-h)\rho^h
\label{eq:full_app_Xi_second_moment-1}\\
&\leq
\frac{B_\lambda^2}{T}
\left(
4+\frac{16C_\pi\rho}{1-\rho}
\right)
\leq
\frac{C_2B_\lambda^2}{T},
\label{eq:full_app_Xi_second_moment}
\end{align}
where the \eqref{eq:full_app_Xi_second_moment-1} uses
$\|\zeta_t\|_{\cls_2}\leq2B_\lambda$, which follows from
\eqref{eq:full_app_Xi_envelope}, Jensen's inequality, and the triangle
inequality, together with the covariance bound
\begin{equation}
\left|
\EE_\pi\langle\zeta_t,\zeta_{t+h}\rangle_{\cls_2}
\right|
\leq
8B_\lambda^2C_\pi\rho^h.
\label{eq:full_app_Xi_covariance_bound}
\end{equation}
The second inequality \eqref{eq:full_app_Xi_second_moment}, uses the
geometric-series bound
$\sum_{h=1}^{T-1}(T-h)\rho^h
\leq T\rho/(1-\rho)$.
It remains to establish \eqref{eq:full_app_Xi_covariance_bound}. We first
reduce the temporal covariance to a conditional-mean bound:
\begin{align}
\left|
\EE_\pi
\langle\zeta_t,\zeta_{t+h}\rangle_{\cls_2}
\right|
&=
\left|
\EE_\pi
\left\langle
\zeta_t,
\EE_\pi[
\zeta_{t+h}\mid X_{t-1},Y_{t-1}]
\right\rangle_{\cls_2}
\right|
\notag\\
&\leq
\EE_\pi\!\left[
\|\zeta_t\|_{\cls_2}
\left\|
\EE_\pi[
\zeta_{t+h}\mid X_{t-1},Y_{t-1}]
\right\|_{\cls_2}
\right].
\label{eq:full_app_Xi_covariance_reduction}
\end{align}
The equality follows from the tower property and the measurability of
$\zeta_t$ with respect to $\sigma(X_{t-1},Y_{t-1})$, while the inequality
is the Cauchy--Schwarz inequality in $\cls_2$. It therefore remains to
control the conditional mean in
\eqref{eq:full_app_Xi_covariance_reduction}.

For any $R$ in the unit ball of
$\cls_2(\clh_1,\widetilde\clh_2)$, define
\begin{equation}
q_R(x,y)
\doteq
\left\langle
\Xi_\lambda(x,y)-\EE_\pi\Xi_\lambda(X_0,Y_0),
R
\right\rangle_{\cls_2}.
\end{equation}
Then
$\EE_\pi q_R(X_0,Y_0)=0$ and
$\|q_R\|_\infty\leq2B_\lambda$.

Fix such an $R$. Conditional on
$(X_{t-1},Y_{t-1})=(x,y)$, linearity of conditional expectation and the
Markov property give
\begin{align}
&
\left|
\left\langle
\EE_\pi\!\left[
\zeta_{t+h}\mid (X_{t-1},Y_{t-1})=(x,y)
\right],
R
\right\rangle_{\cls_2}
\right|
\notag\\
&=
\left|
\EE_\pi\!\left[
\left\langle\zeta_{t+h},R\right\rangle_{\cls_2}
\mid (X_{t-1},Y_{t-1})=(x,y)
\right]
\right|
\notag\\
&=
\left|
\EE_\pi\!\left[
q_R(X_{t+h-1},Y_{t+h-1})
\mid (X_{t-1},Y_{t-1})=(x,y)
\right]
\right|
\notag\\
&=
\left|
\int q_R(z)\,P^h((x,y),dz)
\right|
\notag\\
&=
\left|
\int q_R(z)
\bigl(P^h((x,y),dz)-\pi(dz)\bigr)
\right|
\notag\\
&\leq
2\|q_R\|_\infty
\|P^h((x,y),\cdot)-\pi\|_{\TV}
\notag\\
&\leq
4B_\lambda\rho^hJ(x,y).
\label{eq:full_app_Xi_conditional_bound}
\end{align}
The third equality uses the Markov property, the fourth uses
$\int q_R\,d\pi=0$, the first inequality is the standard
total-variation bound, and the final inequality follows from
\eqref{eq:geom_ergodicity} together with
$\|q_R\|_\infty\leq2B_\lambda$.

Since \eqref{eq:full_app_Xi_conditional_bound} holds for every
$\|R\|_{\cls_2}\leq1$, the Hilbert-space identity
\begin{equation}
\|A\|_{\cls_2}
=
\sup_{\|R\|_{\cls_2}\leq1}
|\langle A,R\rangle_{\cls_2}|
\end{equation}
yields
\begin{equation}
\left\|
\EE_\pi(
\zeta_{t+h}\mid X_{t-1},Y_{t-1})
\right\|_{\cls_2}
\leq
4B_\lambda\rho^hJ(X_{t-1},Y_{t-1}).
\label{eq:full_app_Xi_conditional_norm_bound}
\end{equation}

Combining \eqref{eq:full_app_Xi_covariance_reduction} with
\eqref{eq:full_app_Xi_conditional_norm_bound} and
$\|\zeta_t\|_{\cls_2}\leq2B_\lambda$ gives
\begin{align}
\left|
\EE_\pi
\langle\zeta_t,\zeta_{t+h}\rangle_{\cls_2}
\right|
&\leq
8B_\lambda^2\rho^h
\EE_\pi[J(X_{t-1},Y_{t-1})]
=
8B_\lambda^2C_\pi\rho^h,
\end{align}
which proves \eqref{eq:full_app_Xi_covariance_bound}.

Since
$B_\lambda=C_1\lambda^{-(\gamma_\psi+1)/2}$, \eqref{eq:full_app_Xi_second_moment} together with 
Markov's inequality yields
\begin{equation}
D_{1,\lambda,T}
=
O_\PP\!\left(
T^{-1/2}\lambda^{-(\gamma_\psi+1)/2}
\right),
\end{equation}
which proves \eqref{eq:full_app_D1_rate}.

For \eqref{eq:full_app_D0_rate}, the same argument applies with
$P_\lambda$ replacing $\Xi_\lambda$. By
\eqref{eq:full_app_P_envelope}, the corresponding centered process has
envelope $2H_\lambda$, and the same total-variation covariance argument
gives
\begin{equation}
\EE_\pi\left\|
\frac1T\sum_{t=1}^T
\bigl(
P_\lambda(X_{t-1})-\EE_\pi P_\lambda(X_0)
\bigr)
\right\|_{\cls_2}^2
\leq
\frac{C_3H_\lambda^2}{T}.
\label{eq:expected_norm_bound_op_D_0_lambda_T}
\end{equation}
Since the operator norm is bounded by the Hilbert--Schmidt norm and
$H_\lambda=C_\psi^2\lambda^{-\gamma_\psi}$, Markov's inequality yields
$D_{0,\lambda,T}
=
O_\PP\!\left(
T^{-1/2}\lambda^{-\gamma_\psi}
\right)$,
which proves \eqref{eq:full_app_D0_rate}.
\end{proof}

\begin{lemma}
\label{lem:full_exact_factorization}
Recall the transformed covariance error $B_{\lambda,T}$ from \eqref{eq:full_app_B_def}.
Then,
\begin{equation}
\|B_{\lambda,T}\|_{\op}=D_{0,\lambda,T}.
\label{eq:full_app_B_D0_identity}
\end{equation}
Under the assumptions of Lemma~\ref{lem:full_fluctuation_rates}, on the event
$\{\|B_{\lambda,T}\|_{\op}\leq1/2\}$,
\begin{equation}
\|E_{\lambda,T}(L_{11}+\lambda I_{\clh_1})^{1/2}\|_{\cls_2}
=
O_\PP\!\left(
T^{-1/2}\lambda^{-(\gamma_\psi+1)/2}
\right).
\label{eq:full_app_master_rate}
\end{equation}
\end{lemma}

\begin{proof}
We first prove \eqref{eq:full_app_B_D0_identity}. For $h\in\clh_1$, recall $\psi_{1,\lambda}(x)$ from \eqref{eq:full_regularized_feature} and $P_\lambda(x)$ from \eqref{eq:full_app_P_def}.
Using the representation of $\widehat L_{11}$ as an empirical covariance
operator and the reproducing property,
\begin{align}
&(L_{11}+\lambda I_{\clh_1})^{-1/2}
\widehat L_{11}
(L_{11}+\lambda I_{\clh_1})^{-1/2}h
\notag\\
&=
\frac1T\sum_{t=1}^T
(L_{11}+\lambda I_{\clh_1})^{-1/2}
K_1(X_{t-1},\cdot)
\left\langle
(L_{11}+\lambda I_{\clh_1})^{-1/2}h,
K_1(X_{t-1},\cdot)
\right\rangle_{\clh_1}
\notag\\
&=
\frac1T\sum_{t=1}^T
\psi_{1,\lambda}(X_{t-1})
\left\langle
h,\psi_{1,\lambda}(X_{t-1})
\right\rangle_{\clh_1}
\notag\\
&=
\frac1T\sum_{t=1}^T
P_\lambda(X_{t-1})h.
\label{eq:B-first}
\end{align}
Here the second equality uses the self-adjointness of
$(L_{11}+\lambda I_{\clh_1})^{-1/2}$.
Similarly, at the population level, for $h \in \clh_1$, 
\begin{equation}
(L_{11}+\lambda I_{\clh_1})^{-1/2}L_{11}
(L_{11}+\lambda I_{\clh_1})^{-1/2}h
=
\EE_\pi P_\lambda(X_0)h.
\label{eq:B-second}
\end{equation}
Subtracting the operator identity \eqref{eq:B-second} from \eqref{eq:B-first} and using
\eqref{eq:full_app_B_def} yields
\begin{equation}
B_{\lambda,T}
=
\frac1T\sum_{t=1}^T
\left(
P_\lambda(X_{t-1})-\EE_\pi P_\lambda(X_0)
\right).
\end{equation}
Taking operator norms and recalling the definition of
\(D_{0,\lambda,T}\) in \eqref{eq:full_app_D0_def} proves
\eqref{eq:full_app_B_D0_identity}.

We next prove \eqref{eq:full_app_master_rate}. To this end, we derive a factorization of the KRR operator error that separates the empirical fluctuation in the KRR normal equation from the effect of empirical covariance inversion. Define
\begin{align}
R_{1,\lambda,T}
&\doteq
\big(\widehat{\widetilde L}_{21,T}-A_\lambda(\widehat L_{11}+\lambda I_{\clh_1})\big)
(L_{11}+\lambda I_{\clh_1})^{-1/2},
\label{eq:full_app_R1_def}\\
R_{2,\lambda,T}
&\doteq
(L_{11}+\lambda I_{\clh_1})^{1/2}
(\widehat L_{11}+\lambda I_{\clh_1})^{-1}
(L_{11}+\lambda I_{\clh_1})^{1/2}.
\label{eq:full_app_R2_def}
\end{align}
Together, \eqref{eq:full_app_R1_def} and
\eqref{eq:full_app_R2_def} yield the factorization
\begin{equation}
E_{\lambda,T}(L_{11}+\lambda I_{\clh_1})^{1/2}
=
R_{1,\lambda,T}R_{2,\lambda,T}.
\label{eq:full_app_exact_factorization}
\end{equation}
We control the two factors separately.

The definition $A_\lambda=\widetilde L_{21}(L_{11}+\lambda I_{\clh_1})^{-1}$ implies
\begin{equation}
\widetilde L_{21}-A_\lambda(L_{11}+\lambda I_{\clh_1})=0.
\label{eq:full_app_population_normal_equation}
\end{equation}
Using \eqref{eq:full_app_population_normal_equation}, we first rewrite
$R_{1,\lambda,T}$ as the difference between an empirical operator and its
population counterpart:
\begin{equation}
R_{1,\lambda,T}
=
\big(\widehat{\widetilde L}_{21,T}-A_\lambda\widehat L_{11}\big)
(L_{11}+\lambda I_{\clh_1})^{-1/2}
-
\big(\widetilde L_{21}-A_\lambda L_{11}\big)
(L_{11}+\lambda I_{\clh_1})^{-1/2}.
\label{eq:full_app_R1_difference}
\end{equation}

We now identify the two operators in
\eqref{eq:full_app_R1_difference}. For $h\in\clh_1$,
\begin{align}
&
\big(\widehat{\widetilde L}_{21,T}-A_\lambda\widehat L_{11}\big)
(L_{11}+\lambda I_{\clh_1})^{-1/2}h
\notag\\
&=
\frac1T\sum_{t=1}^T
\left\langle
h,\psi_{1,\lambda}(X_{t-1})
\right\rangle_{\clh_1}
\widetilde K_2((X_{t-1},Y_{t-1}),\cdot)
-
\frac1T\sum_{t=1}^T
\left\langle
h,\psi_{1,\lambda}(X_{t-1})
\right\rangle_{\clh_1}
m_\lambda(X_{t-1})
\notag\\
&=
\frac1T\sum_{t=1}^T
\Xi_\lambda(X_{t-1},Y_{t-1})h.
\label{eq:full_app_R1_empirical_part}
\end{align}
The first equality uses the empirical representations of
$\widehat{\widetilde L}_{21,T}$ and $\widehat L_{11}$, together with the
reproducing property and the self-adjointness of
$(L_{11}+\lambda I_{\clh_1})^{-1/2}$. In particular,
\eqref{eq:full_regularized_feature} gives the common factor
$\langle h,\psi_{1,\lambda}(X_{t-1})\rangle_{\clh_1}$, while
$m_\lambda(x)=A_\lambda K_1(x,\cdot)$ gives the second term.
The final equality follows from the definition of $\Xi_\lambda$ in
\eqref{eq:full_app_Xi_def}.
Similarly, at the population level,
\begin{equation}
\big(\widetilde L_{21}-A_\lambda L_{11}\big)
(L_{11}+\lambda I_{\clh_1})^{-1/2}h
=
\EE_\pi\!\left[
\Xi_\lambda(X_0,Y_0)h
\right].
\label{eq:full_app_R1_population_part}
\end{equation}
Since these identities hold for every $h\in\clh_1$,
\eqref{eq:full_app_R1_difference} yields
\begin{equation}
R_{1,\lambda,T}
=
\frac1T\sum_{t=1}^T
\Xi_\lambda(X_{t-1},Y_{t-1})
-\EE_\pi\Xi_\lambda(X_0,Y_0).
\label{eq:full_app_R1_empirical_identity}
\end{equation}
It follows that 
\begin{equation}\|R_{1,\lambda,T}\|_{\cls_2}=D_{1,\lambda,T}.
\label{eq:app_R1_empirical_D_1_lambda_T_identity}
\end{equation}

From \eqref{eq:full_app_B_def},
\begin{equation}
\widehat L_{11}+\lambda I_{\clh_1}
=
(L_{11}+\lambda I_{\clh_1})^{1/2}
(I_{\clh_1}+B_{\lambda,T})
(L_{11}+\lambda I_{\clh_1})^{1/2}.
\label{eq:full_app_covariance_factorization}
\end{equation}
Thus
$R_{2,\lambda,T}=(I_{\clh_1}+B_{\lambda,T})^{-1}$. On the stated event, it follows from the  Neumann series representation of $R_{2,\lambda, T}$  that 
\begin{equation}
\|R_{2,\lambda,T}\|_{\op}
\leq
(1-\|B_{\lambda,T}\|_{\op})^{-1}
\leq2.
\label{eq:full_app_R2_bound}
\end{equation}
Combining \eqref{eq:full_app_exact_factorization},
\eqref{eq:app_R1_empirical_D_1_lambda_T_identity}, and
\eqref{eq:full_app_R2_bound}, we obtain
\begin{equation}
\|E_{\lambda,T}(L_{11}+\lambda I_{\clh_1})^{1/2}\|_{\cls_2}
\leq
\|R_{1,\lambda,T}\|_{\cls_2}
\|R_{2,\lambda,T}\|_{\op}
\leq
2D_{1,\lambda,T}.
\label{eq:full_app_master_pathwise}
\end{equation}
The rate \eqref{eq:full_app_D1_rate} from
Lemma~\ref{lem:full_fluctuation_rates} then gives
\begin{equation}
\|E_{\lambda,T}(L_{11}+\lambda I_{\clh_1})^{1/2}\|_{\cls_2}
=
O_\PP\!\left(
T^{-1/2}\lambda^{-(\gamma_\psi+1)/2}
\right),
\end{equation}
which proves \eqref{eq:full_app_master_rate}.
\end{proof}

\begin{lemma}
\label{lem:full_empirical_loewner}
On the event $\{\|B_{\lambda,T}\|_{\op}\leq1/2\}$,
\begin{equation}
\widehat L_{11}
\preceq
\frac32(L_{11}+\lambda I_{\clh_1}).
\label{eq:full_app_loewner}
\end{equation}
\end{lemma}

\begin{proof}
Throughout this proof we work on the event $\{\|B_{\lambda,T}\|_{\op}\leq1/2\}$. For $h\in\clh_1$, \eqref{eq:full_app_B_def}  and the self-adjointness of $L_{11} + \lambda I_{\clh_1}$ together imply that 
\begin{align}
\langle h,(\widehat L_{11}-L_{11})h\rangle_{\clh_1}
&=
\left\langle
(L_{11}+\lambda I_{\clh_1})^{1/2}h,
B_{\lambda,T}(L_{11}+\lambda I_{\clh_1})^{1/2}h
\right\rangle_{\clh_1}
\leq
\frac12\langle h,(L_{11}+\lambda I_{\clh_1})h\rangle_{\clh_1}.
\end{align}
Together, the previous display and the fact that
$L_{11}\preceq L_{11}+\lambda I_{\clh_1}$ imply that 
\begin{align}
    \langle h,\widehat L_{11} h\rangle_{\clh_1} \le \frac{1}{2} \langle h, (L_{11} + \lambda I_{\clh_1})h\rangle_{\clh_1} +\langle h,L_{11}h\rangle_{\clh_1} \le \frac{3}{2} \langle h , (L_{11} + \lambda I_{\clh_1}) h \rangle_{\clh_1},
\end{align}
 which completes the proof of  \eqref{eq:full_app_loewner}.
\end{proof}

\begin{lemma}
\label{lem:full_innovation_weighted_error}
Under the assumptions of Theorem~\ref{thm:full_oracle_equivalence},
\begin{equation}
\left\|
\frac1{\sqrt T}\sum_{t=1}^T
\varepsilon_{X,t}
\big( \widehat m_T(X_{t-1})-m(X_{t-1})\big)
\right\|_{\widetilde\clh_2}
=o_\PP(1).
\label{eq:full_app_innovation_target}
\end{equation}
\end{lemma}

\begin{proof}
Recall the regularized features $\psi_{1,\lambda}: \RR \to \clh_1$ defined in \eqref{eq:full_regularized_feature} and define the regularized martingale term
\begin{equation}
W_{\lambda,T}
\doteq
(L_{11}+\lambda I_{\clh_1})^{-1/2}
\frac1{\sqrt T}\sum_{t=1}^T
\varepsilon_{X,t}K_1(X_{t-1},\cdot)
=
\frac1{\sqrt T}\sum_{t=1}^T
\varepsilon_{X,t}\psi_{1,\lambda}(X_{t-1}).
\label{eq:full_app_W_def}
\end{equation}
For $s<t$, $\varepsilon_{X,s}$, $\psi_{1,\lambda}(X_{s-1})$, and
$\psi_{1,\lambda}(X_{t-1})$ are $\calF_{t-1}$-measurable, so 
\begin{align}
&\EE\!\left[
\left\langle
\varepsilon_{X,s}\psi_{1,\lambda}(X_{s-1}),
\varepsilon_{X,t}\psi_{1,\lambda}(X_{t-1})
\right\rangle_{\clh_1}
\right]
\notag\\
&\quad=
\EE\!\left[
\left\langle
\varepsilon_{X,s}\psi_{1,\lambda}(X_{s-1}),
\psi_{1,\lambda}(X_{t-1})
\EE(\varepsilon_{X,t}\mid\calF_{t-1})
\right\rangle_{\clh_1}
\right]
=0.
\label{eq:full_app_W_cross_zero}
\end{align}
Using \eqref{eq:full_app_W_cross_zero}, the fact that 
$\EE(\varepsilon_{X,t}^2\mid\calF_{t-1})=\sigma^2$,  and
\eqref{eq:full_embedding_bound}, we obtain
\begin{equation}
\EE\|W_{\lambda,T}\|_{\clh_1}^2
=
\frac{\sigma^2}{T}\sum_{t=1}^T
\EE\|\psi_{1,\lambda}(X_{t-1})\|_{\clh_1}^2
\leq
\sigma^2C_\psi^2\lambda^{-\gamma_\psi}.
\label{eq:W_lambda_T_clh_1_expected_variance_bound}
\end{equation}
Together, \eqref{eq:W_lambda_T_clh_1_expected_variance_bound} and Markov's inequality imply that 
\begin{equation}
\|W_{\lambda,T}\|_{\clh_1}
=O_\PP(\lambda^{-\gamma_\psi/2}).
\label{eq:full_app_W_rate}
\end{equation}

Using \eqref{eq:full_app_centering_error_decomposition} and \eqref{eq:full_app_W_def}, observe that 
\begin{align}
&\frac1{\sqrt T}\sum_{t=1}^T
\varepsilon_{X,t}
\big(\widehat m_T(X_{t-1})-m(X_{t-1})\big)
=
E_{\lambda,T}(L_{11}+\lambda I_{\clh_1})^{1/2}W_{\lambda,T}
+
\frac1{\sqrt T}\sum_{t=1}^T
\varepsilon_{X,t}r_\lambda(X_{t-1}).
\label{eq:full_app_innovation_factorization}
\end{align}
Recalling \eqref{eq:full_app_master_rate} and \eqref{eq:full_app_W_rate}, we can see that the first
summand on the right side of \eqref{eq:full_app_innovation_factorization} satisfies
\begin{align}
\|E_{\lambda,T}(L_{11}+\lambda I_{\clh_1})^{1/2}W_{\lambda,T}\|_{\widetilde\clh_2}
&\leq
\|E_{\lambda,T}(L_{11}+\lambda I_{\clh_1})^{1/2}\|_{\cls_2}
\|W_{\lambda,T}\|_{\clh_1}
\notag\\
&=
O_\PP\!\left(
T^{-1/2}\lambda^{-(2\gamma_\psi+1)/2}
\right)
=o_\PP(1),
\label{eq:full_app_innovation_operator_product}
\end{align}
where the final identity uses that 
$T\lambda^{2\gamma_\psi+1}\to\infty$. For the second summand on the right side of \eqref{eq:full_app_innovation_factorization}, we first note that 
$r_\lambda(X_{t-1})$ is $\calF_{t-1}$-measurable, so the same conditioning argument used to establish \eqref{eq:full_app_W_cross_zero} can be used to show that
\begin{align}
\EE[ 
\left\langle
\varepsilon_{X,s} r_{\lambda} (X_{s-1}),
\varepsilon_{X,t} r_{\lambda} (X_{t-1})
\right\rangle_{\widetilde{\clh}_2} 
] = 0.
\label{eq:expectation_H2_norm_squared_r_lambda_bias-2}
\end{align}
Together, \eqref{eq:identity_expectation_r_lambda_b_lambda_c_norm}, \eqref{eq:expectation_H2_norm_squared_r_lambda_bias-2}, and  the fact that 
$\EE(\varepsilon_{X,t}^2\mid\calF_{t-1})=\sigma^2$ imply that 
\begin{equation}
\EE\left\|
\frac1{\sqrt T}\sum_{t=1}^T
\varepsilon_{X,t}r_\lambda(X_{t-1})
\right\|_{\widetilde\clh_2}^2
=
\sigma^2b_{\lambda,c}^2 = o(1),
\label{eq:expectation_H2_norm_squared_r_lambda_bias}
\end{equation}
where the final identity in the previous display follows from Assumption~\ref{ass:full_empirical_centering_tuning}. Using Markov's inequality and \eqref{eq:expectation_H2_norm_squared_r_lambda_bias}, we see that 
\begin{align}
\left\|\frac1{\sqrt T}\sum_{t=1}^T
\varepsilon_{X,t}r_\lambda(X_{t-1})\right\|_{\widetilde\clh_2} = o_{\PP}(1). 
\label{eq:H2_norm_squared_r_lambda_bias_stochastic_order}
\end{align}
The conclusion in \eqref{eq:full_app_innovation_target} follows  from \eqref{eq:full_app_innovation_factorization}, \eqref{eq:full_app_innovation_operator_product}, and \eqref{eq:H2_norm_squared_r_lambda_bias_stochastic_order}.
\end{proof}

Recall the population-centered kernel sections $k_t^c$ defined in
\eqref{eq:full_centered_sections} and define the population-centered
cross-operator
\begin{equation}
\widetilde L_{21}^{c}:\mathcal H_1\to\widetilde{\mathcal H}_2,
\qquad
\widetilde L_{21}^{c}h
\doteq
\mathbb E_{\pi}\!\left[k_1^c h(X_0)\right],
\qquad h\in\mathcal H_1,
\label{eq:full_app_Cc_def}
\end{equation}
and its empirical counterpart
\begin{equation}
\widehat{\widetilde L}_{21,T}^{c}:
\mathcal H_1\to\widetilde{\mathcal H}_2,
\qquad
\widehat{\widetilde L}_{21,T}^{c}h
\doteq
\frac{1}{T}\sum_{t=1}^T
k_t^c h(X_{t-1}),
\qquad h\in\mathcal H_1.
\label{eq:full_app_Cchat_def}
\end{equation}
Note that the feature centering in \eqref{eq:full_app_Cchat_def} is done at the population level.

\begin{lemma}
\label{lem:full_cross_operator_rate}
Under Assumptions~\ref{assump:bounded_dynamics}-\ref{assump:gaussian}, we have $\widetilde L_{21}^{c}=0$ and
\begin{equation}
\sqrt T\,
\left\|
\widehat{\widetilde L}_{21,T}^{c}
-
\widetilde L_{21}^{c}
\right\|_{\op}
=
O_{\mathbb P}(1).
\label{eq:full_app_Cc_rate}
\end{equation}
\end{lemma}

\begin{proof}
For $h\in\clh_1$, the population-centering identity
\eqref{eq:full_population_centering_identity} and the tower property of conditional expectation give
\begin{equation}
\widetilde L_{21}^{c} h
=
\EE_\pi\!\left[
h(X_0)\EE_\pi(k_1^c\mid X_0)
\right]
=0, \qquad h \in \clh_1.
\end{equation}
Define the random rank-one Hilbert--Schmidt operator $U_t : \clh_1 \to \widetilde{\clh}_2 $ by
\begin{equation}
U_t
\doteq
k_t^c\otimes K_1(X_{t-1},\cdot),\quad
\text{ so that }\quad
    U_t h = \langle h, K_1(X_{t-1}, \cdot)\rangle_{\clh_1} k_t^c = h(X_{t-1})k_t^c, \qquad h \in \clh_1,
\end{equation}
from which we see that $\widehat{\widetilde L}_{21,T}^{c}h = \frac{1}{T} \sum_{t=1}^{T}  U_t h$.
Boundedness of $K_1$ and $\widetilde{K}_2$ ensure that 
$\|U_t\|_{\cls_2}\leq2\kappa\kappa_2$. Applying  similar arguments to the proof of 
Lemma~\ref{lem:full_fluctuation_rates} and noting that $\|\cdot \|_{\op}$ is dominated by $\|\cdot \|_{\cls_2}$, we see that, for some $C_4 \in (0,\infty)$, 
\begin{equation}
\EE_\pi\left\|
\frac1T\sum_{t=1}^T(U_t-\EE_\pi U_1)
\right\|_{\cls_2}^2
\leq
\frac{C_4}{T}.
\end{equation}
The identity in 
\eqref{eq:full_app_Cc_rate} then follows from Markov's inequality.
\end{proof}

\begin{lemma}
\label{lem:B-control}
Suppose Assumptions~\ref{assump:bounded_dynamics}-\ref{assump:gaussian} and the embedding bound \eqref{eq:full_embedding_bound}. Recall $B_{\lambda,T}$ 
from \eqref{eq:full_app_B_def}.
There exists a constant \(C \in (0,\infty) \), independent of \(T\) and
\(\lambda\), such that
\begin{equation}
\mathbb{P}\!\left(
    \|B_{\lambda,T}\|_{\op}>\frac12
\right)
\leq
\frac{C}{T\lambda^{2\gamma_\psi}}.
\label{eq:norm_B_lambda_T_bad_event_prob}
\end{equation}
Consequently, if \(\lambda=\lambda_T\) satisfies
\(T\lambda_T^{2\gamma_\psi}\to\infty\), then
\begin{equation}
\mathbb{P}\!\left(
    \|B_{\lambda_T,T}\|_{\op}\leq\frac12
\right)\longrightarrow 1.
\label{eq:norm_B_lambda_T_bad_event_prob_to_0}
\end{equation}
\end{lemma}

\begin{proof}
Recall from \eqref{eq:full_app_B_D0_identity} in Lemma \ref{lem:full_exact_factorization} that \(\|B_{\lambda,T}\|_{\op}=D_{0,\lambda,T}\). Then, recalling  \eqref{eq:expected_norm_bound_op_D_0_lambda_T} from the proof of Lemma \ref{lem:full_fluctuation_rates} and using that 
\(\|\cdot\|_{\op}\leq\|\cdot\|_{\mathrm{S}_2}\), we see that,  for some $C_1 \in (0,\infty)$,  
\begin{equation}
\mathbb{E}\|B_{\lambda,T}\|_{\op}^{2}
\leq
\frac{C_0 H_\lambda^{2}}{T}
\leq
\frac{C_1}{T\lambda^{2\gamma_\psi}},
\end{equation}
where, as in that proof,  \(H_\lambda=C_\psi^2\lambda^{-\gamma_\psi}\). The inequality in \eqref{eq:norm_B_lambda_T_bad_event_prob} then follows from Markov's inequality and the conclusion in \eqref{eq:norm_B_lambda_T_bad_event_prob_to_0}  follows from our assumption that $T \lambda_T^{2\gamma_{\psi} }\to \infty$. 
\end{proof}

\section{Auxiliary Results}
\label{app:auxiliary_results}

This section collects auxiliary results used throughout the paper. These results are organized according to their roles.

\begin{lemma}
\label{lem:cond_center_consequences}
Suppose Assumptions~\ref{ass:cond_center} and \ref{assump:bounded_kernel}. Then every $h\in\clh_2$
satisfies
\begin{equation}
\EE_\pi[h(X_0,Y_0)\mid X_0]=0
\qquad\text{a.s.}
\label{eq:lemmaF1-eq1}
\end{equation}
Consequently, for $f \in \clh_1$ and $g \in \clh_2$, 
\begin{equation}
L_{21}^c f=0,
\qquad
L_{12}g=0.
\label{eq:lemmaF1-eq2}
\end{equation}
\end{lemma}

\begin{proof}
We first prove \eqref{eq:lemmaF1-eq1}. The argument is to establish the
conditional mean-zero property on the finite span of kernel sections and
then extend it to all of $\clh_2$ by continuity.

Let
\begin{equation}
h_n
=
\sum_{j=1}^{m_n} a_{j,n}
K_2^c((x_{j,n},y_{j,n}),\cdot)
\end{equation}
be a finite linear combination of kernel sections. By symmetry of
$K_2^c$ and Assumption~\ref{ass:cond_center},
\begin{align}
\EE_\pi[h_n(X_0,Y_0)\mid X_0]
&=
\sum_{j=1}^{m_n} a_{j,n}
\EE_\pi\!\left[
K_2^c((x_{j,n},y_{j,n}),(X_0,Y_0))
\mid X_0
\right]
\notag\\
&=
\sum_{j=1}^{m_n} a_{j,n}
\EE_\pi\!\left[
K_2^c((X_0,Y_0),(x_{j,n},y_{j,n}))
\mid X_0
\right]
=0
\qquad\text{a.s.}
\label{eq:cond_center_finite_span}
\end{align}

Now let $h\in\clh_2$. Since finite linear combinations of kernel sections
are dense in $\clh_2$, choose $\{h_n\}_{n\in\NN}$ such that
\begin{equation}
\|h_n-h\|_{\clh_2}\to0.
\label{eq:hn-h-null}
\end{equation}
To pass the conditional mean-zero property from $h_n$ to $h$, it suffices to
show that the corresponding conditional expectations converge in
$L^2(\pi_X)$. Since
$\EE_\pi[h_n(X_0,Y_0)\mid X_0]=0$, we have
\begin{align}
\left\|
\EE_\pi[h(X_0,Y_0)\mid X_0=\cdot]
\right\|_{L^2(\pi_X)}
&=
\left\|
\EE_\pi[(h-h_n)(X_0,Y_0)\mid X_0=\cdot]
\right\|_{L^2(\pi_X)}
\notag\\
&\leq
\|h-h_n\|_{L^2(\pi)}
\notag\\
&\leq
\kappa\|h-h_n\|_{\clh_2}
\longrightarrow0.
\label{eq:conditional_center_extension}
\end{align}
The first equality follows from \eqref{eq:cond_center_finite_span}, the first inequality follows from the
$L^2$ contraction property of conditional expectation, and the second
follows from Assumption~\ref{assump:bounded_kernel} and the reproducing
property, since
\begin{equation}
|u(x,y)|
\leq
\|u\|_{\clh_2}
K_2^c((x,y),(x,y))^{1/2}
\leq
\kappa\|u\|_{\clh_2},
\qquad u\in\clh_2.
\end{equation}
Therefore,
\begin{equation}
\EE_\pi[h(X_0,Y_0)\mid X_0]=0
\qquad\text{a.s.},
\end{equation}
which proves \eqref{eq:lemmaF1-eq1}.

We next prove the two operator identities in \eqref{eq:lemmaF1-eq2}.  For $f\in\clh_1$, Assumption
\ref{ass:cond_center} and the tower property give, for every $(x,y)\in\RR^2$,
\begin{align}
        (L_{21}^cf)(x,y)
        &=
        \EE_\pi\left[
        K_2^c((X_0,Y_0),(x,y))f(X_0)
        \right]
        =
        \EE_\pi\left[
        f(X_0)
        \EE_\pi\left[
        K_2^c((X_0,Y_0),(x,y))
        \mid X_0
        \right]
        \right]
        =0 .
        \label{eq:L21c_zero_from_cond_center}
\end{align}
Thus $L_{21}^cf=0$.  Finally, for $g\in\clh_2$, the conditional
mean-zero property proved above implies that, for every $x\in\RR$,
\begin{align}
        (L_{12}g)(x)
        &=
        \EE_\pi\left[
        K_1(X_0,x)g(X_0,Y_0)
        \right]
        =
        \EE_\pi\left[
        K_1(X_0,x)
        \EE_\pi[g(X_0,Y_0)\mid X_0]
        \right]
        =0 .
        \label{eq:L12_zero_from_cond_center}
\end{align}
This completes the proof.
\end{proof}

\begin{lemma}\label{lem:gc_equivalence}
The condition $(I_{L^2(\pi)} - R) g_X = 0$ in $L^2(\pi)$ is equivalent to nonlinear Granger non-causality in the stationary $L^2(\pi)$ sense, namely that $g_X(x,y) = f_X(x)$
for $\pi$-a.e. $(x,y)\in\RR^2$ and some $f_X\in L^2(\pi_X)$.
\end{lemma}

\begin{proof}
Suppose first that $g_X(x,y) = f_X(x)$
for $\pi$-a.e. $(x,y)$ and some $f_X\in L^2(\pi_X)$. Then,
\begin{equation}
R g_X(x,y)
=
\EE\left[g_X(X,Y) \mid X = x\right]
=
\EE\left[f_X(X) \mid X = x\right]
=
f_X(x), \qquad \pi\text{-a.e. }(x,y) \in \RR^2.
\end{equation}
Hence $(I_{L^2(\pi)}-R)g_X = 0$ in $L^2(\pi)$. Conversely, suppose
$(I_{L^2(\pi)}-R)g_X = 0$. Then $g_X = R g_X$ in $L^2(\pi)$, and by
definition of $R$,
\begin{equation}
R g_X(x,y)
=
\EE\left[g_X(X,Y) \mid X = x\right],  \qquad \pi\text{-a.e. }(x,y) \in \RR^2.
\end{equation}
Define
\begin{equation}
f_X(x) \doteq \EE\left[g_X(X,Y) \mid X = x\right].
\end{equation}
Then $f_X \in L^2(\pi_X)$ and $g_X(x,y) = f_X(x)$ for $\pi$-a.e. $(x,y)$, which shows that $g_X$ depends only on $x$ in the stationary $L^2(\pi)$ sense.
\end{proof}

\begin{lemma}\label{lem:commute}
Let $\clh_1$ and $\clh_2$ be Hilbert spaces, and let $\mathcal A : \clh_1 \to \clh_2$ be a bounded linear operator with adjoint $\mathcal A^* : \clh_2 \to \clh_1$. Then, for any $\lambda > 0$,
\begin{equation}\label{eq:lem:commute:eq1}
\left(\mathcal A \mathcal A^* + \lambda I_{\clh_2}\right)^{-1}\mathcal A
=
\mathcal A \left(\mathcal A^* \mathcal A + \lambda I_{\clh_1}\right)^{-1}.
\end{equation}
\end{lemma}

\begin{proof}
Observe that
\begin{equation}\label{eq:lem:commute:eq2}
\mathcal A \left(\mathcal A^* \mathcal A + \lambda I_{\clh_1}\right)
=
\mathcal A \mathcal A^* \mathcal A + \lambda \mathcal A
=
\left(\mathcal A \mathcal A^* + \lambda I_{\clh_2}\right)\mathcal A.
\end{equation}
Since $\mathcal A^* \mathcal A$ and $\mathcal A \mathcal A^*$ are self-adjoint and positive semidefinite, the operators
\begin{equation}\label{eq:lem:commute:eq3}
\mathcal A^* \mathcal A + \lambda I_{\clh_1}
\quad \text{and} \quad
\mathcal A \mathcal A^* + \lambda I_{\clh_2}
\end{equation}
are invertible for every $\lambda > 0$. Left-multiplying \eqref{eq:lem:commute:eq2} by $\left(\mathcal A \mathcal A^* + \lambda I_{\clh_2}\right)^{-1}$ yields
\begin{equation}\label{eq:lem:commute:eq4}
\left(\mathcal A \mathcal A^* + \lambda I_{\clh_2}\right)^{-1}
\mathcal A \left(\mathcal A^* \mathcal A + \lambda I_{\clh_1}\right)
=
\mathcal A,
\end{equation}
and right-multiplying \eqref{eq:lem:commute:eq4} by $\left(\mathcal A^* \mathcal A + \lambda I_{\clh_1}\right)^{-1}$ gives the result.
\end{proof}

\begin{lemma}\label{lem:injectivitycovarianceoperator}
    For a probability measure $\mu$ on $\RR^d$, let $\clh \subseteq L^2(\mu)$ be the RKHS induced by a kernel $K : \RR^d \times \RR^d \to \RR$. Consider the operator $L : \clh \to \clh$ given by 
    \begin{equation}
    L h \doteq \EE_{\mu} [ h(U) K( U, \cdot) ], \qquad h \in \clh. 
    \end{equation}
Then, for every $h\in\mathcal H$,
$Lh=0$ in $\mathcal H$ if and only if 
$h=0$ in $L^2(\mu)$.
If, in addition, $\operatorname{supp}(\mu)=\mathbb R^d$
and $\mathcal H\subseteq C(\mathbb R^d:\mathbb R)$,
then $L$ is injective on $\mathcal H$.
\end{lemma}
\begin{proof}
For $h\in\mathcal H$ and $U\sim\mu$, the reproducing
property gives
\begin{align}
\langle h , L h \rangle_{\clh} 
    &= \langle h, \EE_{\mu} [ h(U) K( U, \cdot) ] \rangle_{\clh}
    = \EE_{\mu} [  \langle h , h(U) K(U, \cdot) \rangle_{\clh}]\\
    &= \EE_{\mu} [ h(U) \langle h, K(U, \cdot) \rangle_{\clh}] 
    = \EE_{\mu} [ (h(U))^2]
=
\|h\|_{L^2(\mu)}^2.
\end{align}
Hence $Lh=0$ implies $h=0$ $\mu$-almost everywhere.
Conversely, if $h=0$ $\mu$-almost everywhere, then
$h(U)K(U,\cdot)=0$ almost surely, so $Lh=0$.

If $\operatorname{supp}(\mu)=\mathbb R^d$ and $h$ is
continuous, then $h=0$ $\mu$-almost everywhere implies
$h=0$ everywhere. Indeed, a nonzero value of $h$ would,
by continuity, give an open neighborhood on which $|h|$
is bounded away from zero, contradicting full support.
Thus, under these additional conditions, $L$ is
injective on $\mathcal H$.
\end{proof}

\section{Proofs of results in Section \ref{subsec:weaken_conditional_centering}}\label{appendix:proofSec6}

\begin{proof}[Proof of Theorem~\ref{thm:ZT_limit_weak_centering}]
Under the null hypothesis, the decomposition
\eqref{eq:circ_population_decomposition} gives
\begin{equation}
Z_T^\circ
=
G_T^\circ
+
\sqrt T\,L_{21}^\circ(f_X-\widehat f)
+
\sqrt T\,(\widehat L_{21}^\circ-L_{21}^\circ)
(f_X-\widehat f).
\label{eq:proof_circ_main_decomposition}
\end{equation}
By~\eqref{eq:circ_L21_operator_rate},
\begin{align}
\left\|
\sqrt T\,(\widehat L_{21}^\circ-L_{21}^\circ)
(f_X-\widehat f)
\right\|_{\clh_2^\circ}
\le
\sqrt T\,
\|\widehat L_{21}^\circ-L_{21}^\circ\|_{\op}
\|f_X-\widehat f\|_{\clh_1}
=
O_\PP(1)o_\PP(1)
=
o_\PP(1),
\label{eq:proof_circ_empirical_remainder}
\end{align}
where, under the null, $\widehat f=\widetilde f$ by
\eqref{eq:tildeffhat}, and hence
$\|f_X-\widehat f\|_{\clh_1}=o_\PP(1)$ by
Assumption~\ref{assump:first_stage_consistency}.

The assumed joint convergence in~\eqref{eq:first_stage_functional_clt},
combined with~\eqref{eq:proof_circ_main_decomposition}--
\eqref{eq:proof_circ_empirical_remainder} and Slutsky's theorem, gives
\[
Z_T^\circ
\xrightarrow{\distr}
\mathcal G_2^\circ+\mathcal R^\circ
\qquad
\text{in }\clh_2^\circ.
\]
The convergence of $Q_T^\circ$ follows from the continuous mapping theorem,
since $z\mapsto\|z\|_{\clh_2^\circ}^2$ is continuous.
\end{proof}

\begin{proof}[Proof of Proposition~\ref{prop:autonomous_null}]
Let $h\in L^2(\pi_X)$ be such that the expectation defining
$L_{21}^\circ h$ is well defined. Since
$\pi=\pi_X\otimes\pi_Y$ and the kernel has the product form
\eqref{eq:circ_product_kernel}, for every $(x',y')\in\RR^2$,
\begin{align}
(L_{21}^\circ h)(x',y')
&=
\EE_\pi\!\left[
\bar K_X(X,x')\bar K_Y(Y,y')h(X)
\right]
\notag\\
&=
\EE_{\pi_X}\!\left[
\bar K_X(X,x')h(X)
\right]
\EE_{\pi_Y}\!\left[
\bar K_Y(Y,y')
\right].
\label{eq:last-1}
\end{align}
The second factor is zero by~\eqref{eq:marginal_y_centering}. Hence
$L_{21}^\circ h=0$ for every such $h$.

Under the additional regularity conditions in the proposition,
\eqref{eq:circ_GT_limit} gives
$G_T^\circ\xrightarrow{\distr}\mathcal G_2^\circ$ in
$\clh_2^\circ$. Moreover,
$L_{21}^\circ(f_X-\widehat f)=0$ identically, so
\eqref{eq:first_stage_functional_clt} holds with
$\mathcal R^\circ=0$. The limiting statements therefore follow from
Theorem~\ref{thm:ZT_limit_weak_centering}.
\end{proof}

\section{Additional Sensitivity Analysis}
\label{sec:appendix_bandwidth}

This appendix provides an additional sensitivity analysis of the 
empirical-centering procedure with respect to the kernel bandwidth and
regularization parameters. Throughout, the sample size is fixed at $T=100$, and rejection rates are computed from $500$ Monte Carlo repetitions for each configuration. The regularization parameters are varied over the same grids as in Section~\ref{sec:sensitivity}.
Three bandwidth specifications are considered:
\begin{equation}
\tau=\sqrt{2},
\qquad
\tau=\sqrt{2}\ \text{after standardizing the data},
\qquad
\tau=\text{median pairwise distance}.
\end{equation}
For the median-distance choice, the bandwidths for $K_1(x,x')$ and $\widetilde K_2((x,y),(x',y'))$ are computed separately as the median pairwise distances of the corresponding univariate and bivariate observations. 

In addition to evaluating finite-sample performance, the simulations provide insight into the distinct roles of the bandwidth and regularization parameters, leading to practical guidelines for their selection; see Section~\ref{sec:tuning_discussion} for further discussion.

\subsection{Null Model}

For the null model, data are generated from the bivariate nonlinear VAR model with
\begin{equation}
g_X(x,y)
=
a\sin(x),
\qquad
g_Y(x,y)
=
0.4\cos(x)+0.4\tanh(y),
\end{equation}
where $a\in\{0.5,1,1.5\}$.
The innovations are generated independently from normal distributions with standard deviation $\sigma\in\{0.5,1,1.5\}$.
By varying both $a$ and $\sigma$, we investigate the sensitivity of the procedure to changes in the strength of the self-dependence and the noise level. Figure~\ref{fig:null_bandwidth} reports the empirical rejection rates under the null model for the three bandwidth specifications. Under the null hypothesis, the rejection rates should be close to the nominal significance level of $0.05$.

\begin{figure}[p]
\centering

\IfFileExists{Figure_3a.png}{\includegraphics[height=0.31\textheight]{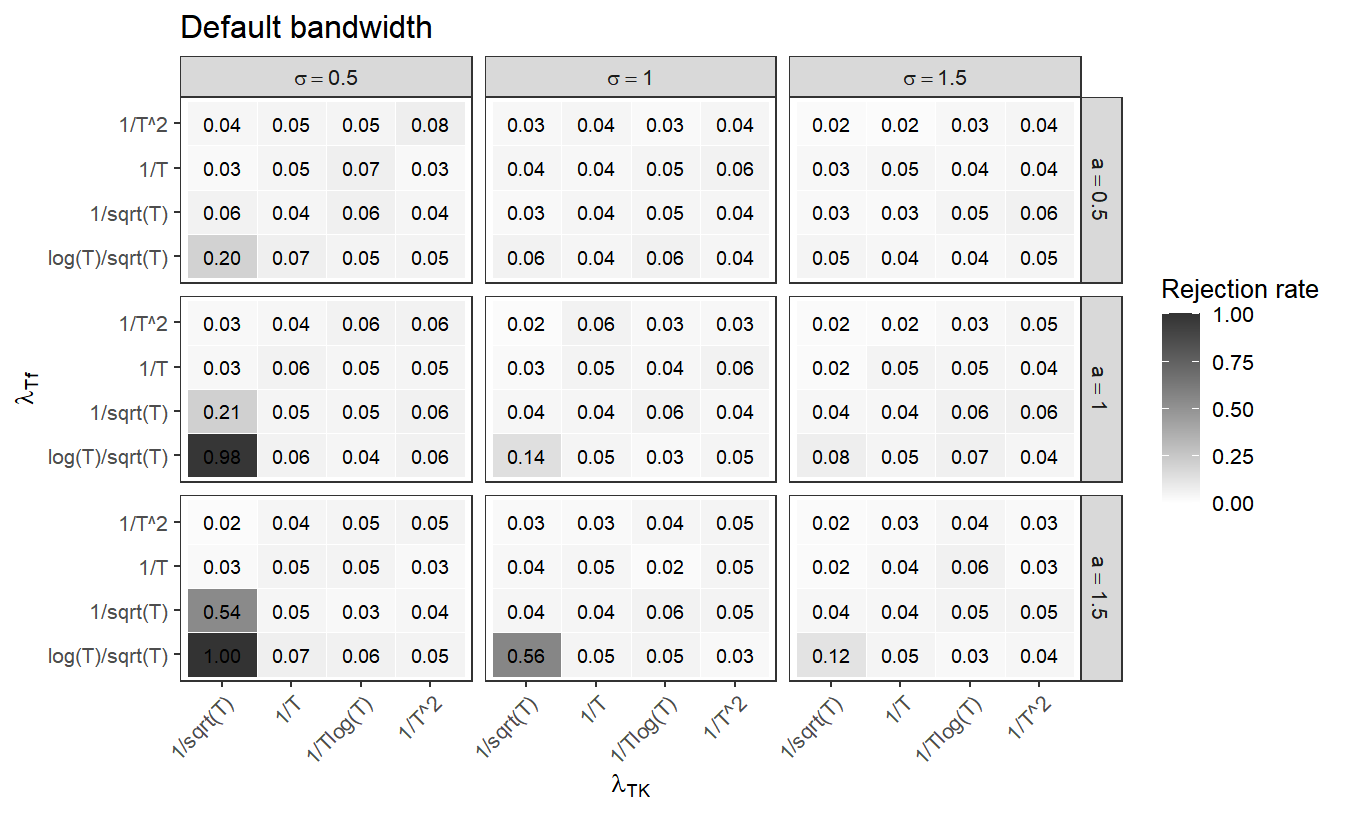}}{\fbox{\parbox[c][0.31\textheight][c]{0.29\textwidth}{\centering Figure omitted.}}}

\IfFileExists{Figure_3b.png}{\includegraphics[height=0.31\textheight]{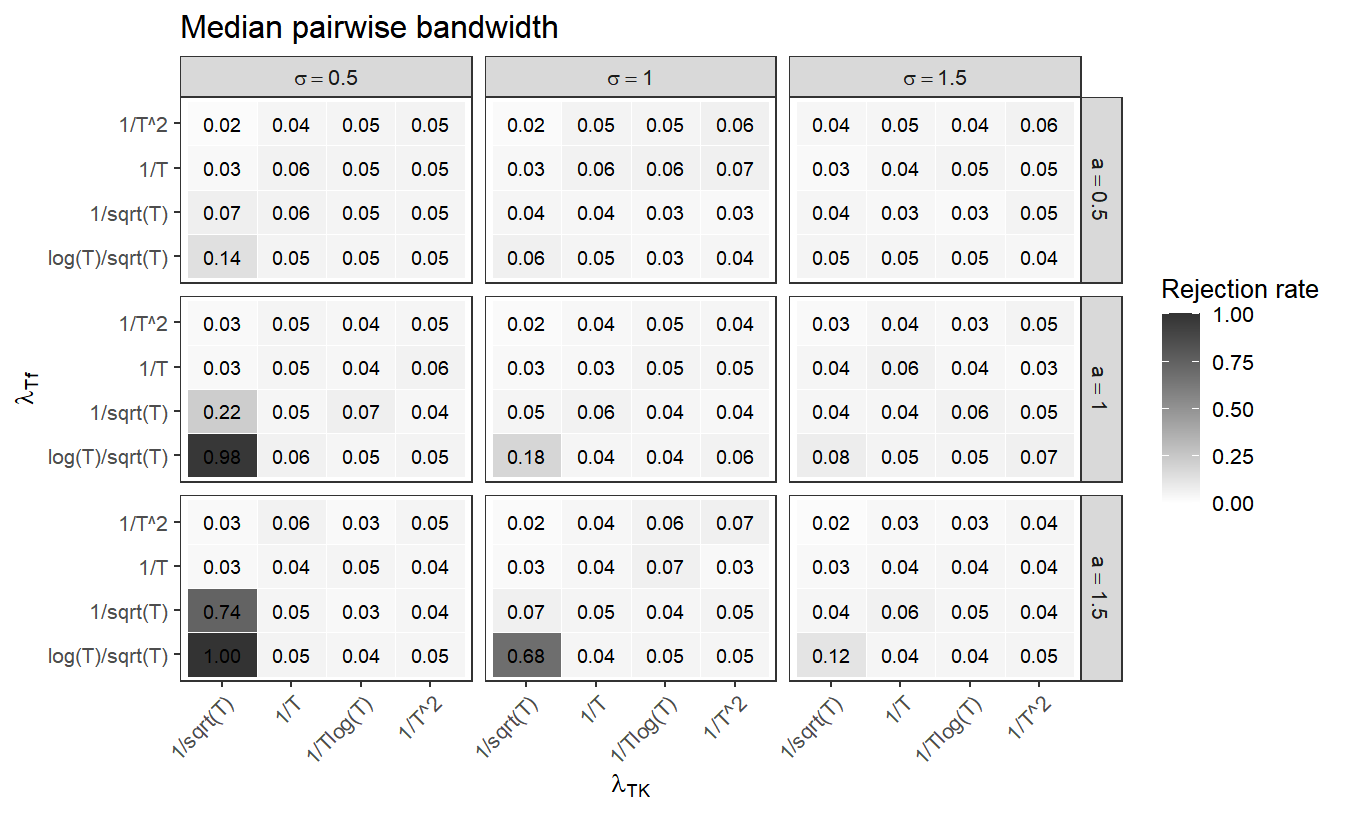}}{\fbox{\parbox[c][0.31\textheight][c]{0.29\textwidth}{\centering Figure omitted.}}}

\IfFileExists{Figure_3c.png}{\includegraphics[height=0.31\textheight]{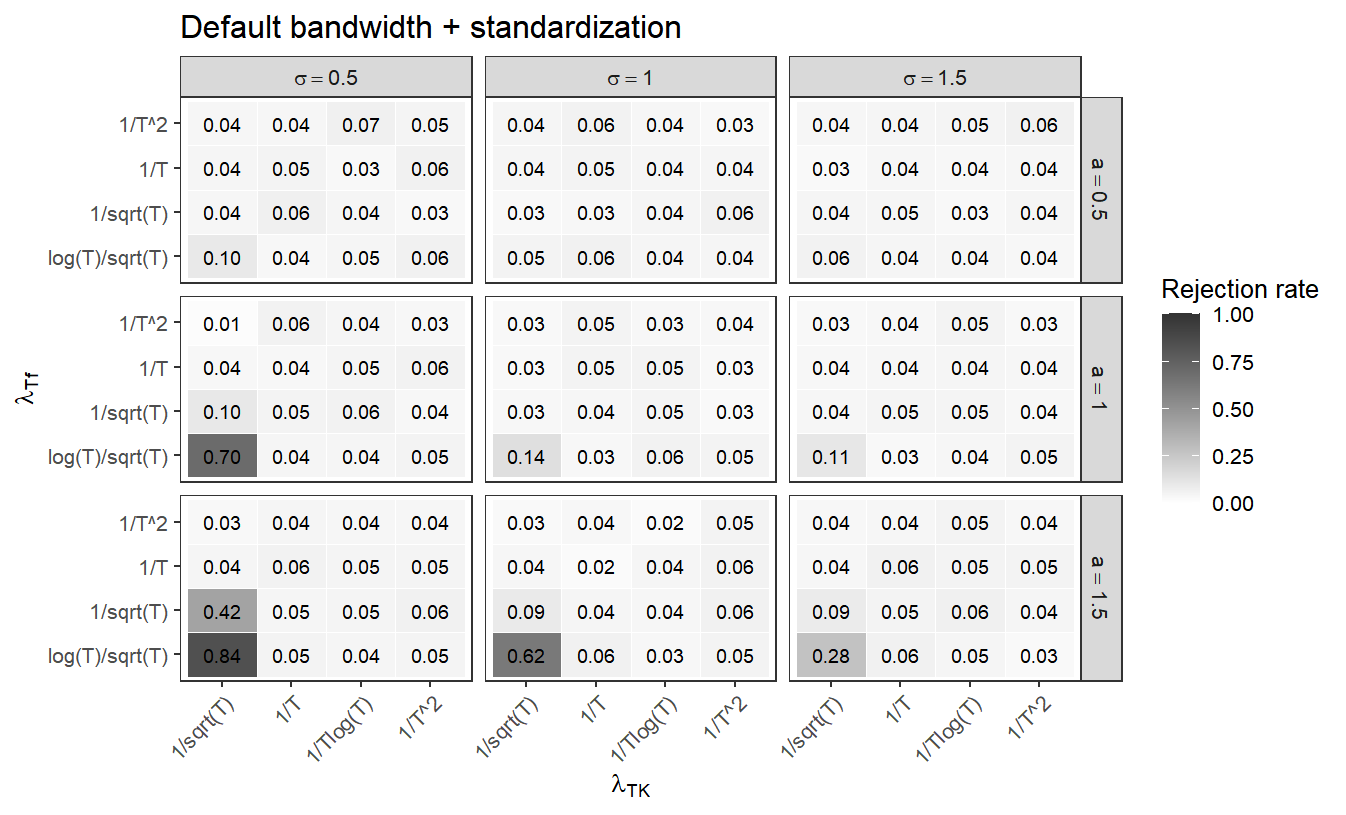}}{\fbox{\parbox[c][0.31\textheight][c]{0.29\textwidth}{\centering Figure omitted.}}}

\caption{Rejection rates under the null model for the three bandwidth specifications.}
\label{fig:null_bandwidth}
\end{figure}

\subsection{Alternative Model}

For the alternative model, data are generated from the bivariate nonlinear VAR model with
\begin{align}
g_X(x,y)
&=
0.5\sin(x)
+
a\tanh(2y)\frac{1+\tanh(5x)}{2},
\\
g_Y(x,y)
&=
0.4\cos(x)+0.4\tanh(y),
\end{align}
where $a\in\{0.25,0.5,0.75,1\}$.
The innovations are generated independently from normal distributions with standard deviation $\sigma\in\{1,1.5,2\}$.
By varying both $a$ and $\sigma$, we investigate the sensitivity of the procedure to changes in the strength of the Granger-causal signal and the noise level. Figure~\ref{fig:alt_bandwidth} reports the empirical rejection rates under the alternative model for the three bandwidth specifications. A powerful procedure should yield rejection rates close to $1$.

\begin{figure}[p]
\centering

\IfFileExists{Figure_4a.png}{\includegraphics[height=0.31\textheight]{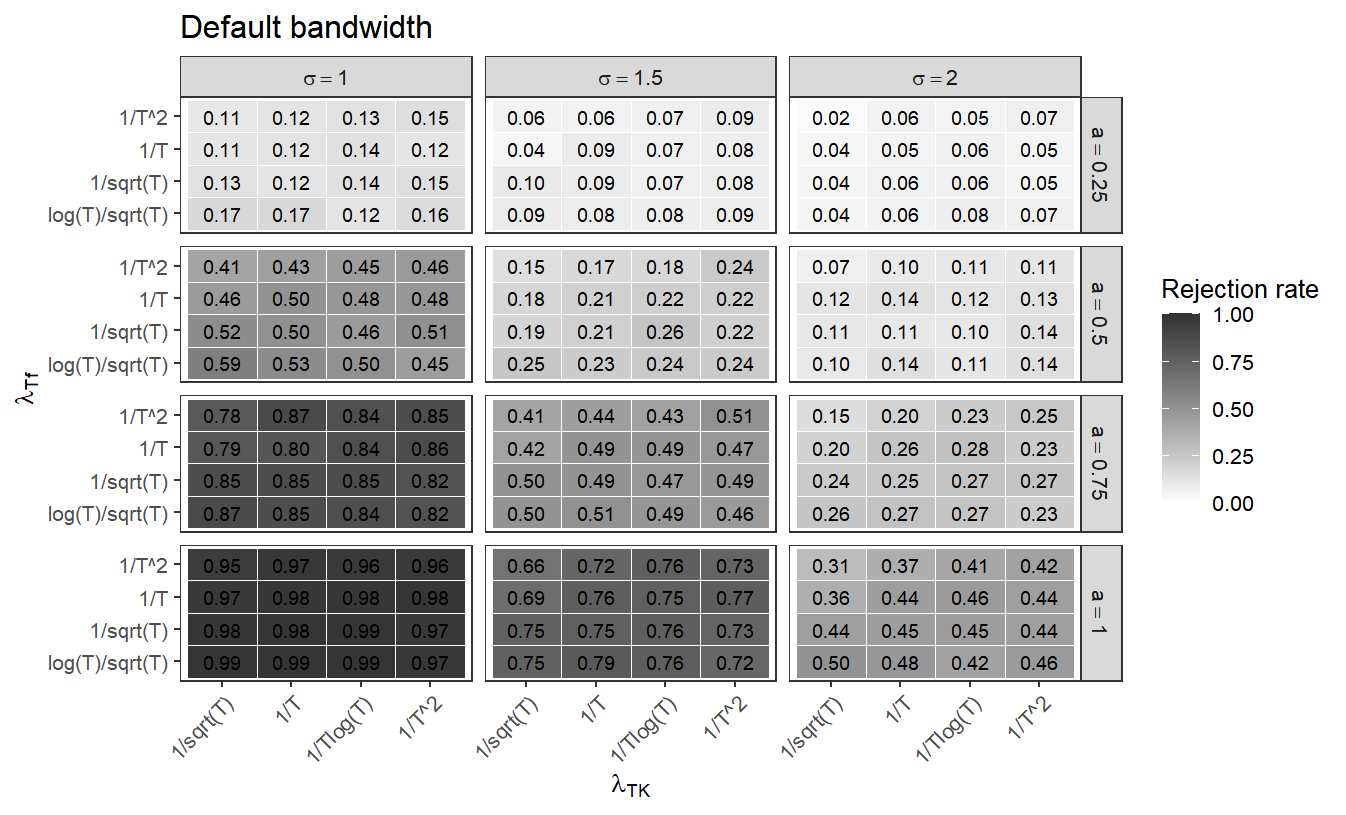}}{\fbox{\parbox[c][0.31\textheight][c]{0.29\textwidth}{\centering Figure omitted.}}}

\IfFileExists{Figure_4b.png}{\includegraphics[height=0.31\textheight]{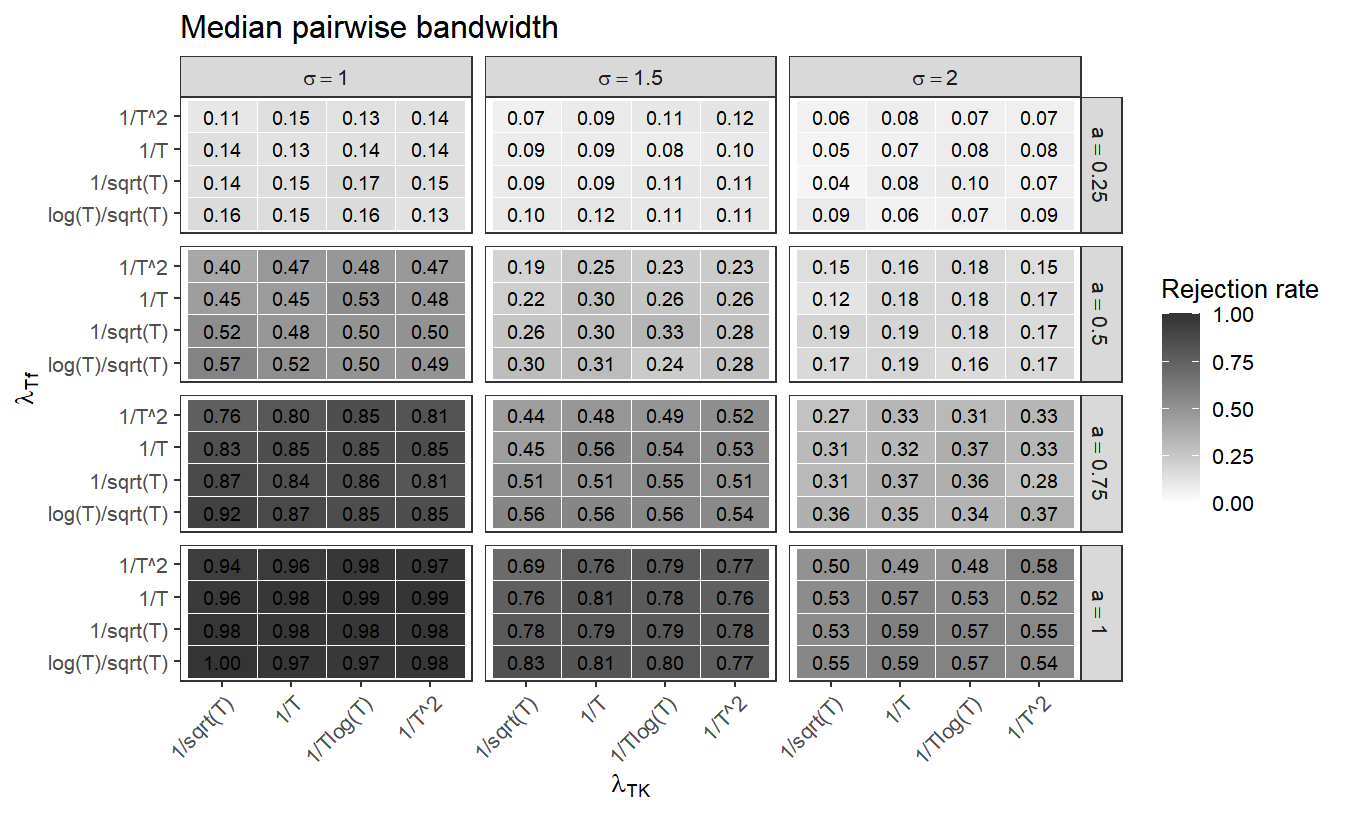}}{\fbox{\parbox[c][0.31\textheight][c]{0.29\textwidth}{\centering Figure omitted.}}}

\IfFileExists{Figure_4c.png}{\includegraphics[height=0.31\textheight]{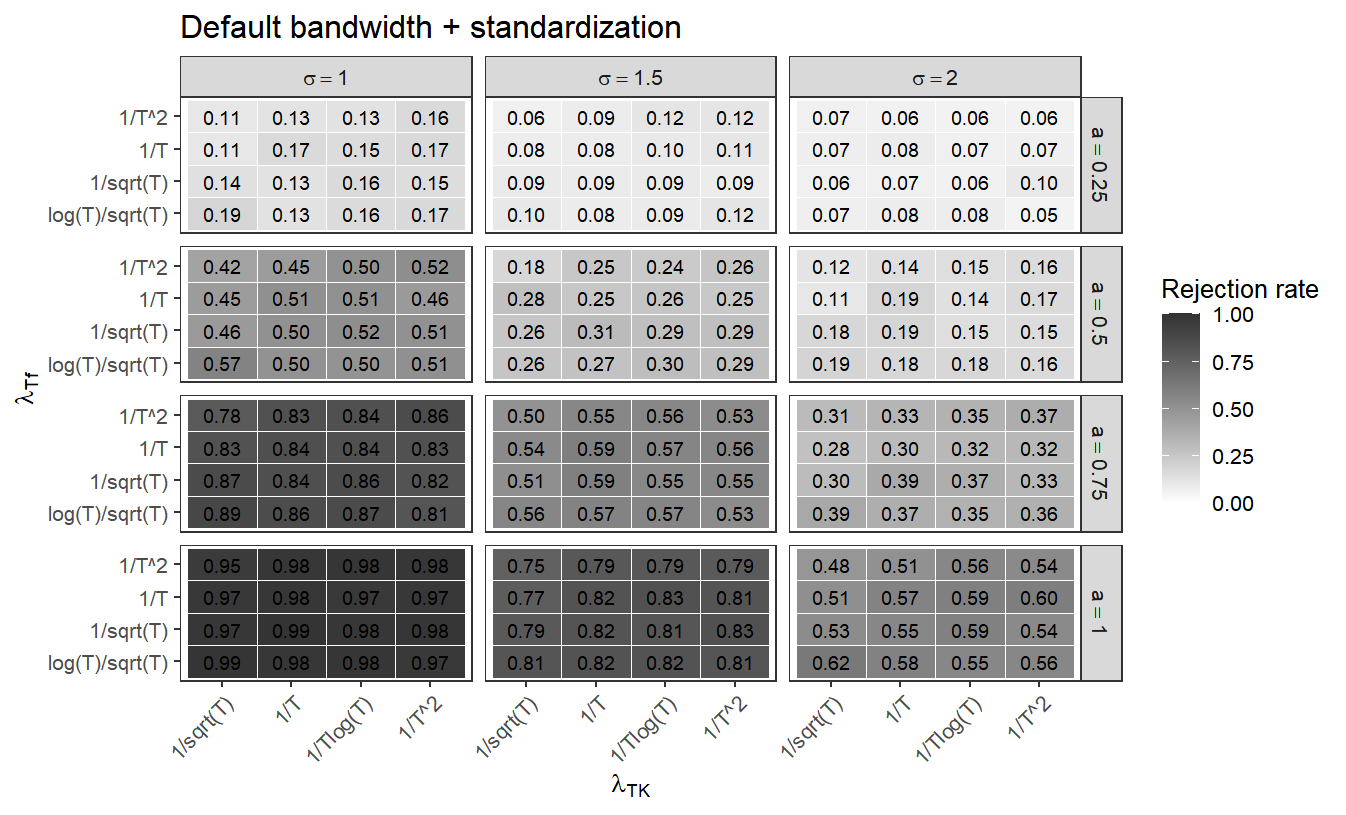}}{\fbox{\parbox[c][0.31\textheight][c]{0.29\textwidth}{\centering Figure omitted.}}}

\caption{
Rejection rates under the alternative model for the three bandwidth specifications.
}
\label{fig:alt_bandwidth}
\end{figure}

\subsection{Discussion of Tuning Parameter Choices}
\label{sec:tuning_discussion}

The simulations in Sections~\ref{sec:sensitivity} and
\ref{sec:appendix_bandwidth} provide additional insight into the roles of
the bandwidth parameter and the regularization parameters
$\lambda_{T,f}$ and $\lambda_{T,K}$.

For Gaussian kernels, the bandwidth determines the effective geometry of
the induced reproducing kernel Hilbert space and therefore influences how
well the $Y$-dependent component $g_{X,0}$ can be represented in finite
samples. As discussed after Theorem~\ref{thm:ZT_alt}, the asymptotic power
of the proposed procedure depends on the image of $g_{X,0}$ under the
operator $L_{22}^c$. Consequently, bandwidth choices that better capture the
structure of the underlying Granger-causal signal can improve finite-sample power.
For example, when the Granger-causal effect exhibits relatively rapid local
variation, smaller bandwidths may be preferable. More generally, when
little prior information is available, either the median pairwise distance
rule or the default choice $\tau=\sqrt{2}$ applied after
standardization provides a practical and robust choice.

The appendix simulations indicate that the bandwidth choice primarily
affects power rather than size control. The differences are most visible
for moderate signal-to-noise settings, whereas for very weak or very strong
signals the three bandwidth specifications yield similar performance.
Moreover, the proposed framework is not restricted to Gaussian kernels.
When prior information regarding the form of the nonlinear relationship is
available, alternative kernels that better represent the underlying signal
may further improve power.

The regularization parameter $\lambda_{T,f}$ controls the shrinkage of the
first-stage kernel ridge regression estimator
$\widehat f$ in~\eqref{eq:fhat_repr}. Consequently, it governs the tradeoff
between leakage of the own-history component $f_X$ and suppression of the
$Y$-dependent component $g_{X,0}$. Larger values of $\lambda_{T,f}$ produce
stronger shrinkage, allowing more of $f_X$ to remain in the residuals while
simultaneously preserving more of the signal contained in $g_{X,0}$. This
behavior is reflected in both the null and alternative simulations, where
larger values of $\lambda_{T,f}$ tend to increase power but may also lead to
size inflation.

The default finite-sample choice
\begin{equation}
\lambda_{T,f}
=
\frac{\log(T)}{\sqrt{T}}
\end{equation}
should be interpreted as an implementation rule for the 
procedure, rather than as the only asymptotically admissible first-stage
regularization. The simulations suggest that this choice is relatively
aggressive in finite samples, yielding increased power at the expense of
some size distortion.

The second regularization parameter $\lambda_{T,K}$ arises in the kernel
centering procedure through \eqref{eq:SK_residualizer_def} and
\eqref{eq:K2hatc_residualized_def}. Smaller values of $\lambda_{T,K}$ reduce
the regularization imposed on the empirical centering fit and allow a larger
portion of the feature variation explained by $X$ to be removed. The
simulations suggest that $\lambda_{T,K}$ mainly affects size control, whereas
$\lambda_{T,f}$ governs the primary power-size tradeoff. In particular, in
the settings considered here, choosing $\lambda_{T,K}\leq 1/T$ mitigates much
of the size inflation associated with aggressive choices of
$\lambda_{T,f}$ while preserving most of the corresponding power gain. This
highlights a practical merit of the two-stage construction, as relatively
aggressive first-stage regularization can be used to increase power, while
the subsequent empirical-centering step controls the accompanying size
inflation without removing most of that gain.

Taken together, these simulations support the default choices
\begin{equation}
\lambda_{T,f}
=
\frac{\log(T)}{\sqrt{T}},
\qquad
\lambda_{T,K}
=
\frac{1}{T\log(T)},
\end{equation}
as finite-sample tuning rules for the procedure, providing a
favorable empirical balance between power and size control. These choices are
not covered by the asymptotic regimes in
Assumption~\ref{ass:full_empirical_centering_tuning}; admissible logarithmic
and polynomial sequences are given in
Proposition~\ref{prop:full_tuning_choices}.

For practitioners seeking a simpler finite-sample tuning rule, the simulation
results also suggest using $\lambda_{T,f}=\lambda_{T,K}=1/T$
for the procedure. This recommendation is empirical and is not an
asymptotic claim under Proposition~\ref{prop:full_tuning_choices}.

\section*{\normalsize{Funding}}

This work was supported by the Deutsche Forschungsgemeinschaft (DFG, German Research Foundation) grant number 585494951.

\section*{\normalsize{Declaration of generative AI and AI-assisted technologies}}

During the preparation of this work, the authors used ChatGPT (OpenAI) to assist with drafting and revising text, standardizing notation, and reviewing and clarifying mathematical arguments and proofs. After using this tool, the authors reviewed and edited the content as needed and take full responsibility for the content of the article.

\bibliographystyle{plainnat}
\bibliography{bib}

\end{document}